\documentclass[preprint]{imsart}

\usepackage{amsmath,amsthm,graphicx,subfig,float,pdflscape,hyperref}
\usepackage[figuresleft]{rotating}

\doi{10.1214/154957804100000000}
\pubyear{0000}
\volume{0}
\firstpage{0}
\lastpage{0}

\startlocaldefs

\newtheorem{proposition}{Proposition}
\newtheorem{theorem}{Theorem}
\theoremstyle{definition}
\newtheorem*{remark}{Remark}

\endlocaldefs

\begin{document}

\begin{frontmatter}

\title{Sequential Quantiles via Hermite Series Density Estimation}
\runtitle{Sequential Quantiles via Hermite Series Density Estimation}

\begin{aug}
  \author{\fnms{Michael} \snm{Stephanou}\corref{}\ead[label=e1]{michael.stephanou@gmail.com}}
  \address{Department of Statistical Sciences, University of Cape Town, Cape Town, South Africa\\ \printead{e1}}
  \author{\fnms{Melvin} \snm{Varughese}\ead[label=e2]{melvin.varughese@gmail.com}}
  \address{Department of Statistical Sciences, University of Cape Town, Cape Town, South Africa,\\ IBM Research, Johanneburg, South Africa\\ \printead{e2}} 
  \and
  \author{\fnms{Iain} \snm{Macdonald}\ead[label=e3]{iain.macdonald@uct.ac.za}}
  \address{Department of Actuarial science, University of Cape Town, Cape Town, South Africa,\\ 
           \printead{e3}}

\runauthor{M. Stephanou et al.}

\end{aug}

\begin{abstract}
Sequential quantile estimation refers to incorporating observations into quantile estimates in an incremental fashion thus furnishing an online estimate of one or more quantiles at any given point in time. Sequential quantile estimation is also known as online quantile estimation. This area is relevant to the analysis of data streams and to the one-pass analysis of massive data sets. Applications include network traffic and latency analysis, real time fraud detection and high frequency trading. We introduce new techniques for online quantile estimation based on Hermite series estimators in the settings of static quantile estimation and dynamic quantile estimation. In the static quantile estimation setting we apply the existing Gauss-Hermite expansion in a novel manner. In particular, we exploit the fact that Gauss-Hermite coefficients can be updated in a sequential manner. To treat dynamic quantile estimation we introduce a novel expansion with an exponentially weighted estimator for the Gauss-Hermite coefficients which we term the Exponentially Weighted Gauss-Hermite (EWGH) expansion. These algorithms go beyond existing sequential quantile estimation algorithms in that they allow arbitrary quantiles (as opposed to pre-specified quantiles) to be estimated at any point in time. In doing so we provide a solution to online distribution function and online quantile function estimation on data streams. In particular we derive an analytical expression for the CDF and prove consistency results for the CDF under certain conditions. In addition we analyse the associated quantile estimator. Simulation studies and tests on real data reveal the Gauss-Hermite based algorithms to be competitive with a leading existing algorithm.
\end{abstract}





\end{frontmatter}

\section{Introduction}

Algorithms for elucidating the statistical properties of streams of data in real time and for the efficient one-pass analysis of massive data sets are becoming increasingly pertinent. These data are being generated by a number of sources including the global financial markets, internet applications, sensors embedded in various devices and data-intensive scientific research endeavours such as the Large Hadron Collider and the Square Kilometre Array (see \cite{bigdata} for a survey of the field of Big Data). Ideally such algorithms should be able to process observations sequentially, without requiring the storage of all observations. In addition, the time taken to process each observation should not grow with the number of previous observations. Certain statistical properties naturally lend themselves to efficient, sequential calculation such as the mean and standard deviation. Depending on the application, these moments may not be sufficient however. This may be true for skewed data for example. In addition, if the data stream being analysed is prone to outliers then more robust statistics may be required. Quantiles are a natural choice in these settings. Examples of areas in which sequential quantile estimation is relevant include network traffic and latency analysis \cite{network}, real time fraud detection \cite{fraud} and high frequency trading (see \cite{hft} for an introduction to sequential algorithms in high frequency trading). Other conceivable applications of sequential quantile estimation include real time detection of anomalies and flagging of noteworthy observations, real time outlier detection and removal, real time provisioning for future demand and load balancing as well as real time risk estimation.\\

In many applications of interest one needs to determine whether a particular value or observation is greater than or less than a certain quantile. In such cases it is more direct to use the cumulative distribution function. Thus, closely linked to quantile estimation is distribution function estimation. In this article we propose new distribution function and quantile estimators based on Hermite series expansions and study their properties. These results are novel and interesting in their own right. That said, the particular setting we consider is that of online estimation and thus existing non-parametric methods are weighed against our methods in this specific context. In the general context, there are of course a number of well established non-parametric distribution function and quantile estimators. The most obvious of these being the sample cumulative distribution function also known as the empirical distribution function. Let  $\mathbf{x}_{i} \sim f(x)$ be i.i.d. random variables drawn from $f(x)$ with cumulative distribution function $F(x)$.\\

The empirical distribution function (EDF) is defined as:

$$\hat{F}_{1}(x)= \frac{1}{n} \sum_{i=1}^{n}  \mathbf{1} \{\mathbf{x}_{i} \leq x  \}.$$

This estimator is consistent in the sense that it converges point-wise to the true CDF.  In fact, it converges uniformly over $x$ (Glivenko-Cantelli theorem). In addition $\hat{F}_{1}(x)$ has an asymptotically normal distribution with the standard $\sqrt{n}$ rate of convergence. The EDF estimator is unbiased and its mean squared error is $\mbox{MSE}(\hat{F}_{1}(x)) = \frac{F(x)(1-F(x))}{n}$. The EDF estimator has however been shown to be asymptotically inferior (asymptotically deficient) compared to the kernel distribution function estimator with an appropriately chosen kernel type (in the MSE sense, see \cite{reiss1981nonparametric} for a precise definition of asymptotic deficiency). In fact the relative deficiency quickly tends to infinity as the sample size increases \cite{reiss1981nonparametric}.  The kernel distribution function estimator is defined as:

$$\hat{F}_{2}(x)= \frac{1}{n} \sum_{i=1}^{n} \int_{-\infty}^{x} \frac{1}{h_{n}}K\left(\frac{\mathbf{x}_{i} -x'}{h_{n}} \right) dx'.$$

where the kernel function $K(x)$ is a non-negative function that integrates to one and has mean zero and the bandwidth, $h_{n}>0$, is a smoothing parameter. $\hat{F}_{2}(x)$ is  asymptotically normally distributed. The almost sure uniform convergence of $\hat{F}_{2}(x)$ has also been proved (see \cite{reiss1981nonparametric}). \\

Closely related to the EDF, the sample quantile is a popular nonparametric estimator of the corresponding population quantile. Define the values \\ $\mathbf{x}_{(1)}, \dots, \mathbf{x}_{(n)}$ to be a permutation of $\mathbf{x}_{1}, \dots, \mathbf{x}_{n}$ such that $\mathbf{x}_{(1)} \leq \mathbf{x}_{(2)} \leq \dots \leq \mathbf{x}_{(n)}$. Here, $\mathbf{x}_{(i)}$, is known as the $i$th order statistic. The EDF can be written in terms of the order statistics as:

$$\hat{F}_{3}(x)= \frac{1}{n} \sum_{i=1}^{n}  \mathbf{1} \{\mathbf{x}_{(i)} \leq x  \}.$$

The inverse cumulative distribution function or quantile function is defined as:
\begin{equation}
	q(p) = \text{inf} \{x : F(x) \geq p\}.
\end{equation}\label{quantileDef}

The p-th quantile can be estimated from the order statistics. In particular if $p \in \left( \frac{i-1}{n}, \frac{i}{n}\right]$ then $\hat{q}(p)=\mathbf{x}_{(i)}$, the $i$th order statistic. The sample quantile estimator is a function of at most two order statistics and thus may suffer a loss of efficiency for certain distributions. A natural way to improve efficiency is to form a weighted average of several order statistics under an appropriate weight function. Such estimators are called L-estimators. The most popular class of L-estimators uses a density function (kernel) as its weight function, these are known as kernel quantile estimators (see \cite{sheather1990kernel}). The kernel quantile estimator is defined as:

$$\hat{q}(p) = \sum_{i=1}^{n} \int_{\frac{i-1}{n}}^{\frac{i}{n}}  \frac{1}{h_{n}} K\left(\frac{s - p}{h_{n}} \right)  \mathbf{x}_{(i)} ds$$

It has been shown that under suitable conditions on $F(x)$ and $h_{n}$ the kernel quantile estimator is more efficient than the sample quantile estimator in the MSE sense \cite{falk1984relative}. Under comparable assumptions to those to prove joint asymptotic normality of a set of empirical quantiles, joint asymptotic normality of the kernel quantile estimator has also been proved \cite{falk1985asymptotic}. The rate is also provided. In  \cite{sheather1990kernel} the asymptotically optimal bandwidth for the kernel quantile estimator is derived and the corresponding MSE is provided. It is also shown that many different variants of the kernel quantile estimator are asymptotically equivalent in the MSE sense. In addition other L-estimators such as Harell-Davis, Kaigh-Lachenbruch and the Brewer estimators are shown to be asymptotically equivalent to the kernel quantile estimator above with a gaussian kernel and certain smoothing parameters. A more modern approach to estimating quantiles based on the Bernstein-Durrmeyer operator is provided in \cite{pepelyshev2014estimation}. Like other L-estimators, the BD estimator constitutes a weighted version of several order statistics. MSE and MISE consistency results are also provided.\\

In the context of sequential distribution function and quantile estimation, the aforementioned estimators have shortcomings. Both the EDF and kernel distribution function estimator only allow sequential estimation of the cumulative probability at a set of fixed $x$ values (see chapters 4 and 5 of \cite{greblicki2008nonparametric} and chapter 7 of \cite{devroye1985nonparametric} for a discussion of recursive kernel estimators). For quantile estimation, both the sample quantile estimator and L-estimators such as the kernel quantile estimator and the Bernstein-Durrmeyer estimator require the storage and updating of one or more order statistics (a sorted sequence of all observations seen so far). Updating the order statistics cannot in general be done in $O(1)$ time. Moreover, the addition of a new observation would in general change a number of order statistics. Finally, in the context of sequential quantile estimation on streaming data, non-stationarity cannot be naturally addressed since these estimators have no means of \textit{forgetting} previous observations (other than windowing and resetting). \\

Approaches specifically to sequential estimator of quantiles have been developed. Sequential quantile estimation algorithms can be differentiated on whether they seek to maintain an online estimate of a single quantile or multiple quantiles. They can be further differentiated on whether they are meant to estimate static quantiles of a stream of data or dynamic quantiles of a stream of data. In the case of static quantile estimation, online quantile estimates pertain to all the data observed so far and quantiles of the stream being analysed are assumed to be fixed. In the dynamic case, quantile estimates pertain to the current behaviour of the process and it is assumed that quantiles may vary over time. A number of algorithms have been proposed for sequential quantile estimation in these settings. In \cite{P2} the P$^{2}$ algorithm was proposed which utilizes parabolic interpolation in order to estimate a particular quantile. In \cite{extP2First} and \cite{extP2Second} this algorithm was extended to the simultaneous estimation of several quantiles. In \cite{P2ewma} the P$^{2}$ algorithm was further extended to treat dynamic quantile estimation via exponentially weighted quantile estimators. Algorithms have also been proposed based on stochastic approximation \cite{sa1}, \cite{sa2}. This approach was extended to dynamic quantile estimation in \cite{ewsa} by the introduction of Exponentially Weighted Stochastic Approximation (EWSA). These existing methods have a major shortcoming however, online estimates can only be obtained for a particular, pre-selected set of quantiles (e.g. $p=0.5, 0.9, 0.99$ etc.).\\

In this article we propose new techniques based on Hermite series estimators to maintain an online estimate of the CDF and the full quantile function in both the static and dynamic settings and thus yield estimates of the cumulative probability at \textit{arbitrary} $x$ and estimates of \textit{arbitrary} quantiles that can be updated in constant time ($O(1)$ time). This is the primary advantage of our suggested approach. Even if our proposed techniques only have comparable accuracy with existing techniques, they would still be valuable. We demonstrate using simulated and real data, that our techniques may in fact be more accurate.\\

We begin by reviewing some background related to Hermite polynomials and the Gauss-Hermite expansion in section \ref{background}. In section \ref{GHQuantiles} we introduce a Gauss-Hermite based estimator for the cumulative distribution function and discuss a numerical means of obtaining arbitrary quantiles. We then set about applying this to sequential quantile estimation in the settings of static and dynamic quantile estimation. In this article, we make our treatment of static and dynamic quantile estimation concrete by considering the cases of independent identically distributed (i.i.d.) data streams and non-identically distributed independent data streams respectively. Observations are continuous random variables that are revealed sequentially (i.e. one at a time). The basic algorithm for the static case is presented in section \ref{GHOnlineStat}.  We then proceed to treat the dynamic case by introducing an exponentially weighted moving average estimator for the Gauss-Hermite coefficients in section \ref{EWGH}. In section \ref{theoreticalquality} we investigate the quality of the Gauss-Hermite CDF and quantile estimators theoretically. We compare the proposed techniques to a leading existing algorithm for both simulated data (in section \ref{simresults}) and real data (in section \ref{realresults}). Finally, we conclude in section \ref{conclusion}. The practical consideration of standardising the observations from the data stream being analysed is treated in appendix \ref{standardization}. Useful MISE results for the exponentially weighted Gauss-Hermite expansion are derived in appendix \ref{theoryEWGH}.

\section{Background} \label{background}

\subsection{Hermite Polynomials}

In this section we introduce the Hermite polynomials which will play a central role in the construction of orthogonal series estimators for probability density functions. 
The Hermite polynomials are a classical orthogonal polynomial sequence. Following standard notation \cite{szeg1939orthogonal} we define the Hermite polynomials:

$$H_k(x)=(-1)^k e^{x^2}\frac{d^k}{dx^k}e^{-x^2}$$

which are orthogonal under the weight function $e^{-x^{2}}$ i.e. 

\begin{equation}
	\int_{-\infty}^{\infty} e^{-x^{2}} H_k(x) H_l(x) dx= \sqrt{\pi} 2^{k} k! \delta_{kl}. \label{orthogHermite}
\end{equation}

The following explicit expression may also be utilised:

\begin{equation}
	H_k(x) = k! \sum_{m=0}^{\lfloor k/2 \rfloor} \frac{(-1)^m}{k!(k - 2m)!} (2x)^{k - 2m}. \label{hermitedef2}
\end{equation}

Finally, some useful inequalities for $H_k(x)$ are as follows \cite{szeg1939orthogonal} (used in \cite{convergence2} for example):

\begin{equation}\label{hermiteInequal1}
	(2^{k} k! \sqrt{\pi})^{-1/2} \left| H_k(x))\right| e^{\frac{-x^{2}}{2}} \leq c_{a} (k+1)^{-1/4}, \, |x| \leq a
\end{equation}

for some constant $c_{a}$ and non-negative $a$.

\begin{equation}\label{hermiteInequal2}
	(2^{k} k! \sqrt{\pi})^{-1/2} \left| x^{-1/3} H_k(x))\right| e^{\frac{-x^{2}}{2}} \leq d_{a} (k+1)^{-1/4}, \, |x| \geq a
\end{equation}

for some constant $d_{a}$ and positive $a$.\\

In the next section we establish the link between expansions in Hermite polynomials and nonparametric density estimation.

\subsection{Gauss-Hermite Expansion} \label{GHBackground}

A number of probability distribution expansions have been defined in terms of the Hermite polynomials. These include the Gram Charlier A series, the Edgeworth series and the Gauss-Hermite expansion. These expansions have been applied to successfully fit astrophysical data that are nearly Gaussian but with small, meaningful deviations for example \cite{astro}. In this research we focus on one expansion, termed the Gauss-Hermite expansion following the terminology of \cite{astro}. This expansion has good convergence properties in practice and is robust to outliers \cite{gaussHermiteRobust}. We define this expansion below:

If $f(x) \in L_{2}$ i.e. $\int_{-\infty}^{\infty} f^{2}(x) dx < \infty$ then $f(x)$ can be expressed as follows: 

\begin{equation}
	f(x) = \sum_{k=0}^{\infty} a_k H_{k} (x) Z(x), \label{GHDefinition}
\end{equation}

where $Z(x) = \frac{1}{\sqrt{2\pi}} e^{-x^{2}/2}$ is the standard normal probability density function and 

\begin{equation}
	a_k = \alpha_{k} \int_{-\infty}^{\infty} Z(x) f(x) H_{k} (x) dx, \label{GHExactCoeff}
\end{equation}

where $\alpha_{k} = \frac{\sqrt{\pi}}{2^{k-1}k!}$. In what follows, we refer to (\ref{GHDefinition}) along with the coefficients (\ref{GHExactCoeff}) as the Gauss-Hermite expansion. The fact that $f(x)$ can be expanded in this manner is a consequence of the fact that the normalised Hermite functions: 
$$h_{k}=\left(2^{k}k!\sqrt{\pi}\right)^{-\frac{1}{2}} e^{-\frac{x^2}{2}} H_{k}(x),$$ 

are an orthonormal basis for $L_{2}$ (the space of square integrable functions). The Gauss-Hermite expansion is in fact entirely equivalent to the expansion:

\begin{equation}
	f(x) = \sum_{k=0}^{\infty} \tilde{a}_k h_{k}(x),\label{HermiteSeriesEstimator}
\end{equation}

$$ \tilde{a}_k =\sqrt{\frac{1}{\alpha_{k}}}a_k .$$

which is the expansion used to define what is termed the Hermite series estimator in \cite{convergence1}, \cite{convergence2}. We therefore use the descriptions interchangeably. The usefulness of the Gauss-Hermite form of the expansion (\ref{HermiteSeriesEstimator}) is that it makes explicit the role of the normal distribution. Indeed it is explicit that one would expect near Gaussian distributions to be well represented with just a few coefficients.\\

The $N+1$ term truncated Gauss-Hermite expansion is defined as:

\begin{equation}
	f_{N}(x) = \sum_{k=0}^{N} a_k H_{k} (x) Z(x). \label{GHTruncDefinition}
\end{equation}

It is noteworthy that the coefficients  (\ref{GHExactCoeff}) are such that the $L_{2}$ distance between $f(x)$ and $f_{N}(x)$ is minimised i.e. no other choice of $a_k$ would lead to a better approximation of $f(x)$ in the $L_{2}$ distance sense. This follows from the fact that $f(x) \in L_{2}$ and the fact that the first $N+1$ Hermite functions constitute an orthonormal basis for a $N+1$ dimensional subspace of $L_{2}$. See \cite{orthogbook} for a succinct proof using these facts. At this point, $f(x)$ is an arbitrary function in $L_{2}$, it need not be a probability density function. For the purposes of density estimation however we will regard $f(x)$ as a probability density function (a non-negative function that integrates to one) that is also in $L_{2}$.\\

\subsubsection{Truncated Gauss-Hermite Expansions and Nonparametric density estimation} \label{miseSection}

In the context of nonparametric density estimation, a natural estimator for the $a_k$ coefficients is (following from equation (\ref{GHExactCoeff}) and the law of large numbers):

\begin{equation}
	\hat{a}_k =  \frac{\alpha_{k}}{n} \sum_{i=1}^{n} Z(\mathbf{x}_{i}) H_{k} (\mathbf{x}_{i}), \quad   \mathbf{x}_{i} \sim f(x), \label{coeffEst}
\end{equation}

i.e. the $\mathbf{x}_{i}$'s are observations from the probability distribution that we wish to estimate non-parametrically. \\

The $N+1$ term truncated Gauss-Hermite expansion with estimated coefficients is thus:

\begin{equation}
	\hat{f}_{N}(x) = \sum_{k=0}^{N} \hat{a}_k H_{k} (x) Z(x). \label{estTruncGH}
\end{equation}

Note that this is a biased estimator of $f(x)$. One common measure of the quality of the estimate of an unknown probability distribution is the mean integrated squared error (MISE). Let the true probability density function $f(x) \in L_2$. The mean integrated squared error associated with the truncated Gauss-Hermite expansion of $f(x)$ is:

\begin{equation} \label{MISE}
	E \int_{-\infty}^{\infty} (\hat{f}_{N}(x) - f(x) )^{2} dx = E\left[\sum_{k=0}^{N} \frac{2^{k-1}k!}{\sqrt{\pi}} (\hat{a}_k - a_k)^{2}\right] + \sum_{k=N+1}^{\infty} \frac{2^{k-1}k!}{\sqrt{\pi}} a_{k}^{2} \nonumber\\
\end{equation}

where we have made use of Parseval's identity, $\sum_{k=0}^{\infty} \frac{1}{\alpha_k^{2}} a_{k}^{2} = \int_{-\infty}^{\infty} f^{2}(x)dx $ and the definition of $a_{k}$ (\ref{GHExactCoeff}).   

The first term is associated with the error due to using estimates of the coefficients, $\hat{a}_k$ instead of the true, unknown, coefficients $a_k$. This is the integrated variance term. The second term of the MISE is associated with the error due to truncation. This is the integrated squared bias term. In \cite{convergence2} the MISE consistency was proved under certain conditions. \\

A few further comments are in order.  In practice, the Gauss-Hermite expansion provides a good fit to a wide variety of probability density functions \cite{gaussHermiteRobust}. For completeness however, the following shortcomings should be noted as they may be important depending on the application of the Gauss-Hermite estimate of the density. In principle, for the truncated series, the probability density function that results may be negative at certain values of $x$. Also, truncated Gauss-Hermite expansions should capture nearly Gaussian distributions well in a relatively small number of coefficients. However, for distributions that differ greatly from the Gaussian distribution, a large number of coefficients may be required for a satisfactory fit, even if convergence is guaranteed in principle.\\

In the next section we define an estimator for the cumulative distribution function associated with $f(x)$ using the Gauss-Hermite expansion and discuss numerical means of obtaining quantiles.

\section{Estimating Quantiles using the Gauss-Hermite Expansion} \label{GHQuantiles}

\subsection{Cumulative Distribution Function}

In this section we derive an analytical expression for a Gauss-Hermite based cumulative distribution function estimator. To the best of our knowledge, this is the first such analytical derivation of a cumulative distribution function estimator based on the Gauss-Hermite expansion (i.e. based on Hermite series estimators). We utilise this expression to numerically obtain quantiles. We have discussed a number of well-established results on smooth CDF estimators based on other nonparametric techniques in the introduction.  \\

\noindent
Before we begin, we recall the definitions of the Gamma functions:\\

$\Gamma(a,x)$ is the upper incomplete Gamma function defined as:

$$\Gamma(a,x) = \int_x^{\infty} t^{a-1}\,e^{-t}\,{\rm d}t,$$

$\gamma(a,x)$ is the lower incomplete Gamma function defined as:

$$\gamma(a,x) = \int_0^{x} t^{a-1}\,e^{-t}\,{\rm d}t,$$ 

and $\Gamma(a)$ is the usual Gamma function defined as:

$$\Gamma(a) = \int_0^{\infty} t^{a-1}\,e^{-t}\,{\rm d}t.$$

Now, utilising (\ref{hermitedef2}) and (\ref{estTruncGH}) a natural estimator for the cumulative distribution function associated with $f(x)$ can be obtained as follows:\\

For $x < 0$:
\begin{eqnarray}
\hat{F}_{N}(x) &=& \int_{-\infty}^{x} \hat{f}_{N}(x') dx' \nonumber \\
&=& \sum_{k=0}^{N} \hat{a}_k \int_{-\infty}^{x} H_{k} (x') Z(x') dx' \nonumber\\
&=& \sum_{k=0}^{N} \hat{a}_k k! \sum_{l=0}^{\lfloor k/2 \rfloor} \frac{(-1)^l 2^{k - 2l}}{l!(k - 2l)! \sqrt{2\pi}} \int_{-\infty}^{x} (x')^{k - 2l} e^{-x'^{2}/2} dx'\nonumber\\
&=& \sum_{k=0}^{N} \hat{a}_k k! \sum_{l=0}^{\lfloor k/2 \rfloor} \frac{(-1)^l 2^{k - 2l}}{l!(k - 2l)! \sqrt{2\pi}} (-1)^{k-2l} 2^{\frac{k}{2}-l-\frac{1}{2}} \Gamma(-l+\frac{k}{2}+\frac{1}{2},\frac{x^2}{2}) \nonumber
\end{eqnarray}

For $x \geq 0$:
\begin{eqnarray}
\hat{F}_{N}(x) &=& \int_{-\infty}^{x} \hat{f}_{N}(x') dx' \nonumber \\
&=&  \int_{-\infty}^{\infty} \hat{f}_{N}(x') dx' -\int_{x}^{\infty} \hat{f}_{N}(x') dx' \nonumber \\
&=& \sum_{k=0}^{N} \hat{a}_k \left[ \int_{-\infty}^{\infty} H_{k} (x') Z(x') dx'  - \int_{x}^{\infty} H_{k} (x') Z(x') dx' \right]\nonumber\\
&=& \sum_{k=0}^{N} \hat{a}_k k! \sum_{l=0}^{\lfloor k/2 \rfloor} \frac{(-1)^l 2^{k - 2l}}{l!(k - 2l)! \sqrt{2\pi}} \left[\int_{-\infty}^{\infty} (x')^{k - 2l} e^{-x'^{2}/2} dx' \right. \nonumber\\
&-&\left. \int_{x}^{\infty} (x')^{k - 2l} e^{-x'^{2}/2} dx' \right]\nonumber\\
&=& \sum_{k=0}^{N} \hat{a}_k k! \sum_{l=0}^{\lfloor k/2 \rfloor} \frac{(-1)^l 2^{k - 2l}}{l!(k - 2l)! \sqrt{2\pi}}  2^{\frac{k}{2}-l-\frac{1}{2}} \times \nonumber\\
&\times&\left[ \left[(-1)^{k-2l}+1\right]\Gamma(-l+\frac{k}{2}+\frac{1}{2}) - \Gamma(-l+\frac{k}{2}+\frac{1}{2},\frac{x^2}{2}) \right]\nonumber\\
&=&\sum_{k=0}^{N} \hat{a}_k k! \sum_{l=0}^{\lfloor k/2 \rfloor} \frac{(-1)^l 2^{k - 2l}}{l!(k - 2l)! \sqrt{2\pi}}  2^{\frac{k}{2}-l-\frac{1}{2}} \times \nonumber\\
&\times&\left[ \left[(-1)^{k-2l}\right]\Gamma(-l+\frac{k}{2}+\frac{1}{2}) + \gamma(-l+\frac{k}{2}+\frac{1}{2},\frac{x^2}{2}) \right]\nonumber
\end{eqnarray}

Thus,

\begin{equation}\label{cumulativedist}
\hat{F}_{N}(x)\!=\!\begin{cases}
	 \sum_{k=0}^{N} \hat{a}_k k! \sum_{l=0}^{\lfloor k/2 \rfloor} \frac{(-1)^{l} 2^{\frac{3k}{2} - 3l -1} \left[ \left[(-1)^{k-2l}\right]\Gamma(-l+\frac{k}{2}+\frac{1}{2}) + \gamma(-l+\frac{k}{2}+\frac{1}{2},\frac{x^2}{2}) \right]}{l!(k - 2l)!\sqrt{\pi}}\! & \\ \quad \mbox{if } x \geq 0, \\
	\sum_{k=0}^{N} \hat{a}_k k! \sum_{l=0}^{\lfloor k/2 \rfloor} \frac{(-1)^{-l+k} 2^{\frac{3k}{2} - 3l -1} \Gamma(-l+\frac{k}{2}+\frac{1}{2},\frac{x^{2}}{2})}{l!(k - 2l)!\sqrt{\pi}}& \\ \quad \mbox{if } x < 0.
	\end{cases}
\end{equation}

The expression (\ref{cumulativedist})  allows us to directly estimate the cumulative distribution function without the need to numerically integrate the estimated probability density function (\ref{estTruncGH}). \\

An alternative estimator for the cumulative distribution function is as follows:

\begin{equation}\label{cumulativedistalt}
\hat{F}_{N}(x) =\begin{cases}
		1-\sum_{k=0}^{N} \hat{a}_k k! \sum_{l=0}^{\lfloor k/2 \rfloor} \frac{(-1)^{l} 2^{\frac{3k}{2} - 3l -1} \Gamma(-l+\frac{k}{2}+\frac{1}{2},\frac{x^{2}}{2})}{l!(k - 2l)!\sqrt{\pi}} \mbox{ if } x \geq 0 \\
		\sum_{k=0}^{N} \hat{a}_k k! \sum_{l=0}^{\lfloor k/2 \rfloor} \frac{(-1)^{-l+k} 2^{\frac{3k}{2} - 3l -1} \Gamma(-l+\frac{k}{2}+\frac{1}{2},\frac{x^{2}}{2})}{l!(k - 2l)!\sqrt{\pi}} \,\, \mbox{ if } x < 0
	\end{cases}
\end{equation}

This expression is derived in the same manner as (\ref{cumulativedist}) except that the integral $\int_{-\infty}^{\infty} \hat{f}_{N}(x') dx'$ is replaced with unity. We have found that empirically this estimator yields more accurate results. This may depend on the situation however.\\

\subsection{Inverse Cumulative Distribution Function} \label{quantilefnc}

The inverse cumulative distribution function or quantile function is defined in equation \eqref{quantileDef}. We can utilise (\ref{cumulativedist}) (or (\ref{cumulativedistalt})) along with a numerical root-finding algorithm to determine the value of the p-th quantile, $x_{p} = q(p)$,  $0<p<1$. Newton's method can be applied for example. In this case, the following equation is iteratively evaluated until sufficient accuracy is achieved:

\begin{equation}
	\hat{x}_{p}^{(i+1)} = \hat{x}_{p}^{(i)} - \frac{\hat{F}_{N}(\hat{x}_{p}^{(i)})-p}{\hat{f}_{N}(\hat{x}_{p}^{(i)})}, \label{Newton}
\end{equation}

where $\hat{f}_{N}(x)$ is given in (\ref{estTruncGH}) and $\hat{F}_{N}(x)$ is given in (\ref{cumulativedist}) (or (\ref{cumulativedistalt})). Naturally, convergence behaviour will depend on the choice of initial value, $\hat{x}_{p}^{(0)}$, and the properties of $\hat{f}_{N}(x)$. In principle convergence may be slow, or the method may not converge at all. If techniques such as Newton's method and related methods prove unstable in a given setting, more robust numerical root-finding algorithms can be applied. The best choice of root-finding algorithm and optimal initial value selection are areas for future research. \\

It is important to note that in many cases of interest, the quantile itself is not required but rather it is necessary to determine whether an observation is above or below a particular quantile (consider outlier detection for example). In this case, no root finding is required. One simply plugs the observation into the cumulative distribution function and determines whether the cumulative probability is less than or greater than $p$. 

\section{Online Quantile Estimation: Static Quantiles} \label{GHOnlineStat}

The problem we treat in this section involves estimating an unknown cumulative distribution function (and associated inverse cumulative distribution function) from a stream of independent and identically-distributed (i.i.d.) continuous random variable data. The Gauss-Hermite expansion furnishes an efficient means to achieve this. The primary reason for this is that the coefficients in the expansion can be updated with each new observation without recalculating the entire sum in (\ref{coeffEst}). We can just incorporate a new term corresponding to the new observation and maintain a running average for each coefficient. Moreover, we can simply plug in these updated coefficients into the analytical expression for the cumulative distribution function we have derived (\ref{cumulativedist}) (or (\ref{cumulativedistalt})). Any quantile can then be obtained by a simple numerical root finding procedure (section \ref{quantilefnc}). \\

The basic algorithm can be summarised as follows:

\begin{enumerate}
	\item Initialize $N+1$ coefficients as follows $\hat{a}_k^{(0)} = \alpha_{k} Z(\mathbf{x}_{0}) H_{k} (\mathbf{x}_{0})$, where $\mathbf{x}_{0}$ is the first observation from the data stream and $k=0 \dots N$.
	\item For each new observation, $\mathbf{x}_{i}$, update $\hat{a}_0, \dots \hat{a}_N$ as follows:
	\begin{equation}
		\hat{a}_k^{(i)} = \frac{1}{i} \left[(i-1)\hat{a}_k^{(i-1)} + \alpha_{k} Z(\mathbf{x}_{i}) H_{k} (\mathbf{x}_{i})\right], \,  k = 0, \dots, N. \label{movingAvg}
	\end{equation}
	\item Plug the updated coefficients $\hat{a}_0, \dots \hat{a}_N$ into the expressions (\ref{estTruncGH}) and (\ref{cumulativedist}) (or (\ref{cumulativedistalt})) to obtain updated estimates of the probability density function and cumulative distribution function respectively.
	\item Utilise a numerical root finding algorithm, along with the updated cumulative distribution function (and potentially the updated probability density function) to determine any arbitrary quantile.
\end{enumerate}

The above algorithm can also be applied to summarise the distribution of massive datasets in an efficient, one-pass manner which should be particularly useful when the size of the dataset is larger than the available memory. \\

The computational cost of \textit{updating} each of the coefficients (\ref{coeffEst}) is manifestly constant ($O(1)$) and does not depend on the number of previous observations. Also, since the cumulative distribution function only depends on the coefficients and has no explicit dependence on the observations, the time complexity of \textit{updating} the cumulative distribution function following the arrival of a new observation is also $O(1)$. Similarly, the computational cost of the numerical root finding algorithm yielding any quantile of interest from the updated cumulative distribution does not depend on the number of previous observations. This is to be contrasted with deterministic  approaches to obtaining quantiles such as efficient heap based median maintenance which has a time complexity $O(\mbox{log i})$ at the i-th observation and has growing space requirements.  \\

While the updating procedure for the coefficients is fully sequential, it is clear that since we use a fixed and constant $N$ the resultant Gauss-Hermite estimate of the probability density function is biased and thus, the resultant CDF and quantile estimates will in general be biased too. Thus our quantile estimator is sequential but biased. This bias does not prevent the estimator from being useful however.\\

Data streams that have static quantiles are likely to be less prevalent than those that have dynamic quantiles (quantiles that vary over time). Indeed, many real-world data streams of interest exhibit non-stationarity. In section \ref{EWGH} we consider how the proposed algorithm can be modified to treat quantile estimation in the more realistic, dynamic setting.

\subsection{Selection of $N$} \label{paramselect1}

It is natural to assume that the quality of the CDF and quantile estimates is related to the MISE of $\hat{f}_{N}(x)$. Indeed, we demonstrate in section \ref{theoreticalquality} that the MISE of $\hat{f}_{N}(x)$ directly determines bounds on the MSE of the CDF and the MAE of the resultant quantile estimates for distributions with non-negative support under certain conditions.  When viewed in the context of the MISE of $\hat{f}_{N}(x)$, the choice of $N$ controls the trade-off between the integrated variance and integrated squared bias of the estimate of $f(x)$. For the Hermite series estimators, the integrated variance term vanishes as $n \to \infty$ for fixed $N$ (under certain conditions on $f(x)$, see \cite{convergence2}). This leaves the contribution from the integrated bias term. The higher the value of $N$, the smaller the integrated bias. Thus in the setting of streaming data or one pass analysis of a massive data set, where we regard $n \to \infty$, $N$ would naively be made as large as possible to minimise the bias and hence the MISE (recall that the MISE controls the quality of the CDF and quantile estimates, see section \ref{qualityQuantile}). It is worth noting however that memory requirements and processing time increase with N since more coefficients have to be stored and updated. In addition, early quantile estimates could be poor for large N. If the intended application is sensitive to poor early quantile estimates, a small sample from the data stream or massive data set can be analysed in order to select $N$. While it is clear that in general the optimal N is different for the PDF, CDF and quantile estimates, our results suggest that it is a reasonable starting point to attempt to minimise the MISE of $\hat{f}_{N}(x)$. Principled techniques exist to select the (MISE) optimal $N$ such as the Kronmal-Tarter optimal stopping rule algorithm \cite{kronmal1968estimation}, \cite{ott1976some}. This algorithm must be applied with care though as it may perform poorly if $f(x)$ is multimodal or peaked as pointed out in \cite{diggle1986selection}. The reader is referred to \cite{diggle1986selection} and \cite{hart1985data} for improvements to the algorithm. Data driven selection of $N$ specific to the CDF and quantile estimates is an area of future research. In our simulation studies we demonstrate that above a minimum size of N, the effectiveness of the algorithm is in fact not critically dependent on the choice of N. Good results can be obtained for a range of values of N. Through extensive empirical studies we have determined that a value of N = 6 yields good results for data with a unimodal distribution for example.
 
\section{Online Quantile Estimation: Dynamic Quantiles} \label{EWGH}

The problem we treat in this section involves obtaining a local estimate of an unknown cumulative distribution function (and associated inverse cumulative distribution function) from a stream of continuous random variable data with dynamic quantiles. In order to do this we replace the Gauss-Hermite coefficient estimator defined in (\ref{coeffEst}) with an exponentially weighted moving average estimator for the coefficients. We will term the resulting expansion an exponentially weighted Gauss-Hermite expansion (EWGH expansion). The new estimator for the coefficients is given by:

\begin{eqnarray}
	\hat{a}^{(i)}_k &=& \lambda \left[\alpha_{k} Z(\mathbf{x}_{i}) H_{k} (\mathbf{x}_{i})\right] + (1-\lambda)\hat{a}^{(i-1)}_k, \nonumber \\ 
\hat{a}^{(0)}_k &=& \left[\alpha_{k} Z(\mathbf{x}_{0}) H_{k} (\mathbf{x}_{0})\right], \,  k = 0, \dots, N, \label{GHCoeffEWMA}
\end{eqnarray}

where $0<\lambda \leq 1$ controls the weight of new observations (and controls how rapidly the weightings of older observations decrease). This weighting scheme allows the local behaviour of the data stream to be tracked. The algorithm for obtaining arbitrary quantiles presented in the previous section is essentially unchanged except that we replace (\ref{movingAvg}) with (\ref{GHCoeffEWMA}). To re-iterate, the updated coefficients can then be plugged into into the expressions (\ref{estTruncGH}) and (\ref{cumulativedist}) (or (\ref{cumulativedistalt})) to obtain updated estimates of the probability density function and cumulative distribution function respectively. A numerical root finding procedure can again be applied to obtain arbitrary quantiles.\\

\subsection{Selection of the Parameters $\lambda$ and $N$}

For the choice of $N$, the same broad considerations apply as in section \ref{paramselect1} i.e. the more complex the probability distribution of the data being analyzed, the higher the appropriate $N$ to ensure a sufficiently low bias. In our simulation studies we have observed that a value of $N=6$ gives competitive performance for all choices in the set $\lambda=0.01, 0.05, 0.1$, for unimodal distributions.

Given a choice of $N$, there are two factors to consider when selecting $\lambda$. The first is how quickly the quantiles of the non-stationary data stream are expected to vary. We expect that the optimal $\lambda$ will be smaller for slowly varying quantiles and larger for rapidly varying quantiles. By selecting $\lambda$, one is essentially selecting an effective window size of previously observed data to include in the quantile estimation. This follows from the fact that more recent data is weighted more heavily than older data. Consider the fraction of the weight included in the most recent $r$ terms.

$$\lambda \sum_{j=0}^{r-1} (1-\lambda)^{j} = 1-(1-\lambda)^{r}.$$

This is to be contrasted with the weight of the first (oldest) term:

$$(1-\lambda)^{r}$$.

If we define the effective window size as that number of observations for which $99.9\%$ of the weight is contained in the most recent $r$ terms or equivalently that the remaining $0.1 \%$ of the weight is associated with the first (oldest) term, then we see that the effective window size is:

$$r = \frac{\log{0.001} }{\log{(1-\lambda)}}$$

We include below a tabulation of some commonly used values of $\lambda$ in EWMA applications and their associated effective window size in number of observations.

\begin{center}
		\begin{tabular}{ l | c}
		    \hline
		    $\lambda$ & Effective Window Size  \\ \hline
			0.01 & 687   \\ 
			0.05 & 135 \\ 
			0.1 &  66  \\
			0.2 & 31\\
		    \hline
		\end{tabular}	
\end{center}

The ideal effective window size should be selected by judgement and domain specific knowledge of the data being analysed. For example, when analysing high frequency forex return data, one may expect that the most recent observations are the most pertinent and only the previous few minutes of observations would be relevant to estimating the current, local behaviour of the process. The second factor to consider is that $\lambda$ cannot be too large. As we demonstrate in section \ref{theoreticalquality}, the MISE of $\hat{f}_N(x)$ determines the quality of both CDF and quantile estimates under certain conditions. Using the bound on the MISE that we derive in theorem \ref{theoremEWGHMISE1} we see that we can achieve a small integrated variance term by ensuring: 

$$\lambda N^{1/2}$$ 

is sufficiently small. As a starting point, we have found through extensive empirical studies that the common choices of $\lambda=0.01, 0.05, 0.1$ yield good performance of the algorithm in a number of scenarios. Indeed, in our simulation studies we demonstrate that the EWGH algorithm provides good results over this range of values for $\lambda$. \\

\section{Quality of CDF and Quantile estimates}\label{theoreticalquality}

In this section we derive error bounds on the mean squared error of the Gauss-Hermite CDF estimator and the mean absolute error (MAE) of the quantile estimator for distributions with support on the positive half-real line, subject to some additional conditions. In particular, we demonstrate that these error bounds directly depend on the MISE of the associated probability density function estimates. This greatly simplifies obtaining the aforementioned error bounds and allows us to examine the asymptotic behaviour of the Gauss-Hermite based CDF estimator and the associated quantile estimator. While these results \textit{do not} directly apply to the sequential quantile estimation setting (since $N$ and $\lambda$ are fixed in that case), they are novel and interesting in their own right. These asymptotic results provide general context and provide comfort that the behaviour of estimators constructed from the Hermite series probability density estimators have sensible asymptotic properties. To obtain asymptotic results we utilise existing MISE consistency results along with the associated rates for the standard Gauss-Hermite expansion \cite{convergence2} as well as novel MISE results for the exponentially weighted Gauss-Hermite expansion which we derive in appendix \ref{theoryEWGH}. All the results derived below are easily extended to $f(x)$ with support on a bounded interval, $[a,b]$ with $\hat{F}_{N}(x) =\int_{a}^{x} \hat{f}_{N}(x') dx'$(omitted for brevity). Our approach below is not directly generalisable to $(-\infty,\infty)$ however and deriving error bounds in that case is an interesting problem for further study.

\subsection{Quality of the Cumulative Distribution Function Estimate} \label{cumqual}

In what follows we assume that $f(x)$ is supported on $[0,\infty)$. \\

For non-negative random variables we obtain the simpler estimator:

\begin{eqnarray}\label{cumulativedistpos}
\hat{F}_{N}(x) &=&\int_{0}^{x} \hat{f}_{N}(x') dx' \nonumber \\
	&=& \sum_{k=0}^{N} \hat{a}_k k! \sum_{l=0}^{\lfloor k/2 \rfloor} \frac{(-1)^{l} 2^{\frac{3k}{2} - 3l -1} \left[\gamma(-l+\frac{k}{2}+\frac{1}{2},\frac{x^2}{2}) \right]}{l!(k - 2l)!\sqrt{\pi}},  \, x \geq 0.
\end{eqnarray}

The expression (\ref{cumulativedistpos}) differs from (\ref{cumulativedist}) since the domain of integration is different. We consider the Mean Squared Error (MSE) criterion and an integrated weighted MSE criterion (inspired by the Cramer-von Mises criterion) as measures of the quality of the estimated cumulative distribution function. The integrated weighted MSE criterion is defined as follows:

\begin{eqnarray}
	\overline{\omega^{2}} &=& E\int_{0}^{\infty} \left[\hat{F}_{N}(x) - F(x) \right]^{2} f(x) dx \nonumber \\
	&=&\int_{0}^{\infty} E\left[\hat{F}_{N}(x) - F(x) \right]^{2} f(x) dx, \label{CramerVonCrit}
\end{eqnarray}

where we have made use of Fubini's theorem which allows us to interchange the ordering of the integrals. Note that the PDF $f(x)$ is the weighting factor in (\ref{CramerVonCrit}).

\begin{proposition} \label{propositionMSE}

Suppose $f(x)$ is supported on $[0,\infty)$ and $f(x) \in L_{2}$ then we have:

$$E \left| \hat{F}_{N}(x) - F(x) \right|^2 \leq  x \, \mbox{MISE} (\hat{f}_{N}),$$

for fixed $x$. \\

\end{proposition}

\begin{proof}

For a non-negative random variable:

$$\left| \hat{F}_{N}(x) - F(x) \right| = \left| \int_{0}^{x} \left[\hat{f}_{N}(x') - f(x') \right]dx' \right|.$$

By the Cauchy-Schwarz inequality:

$$\left| \hat{F}_{N}(x) - F(x) \right|^2 \leq \left[\int_{0}^{x} \left| \hat{f}_{N}(x') - f(x') \right|^{2} dx' \right] \left[ \int_{0}^{x} dx' \right].$$

Thus:

$$\left| \hat{F}_{N}(x) - F(x) \right|^2 \leq x \left[\int_{0}^{x} \left| \hat{f}_{N}(x') - f(x') \right|^{2} dx' \right].$$ 

Now, since $(\hat{f}_{N}(x') - f(x') )^{2}$ is non-negative we have, 

$$ \int_{0}^{x} (\hat{f}_{N}(x') - f(x') )^{2} dx' \leq \int_{-\infty}^{\infty} (\hat{f}_{N}(x') - f(x') )^{2} dx'.$$

Thus:

\begin{eqnarray}
	E\left| \hat{F}_{N}(x) - F(x) \right|^2 &\leq & x E\left[\int_{-\infty}^{\infty} \left| \hat{f}_{N}(x') - f(x') \right|^{2} dx' \right] \nonumber\\
	&=& x \, \mbox{MISE} (\hat{f}_{N}).\label{misetomse}
\end{eqnarray}

This implies that if we have an upper bound for the MISE of $\hat{f}_{N}(x)$ we can bound the MSE of $\hat{F}_{N}(x)$. 

\end{proof}

\begin{proposition} \label{propositionOmega}

Suppose $f(x)$ is supported on $[0,\infty)$ and $f(x)$ has a finite mean, $\mu<\infty$ then we have:

$$\overline{\omega^{2}} \leq \mbox{MISE} (\hat{f}_{N}) \mu.$$

\end{proposition}

\begin{proof}

Utilizing proposition \ref{propositionMSE} we have, 

\begin{eqnarray}
	\overline{\omega^{2}} &\leq& \mbox{MISE} (\hat{f}_{N})  \int_{0}^{\infty} x f(x) dx \nonumber \\
	&=& \mbox{MISE} (\hat{f}_{N}) \mu, \label{cramer}
\end{eqnarray}

where $\mu$ is the mean of $f(x)$. 

\end{proof}

We now consider the consistency and associated rate of convergence for the CDF estimator (\ref{cumulativedistpos}) defined from the standard Gauss-Hermite coefficients (\ref{coeffEst}).

\begin{theorem}\label{theoremPointConsis}
Suppose $f(x)$ is supported on $[0,\infty)$ and $f(x) \in L_{2}$. In addition suppose that $\frac{N^{\frac{1}{2}}(n)}{n} \to 0$ as $N(n), n \to \infty$ and $E|X|^{\frac{2}{3}}<\infty$ then we have:

$$E \left| \hat{F}_{N}(x) - F(x) \right|^2 \to 0. $$
\end{theorem}

\begin{proof}
The result follows from proposition \ref{propositionMSE} and the fact that  $ \mbox{MISE} (\hat{f}_{N}) \to 0$ under the conditions $\frac{N^{\frac{1}{2}}(n)}{n} \to 0$ as $N(n), n \to \infty$ and $E|X|^{\frac{2}{3}}<\infty$ \cite{convergence2}.
\end{proof}

\begin{theorem} \label{propositionMSERate}

Suppose $f(x)$ is supported on $[0,\infty)$, $f(x) \in L_{2}$, $r \geq1$ derivatives of $f(x)$ exist and $(x-\frac{d}{dx})^r f(x) \in L_{2}$. Suppose in addition that $\frac{N^{\frac{1}{2}}(n)}{n} \to 0$ as $N(n), n \to \infty$ and $E|X|^{\frac{2}{3}}<\infty$ then if:

$$N(n) \sim n^{2/(2r+1)}, \mbox{ we have }$$ 

$$E \left| \hat{F}_{N}(x) - F(x) \right|^2 = x \, O\left(n^{-2r/(2r+1)}\right).$$ \\

\end{theorem}

\begin{proof}

In \cite{convergence2} it is established that provided $f(x) \in L_{2}$, $r \geq 1$ derivatives of $f(x)$ exist, $(x-\frac{d}{dx})^r f(x) \in L_{2}$, $\frac{N^{\frac{1}{2}}(n)}{n} \to 0$ as $N(n), n \to \infty$ and $E|X|^{\frac{2}{3}}<\infty$ then if:

$$N(n) \sim n^{2/(2r+1)}, \mbox{ we have }$$ 

$$\mbox{MISE} (\hat{f}_{N}) = O\left(n^{-2r/(2r+1)}\right).$$ 

Combining this with proposition \ref{propositionMSE} completes the proof.
 
\end{proof}

Note that for $r=1$ the rate is $O(n^{-2/3})$. It is important to note that this rate is suboptimal compared to the smooth kernel CDF estimate rate which is $O(n^{-1})$.  For smooth probability density functions ($r\to \infty$) satisfying the appropriate conditions, the rate for the Gauss-Hermite CDF estimator approaches $O(n^{-1})$.

\begin{theorem}\label{cramerconsis}

Suppose $f(x)$ is supported on $[0,\infty)$ and $f(x)$ has a finite mean, $\mu<\infty$. In addition suppose that $\frac{N^{\frac{1}{2}}(n)}{n} \to 0$ as $N(n), n \to \infty$ then we have:

$$\overline{\omega^{2}} \to 0.$$

\end{theorem}

\begin{proof}

The result follows directly from proposition \ref{propositionOmega} and the fact that\\ $\mbox{MISE} (\hat{f}_{N}) \to 0$ under the conditions $\frac{N^{\frac{1}{2}} (n)}{n} \to 0$ as $N(n), n \to \infty$ and $E|X|^{\frac{2}{3}}<\infty$  \cite{convergence2} (for positive random variables, we have $E(X)<\infty$ implying $E|X|^{\frac{2}{3}}<\infty$ by the Lyapunov inequality).

\end{proof}

\begin{theorem}\label{cramerrate}

Suppose $f(x)$ is supported on $[0,\infty)$, $f(x)$ has a finite mean, $\mu<\infty$, $f(x) \in L_{2}$, $r \geq1$ derivatives of $f(x)$ exist and $(x-\frac{d}{dx})^r f(x) \in L_{2}$. Suppose in addition that $\frac{N^{\frac{1}{2}}(n)}{n} \to 0$ as $N(n), n \to \infty$ then if:

$$N(n) \sim n^{2/(2r+1)}, \mbox{ we have }$$ 

$$\overline{\omega^{2}} =O\left(n^{-2r/(2r+1)}\right).$$

\end{theorem}

\begin{proof}

This follows directly from proposition \ref{propositionOmega} and the MISE result referenced in the proof of theorem  \ref{propositionMSERate}.

\end{proof}

\begin{remark}
An important class of distributions for which we can gain further insight into the theoretical performance of the Gauss-Hermite CDF estimator (and quantile estimator in principle) is power-law distributions. These are heavy-tailed distributions and are highly relevant in a number of fields including finance \cite{clauset2009power}. We utilise the following definition of the power law distribution:

$$f(x) = \frac{\alpha -1}{x_{\mbox{min}}} \left( \frac{x}{x_{\mbox{min}}}\right)^{-\alpha},$$

where $\alpha>1$ is a requirement for normalisability. In addition, we assume $x_{\mbox{min}}>0$. Thus $f(x)$ is supported on $[0,\infty)$, $f(x) \in L_{2}$ and all derivatives of $f(x)$ exist for $x \geq x_{\mbox{min}}$. If we require in addition that $\alpha>2$, then we have a finite mean $E(X)<\infty$. The condition in the theorems proven above that $(x-\frac{d}{dx})^r f(x) \in L_{2}$ can be related to $\alpha$ as follows:\\

Denote the operator $\frac{d}{dx}$ as $D$. Now:

\begin{align*}
	(x-D)^r f(x)&=\frac{\alpha -1}{\left(x_{\mbox{min}}\right)^{1-\alpha}} \left(x^r + x^{r-1}D +x^{r-2}Dx + x^{r-2}D^{2} + \dots+D^r \right)x^{-\alpha}\nonumber\\
	&= O(x^{-(\alpha - r)}),\nonumber
\end{align*}

since $x$ and $D$ are non-commutative operators. Thus\\ $\left[(x-D)^r f(x)\right]^{2} = O(x^{-2(\alpha - r)})$. This implies that for $\left[(x-D)^r f(x)\right]^{2}$ to be integrable, we must have $-2(\alpha - r) <-1$. Which implies $r< \alpha- \frac{1}{2}$ and thus $r=\lceil \alpha- \frac{1}{2} \rceil -1$. Thus if we also have $\frac{N^{\frac{1}{2}}(n)}{n} \to 0$ as $N(n), n \to \infty$, theorem \ref{theoremPointConsis} and theorem \ref{propositionMSERate} imply that the Gauss-Hermite CDF estimator is consistent for power-law distributions and the rate is $E \left| \hat{F}_{N}(x) - F(x) \right|^2 = x \, O\left(n^{-(2\lceil \alpha- \frac{1}{2} \rceil -2)/(2\lceil \alpha- \frac{1}{2} \rceil -1)}\right)$. Similarly theorem \ref{cramerconsis} and theorem \ref{cramerrate} imply that $\overline{\omega^{2}} \to 0$ and the rate is $\overline{\omega^{2}} =O\left(n^{-(2\lceil \alpha- \frac{1}{2} \rceil -2)/(2\lceil \alpha- \frac{1}{2} \rceil -1)}\right)$. Thus for power-law distributions that have a finite mean, the rates are $O\left(n^{-2/3}\right)$ or better. For power-law distributions that have a finite variance, $\alpha > 3$ and thus the rates are $O\left(n^{-4/5}\right)$ or better. As $\alpha \to \infty$ (along with the number of finite moments of the power-law distribution), the Gauss-Hermite rates approach $O\left(n^{-1}\right)$.\\
\end{remark}

To summarise, we have derived the asymptotic results above in the setting where $N$ depends on $n$, i.e. $N=N(n)$. These results give comfort that the Gauss-Hermite based CDF estimator has sensible asymptotic behaviour. Our proposed online estimators have $N$ fixed however. Thus results such as theorem \ref{theoremPointConsis} do not apply directly. Instead, as $n \to \infty$, these online estimators have MSE bounds determined by the integrated squared bias of the truncated Gauss-Hermite PDF estimators on which the plug-in CDF estimators are based. While this bias persists for fixed $N$ even as $n \to \infty$, these estimators are nonetheless very useful in practice as we have demonstrated in our simulation and real data results. In addition, we can still gain insight into the behaviour of these estimators at fixed $N$. We now consider the behaviour of the EWGH CDF estimator defined from utilising the exponentially weighted Gauss-Hermite coefficients (\ref{GHCoeffEWMA}) in the expression (\ref{cumulativedistpos}) for the CDF where $\lambda$ is fixed, $N$ is fixed and sufficiently large and $n \to \infty$. We begin with the case of i.i.d. data drawn from $f(x)$.

\begin{theorem} \label{EWGHCDFIID}
Suppose $f(x)$ is supported on $[0,\infty)$, $f(x)$ has a finite mean, $\mu<\infty$, $f(x) \in L_{2}$, $r \geq1$ derivatives of $f(x)$ exist and $(x-\frac{d}{dx})^r f(x) \in L_{2}$. Then for $N$ fixed and sufficiently large and $n \to \infty$.

$$E \left| \hat{F}_{N}(x) - F(x) \right|^2 = x \, \left[ O(N^{1/2}) \left[ \frac{\lambda}{2-\lambda} \right] + O(N^{-r}) \right],$$ 

and

$$\overline{\omega^{2}} = O(N^{1/2}) \left[ \frac{\lambda}{2-\lambda} \right] + O(N^{-r}).$$

\end{theorem}

\begin{proof}
The result follows from proposition \ref{propositionMSE} and \ref{propositionOmega} respectively along with theorem \ref{theoremEWGHMISE1}.
\end{proof}

We now treat the case of independent, non-identically distributed data. In particular, we consider the case of a change point where the distribution changes from $f_{1}(x)$ to $f_{2}(x)$.  We consider this a fundamental example of non-identically distributed data. 

\begin{theorem}
Suppose $s+1$ observations are drawn from a probability distribution $f_{1} \in L_{2}$ followed by a further $t$ observations from a second distribution $f_{2} \in L_{2}$ i.e. we assume an independent sequence of $s+1$ observations $\mathbf{x}_{0}, \dots, \mathbf{x}_{s} \sim f_{1}(x)$ followed by an independent sequence of $t$ observations, $\mathbf{x}_{s+1}, \dots, \mathbf{x}_{t+s} \sim f_{2}(x)$. If $r \geq 1$ derivatives of $f_{2}(x)$ exist and $(x-\frac{d}{dx})^r f_{2}(x) \in L_{2}$ and both distributions have a finite mean, then:

\begin{align*}
&E \left| \hat{F}_{N}(x) - F_{2}(x) \right|^2 \nonumber\\
&=x \left[ O(N^{1/2}) \left[4(1-\lambda)^{2t}\!+\!\frac{\lambda}{2-\lambda}\left[1\!-\!(1\!-\!\lambda)^{2(s+t)} \right]\!+\!(1-\lambda)^{2(s+t)} \right]\!+\!O(N^{-r})\right]\nonumber
\end{align*}

and

$$\overline{\omega^{2}}\!=\!O(N^{1/2})\left[4(1\!-\!\lambda)^{2t}\!+\!\frac{\lambda}{2\!-\!\lambda}\left[1\!-\!(1\!-\!\lambda)^{2(s+t)}\right]\!+\!(1\!-\!\lambda)^{2(s+t)} \right]\!+\!O(N^{-r}),$$

where $F_{2}(x)$ is the CDF of $f_{2}(x)$ and $\overline{\omega^{2}} = \int_{0}^{\infty} E\left[\hat{F}_{N}(x) - F_{2}(x) \right]^{2} f_{2}(x) dx$.

\end{theorem}

\begin{proof}
The result follows from proposition \ref{propositionMSE} and \ref{propositionOmega} respectively along with theorem \ref{theoremEWGHMISE2}.
\end{proof}

Note that asymptotically we can choose $\lambda(n) \to 0, \, n \to \infty$ such that\\ $E \left| \hat{F}_{N}(x) - F(x) \right|^2 \to 0$ and $\overline{\omega^{2}} \to 0$ for both the i.i.d. case and the non-identically distributed, independent (change point) case. We do not present the proof here for the sake of brevity.  

\subsection{Quality of Quantile Estimate}\label{qualityQuantile}

\begin{theorem}
Suppose $f(x)$ is supported on $[0,\infty)$. In addition we suppose that the true quantile $x_{p}$ lies between $[x_{p}^{min}, x_{p}^{max}]$ and that $f(x) \geq d, \, d>0$ for $x \in [x_{p}^{min}, x_{p}^{max}]$. Finally we assume that we know $x_{p}^{min}, x_{p}^{max}, d$ allowing us to refine the Gauss-Hermite quantile estimate as follows:

\begin{equation}\label{refinedEst}
\hat{x}_{p} \!=\!\begin{cases}
		\hat{F}_{N}^{-1}(p) &\mbox{if }\hat{F}_{N}^{-1}(p)\mbox{ exists in }[x_{p}^{min},x_{p}^{max}]\\
		& \mbox{ and }\hat{f}_{N}(x)\!\geq\!d,\, x\!\in\![x_{p}^{min}, x_{p}^{max}]\\
		\mbox{undefined} &\mbox{otherwise}.
	\end{cases}
\end{equation}

For well defined $\hat{x}_{p}$:

\begin{equation}
E\left|  x_{p} -\hat{x}_{p}  \right| \leq  \frac{\sqrt{x_{p}}}{d}  \sqrt{\mbox{MISE}(\hat{f}_N(x))}.\nonumber
\end{equation}
 
\end{theorem}

\begin{proof}

In the proof of proposition \ref{propositionMSE} we established:

$$\left| \hat{F}_{N}(x) - F(x) \right|^2 \leq x \int_{-\infty}^{\infty} (\hat{f}_{N}(x') - f(x') )^{2} dx'.$$

Thus:

\begin{eqnarray}
\left| \hat{F}_{N}(x_{p}) - F(x_{p}) \right| &\leq& \sqrt{x_{p}} \sqrt{\int_{-\infty}^{\infty} (\hat{f}_{N}(x') - f(x') )^{2} dx'},\\ 
\left| \hat{F}_{N}(x_{p}) - p \right| &\leq& \sqrt{x_{p}} \sqrt{\int_{-\infty}^{\infty} (\hat{f}_{N}(x') - f(x') )^{2} dx'},
\end{eqnarray}

where $x_{p} = q(p)$. Provided $\hat{x}_{p} =\hat{F}_{N}^{-1}(p)$ exists, this implies:

\begin{eqnarray}
\left| \hat{F}_{N}(x_{p}) - \hat{F}_{N}(\hat{x}_{p} ) \right| &\leq& \sqrt{x_{p}} \sqrt{\int_{-\infty}^{\infty} (\hat{f}_{N}(x') - f(x') )^{2} dx'}\\
\left| \int_{\hat{x}_{p}}^{x_{p} }\hat{f}_{N}(x) dx \right| &\leq& \sqrt{x_{p}} \sqrt{\int_{-\infty}^{\infty} (\hat{f}_{N}(x') - f(x') )^{2} dx'}\\
\left|  x_{p} -\hat{x}_{p}  \right| \left|  \hat{f}_{N}(\hat{c})  \right| &\leq& \sqrt{x_{p}} \sqrt{\int_{-\infty}^{\infty} (\hat{f}_{N}(x') - f(x') )^{2} dx'},
\end{eqnarray}

where $\hat{c}$ lies in the interval $(x_{p}, \hat{x}_{p})$ if $\hat{x}_{p} > x_{p}$ or $(\hat{x}_{p},x_{p})$ if $\hat{x}_{p} < x_{p}$. We have applied the mean value theorem and thus $\hat{f}_{N}(\hat{c})$ is equal to the {\it mean} value of $\hat{f}_{N}$ in the interval i.e. $\hat{f}_{N}(\hat{c}) = \frac{1}{x_{p} - \hat{x}_{p} }\int_{x_{p}}^{\hat{x}_{p} }\hat{f}_{N}(x) dx$. \\

Thus, provided $\hat{f}_{N}(\hat{c}) \neq 0$, we have:

\begin{equation}
\left|  x_{p} -\hat{x}_{p}  \right| \leq \frac{\sqrt{x_{p}}}{\left|  \hat{f}_{N}(\hat{c})  \right| }  \sqrt{\mbox{ISE}(\hat{f}_{N})},
\end{equation}

To get a more concrete result we assume we utilize our refined quantile estimator. This estimator is quite natural as it restricts the quantile estimates to those formed from bona-fide probability densities i.e. those $\hat{f}_{N}(x) $ that are non-negative. Moreover, requiring $\hat{f}_{N}(x)>0$ ensures that there is a unique solution for $\hat{x}_{p} =\hat{q}_{N}(p)$. Finally, it allows us to reject quantile estimates that are out of bounds. For the estimates that are not undefined we have, via the Cauchy-Schwarz inequality:

\begin{equation}
E\left|  x_{p} -\hat{x}_{p}  \right| \leq  \frac{\sqrt{x_{p}}}{d}  \sqrt{\mbox{MISE}(\hat{f}_N)}.
\end{equation}

\end{proof}

This result again demonstrates the direct link to the MISE of the probability density function and allows one to, in principle, establish asymptotic properties under certain conditions (similar to above results for the CDF). As a note to the practitioner, this refined quantile estimator can be viewed as one where we reject estimates that are out of bounds and those that are not constructed from bona-fide probability density estimates. We expect this to be an infrequent scenario unless $N$ is too small and the probability density estimator is heavily biased or too few observations have been incorporated into the estimate. We base this expectation on extensive empirical analysis.

\section{Simulation Results} \label{simresults}

In this section we evaluate the behaviour of the Gauss-Hermite (GH) online quantile estimation algorithm presented in section \ref{GHOnlineStat} and the Exponentially Weighted Gauss Hermite (EWGH) algorithm presented in section \ref{EWGH} on simulated data. We also compare the performance of these algorithms to a leading existing algorithm for online quantile estimation, namely Exponentially Weighted Stochastic Approximation (EWSA), which has been shown to be competitive with a number of other algorithms for online quantile estimation \cite{ewsa}. The EWSA algorithm is an exponentially weighted version of the stochastic approximation algorithm of \cite{sa2}. EWSA has two parameters, one that controls the size of the batches of data used to update the quantile estimates (denoted $M$) and a weighting factor $w$ that controls the weighting of updates to the density and quantile estimates in the stochastic approximation scheme (see \cite{ewsa} for a detailed description of the algorithm). \\

In our investigations both i.i.d and non-identically distributed simulated data are considered. In particular, i.i.d data from a chi-squared distribution with five degrees of freedom, $\chi_{5}$, and an exponential distribution with mean and variance equal to one are considered. In the i.i.d. case, the data have static quantiles. \\

In the non-identically distributed setting we consider simulated data drawn from distributions with non-stationary parameters. In particular we consider simulated data drawn from a normal distribution with variance one and a mean of $.006j$ at update $j$ to simulate data with a linear trend in the mean. We then consider an exponential distribution with mean and variance equal to $1+0.006j$ at update $j$ to simulate data with a non-stationary mean and standard deviation. This is an interesting test in that the dispersion of the data increases as the number of observations grows. These non-identically distributed simulated data have dynamic quantiles that change over time. Both of these models were studied in the simulations of \cite{ewsa}. \\

The choice of these test distributions can be motivated by the diversity of their properties and the frequent appearance of these distributions in statistical applications. For each distribution, three quantiles are estimated, namely the 0.5 (median), 0.9 and 0.99 quantiles following \cite{ewsa}. Note that in the case of the Gauss-Hermite based algorithms arbitrary quantiles can be obtained at any point in time whereas algorithms such as the EWSA algorithm require the quantiles to be specified upfront. Online, arbitrary quantiles are not available for algorithms such as EWSA. The maximum number of observations, $m=4000$ per run in the i.i.d. case and $m=1000$ in the non-identically distributed case (corresponding to the maximum number of \textit{updates} in \cite{ewsa}). There were 1000 runs in total for each distribution. We utilise the empirical root mean squared error (RMSE) to evaluate the performance of the online quantile estimation algorithms in estimating the quantile $q$ after $j$ observations (where we denote the quantile estimate $\hat{q}_j$). The RMSE at updating step $j$ is defined by:

\begin{equation}
		\mbox{RMSE}(\hat{q}_j) = \left[ E(\hat{q}_j - q_j)^{2}\right]^{\frac{1}{2}},
\end{equation}

and is estimated by averaging the squared difference between $q_j$ and $\hat{q}_j$ over the 1000 simulation runs and then taking the square root. As a measure of the error in the RMSE estimate at updating step $j$, we construct a $95 \%$ percentile bootstrap confidence interval:

$$(\mbox{RMSE}^{*}_{(0.025)}, \, \mbox{RMSE}^{*}_{(0.975)}),$$

where $\mbox{RMSE}^{*}_{(0.025)}$  and $\mbox{RMSE}^{*}_{(0.975)}$ denote the $0.025$ and $0.975$ quantiles of the bootstrap estimates of the RMSE respectively ($1000$ bootstrap estimates were utilised in constructing the intervals for each $j$).\\

In \cite{ewsa} a range of values for the EWSA algorithm parameters $M$ and $w$ were investigated for performance. It was demonstrated that $w = 0.05$ gives the best trade-off between bias and long-run variability. In addition, the value of $M=15$ was shown to be superior in long-run performance to other values of $M$ tested. Since we evaluate the same i.i.d. and non-independently distributed data stream models as [11] in our simulations (except for the addition of the chi-squared i.i.d. model which is qualitatively similar), we regard these parameters as principled choices for good performance of the EWSA algorithm in these settings.\\

In our simulation studies for the GH and EWGH algorithms we demonstrate the effectiveness of the algorithms over a range of values for $N$ and $\lambda$ and identify particularly good choices. Concretely for the GH algorithm we consider $N=4,6,8,10,12$ at $m=100,400,4000$ observations to evaluate the bias-variance trade-off at different numbers of observations. We also present RMSE curves where we compare the results to the EWSA algorithm. For the EWGH algorithm we consider $\lambda=0.01, 0.05, 0.1$ ($N=6$) at $m=100,400,1000$ observations and present RMSE curves comparing to the GH ($N=6)$ and EWSA algorithms.\\

Finally, in order to practically apply the Gauss-Hermite based algorithms more effectively, an online standardisation procedure was applied to the data as outlined in appendix \ref{standardization}. This is not a pre-processing step (this would defeat the purpose of an online algorithm) but rather part of the online algorithm. Also, this procedure may not always be necessary, depending on the application. It is also noteworthy that we utilise the alternate CDF estimator (\ref{cumulativedistalt}) in estimating quantiles for the GH and EWGH algorithms as we have found this estimator to yield better results empirically.

\subsection{IID Data}

For the i.i.d. simulated data, we evaluate the performance of the GH algorithm for various choices of $N$ and we use the EWGH algorithm ($\lambda=0.05$) and the EWSA algorithm ($M=15$, $w=0.05$) for comparison. The GH algorithm performs well for most choices of $N$, illustrating that the effectiveness of the algorithm is not critically dependent on this choice. That said, the value of $N=6$ appears to be the best choice in most cases when viewed from a RMSE perspective. In addition, $N=6$ is a good choice from a computational speed and efficiency viewpoint since there are fewer coefficients to store and update than for higher choices of $N$.  When compared to the EWSA algorithm, the GH algorithm performs better in most cases. The EWGH algorithm (with $\lambda=0.05$) has a larger error than the GH algorithm in all cases. This is not entirely surprising however. The EWGH algorithm trades extra variance in the estimates of the coefficients for the ability to track dynamic quantiles. The individual tests are discussed below.

\subsubsection{The Chi-Squared Distribution}
The chi-squared distribution appears frequently in statistics and is a useful test distribution in that it has support on the half real line $[0,\infty)$ and not the full real line. This distribution is used to simulate skewed data. This is a challenging test case for the Gauss-Hermite based algorithms since the chi-squared distribution is a considerable departure from the normal distribution which undergirds the Gauss-Hermite expansion. We consider the chi-squared distribution with five degrees of freedom in particular. See figure \ref{chi2figsBiasVar} for plots of the GH RMSE for $N=4,6,8,10,12$ at $m=100,400,4000$ observations. The results for the EWGH and EWSA algorithms are also included for comparison. These figures illustrate that good results are achieved for all values of $N$ considered for the GH algorithm. $N=6$ appears to provide the best results. It is interesting to note that upon first inspection, it is counter-intuitive that the RMSE increases at higher values of $N$ even when the number of observations is large. We suspect that this is due to the additional bias introduced by the online standardisation procedure as well as by using the CDF estimator (\ref{cumulativedistalt}) instead of  (\ref{cumulativedist}). See figure \ref{chi2figs} for a comparison of the GH algorithm ($N=6$) and the EWSA algorithm. Note that the EWSA results were excluded from figure \ref{chi2figsBiasVar}\subref{chi2NinetyNine4000} and figure \ref{chi2figs}\subref{chi2NinetyNineCurve} since they were disproportionately large and would obscure the GH results when presented on a common scale. This may indicate instability in the EWSA algorithm for estimating tail quantiles such as $p=0.99$.

\subsubsection{The Exponential Distribution}
The exponential distribution is another commonly occurring distribution. The distribution also has support on the half real line $[0,\infty)$ and is used to simulate skewed data. The exponential distribution is an even more challenging test case for the Gauss-Hermite based algorithms than the chi-squared distribution. This is due to the fact that the exponential distribution's mode occurs at the start of its domain which is to be contrasted with the mode of the normal distribution which is equal to its median. We consider the exponential distribution with mean and variance equal to one. See figure \ref{expfigsBiasVar} for plots of the GH RMSE for $N=4,6,8,10,12$ at $m=100,400,4000$ observations. The results for the EWGH and EWSA algorithms are also included for comparison. These figures again illustrate that good results are achieved for all values of $N$ considered and that the value of $N=6$ appears to provide the best results. See figure \ref{expfigs} for a comparison of the GH algorithm ($N=6$) and the EWSA algorithm. Note that the EWSA results were excluded from figure \ref{expfigsBiasVar}\subref{expNinetyNine4000} and figure \ref{expfigs}\subref{expNinetyNineCurve} since they were again disproportionately large and would obscure the GH results when presented on a common scale.

\subsection{Non-identically Distributed Data}

For the non-identically distributed simulated data, we evaluate the performance of the EWGH algorithm for various choices of $\lambda$ and we use the GH algorithm ($N=6$) and the EWSA algorithm ($M=15$, $w=0.05$) for comparison. Motivated by the analysis of the GH algorithm in the i.i.d. setting we set $N=6$ for all tests of the EWGH algorithm. The EWGH algorithm compares favourably with the EWSA algorithm for all values of $\lambda$, illustrating that the effectiveness of the algorithm is not critically dependent on this choice in the models we studied. The GH and EWGH algorithms outperform the EWSA algorithm in almost all cases. The dynamic quantile tracking ability of the EWGH algorithm is apparent in that it achieves better results than the GH algorithm when using an appropriate value of $\lambda$. The individual tests are discussed below.

\subsubsection{Normal Distribution with Drift}
The normal distribution is ubiquitous. We consider simulated data drawn from a normal distribution with variance one and a mean of $.006j$ at update $j$ to simulate data with a linear trend in the mean. See figure \ref{normalDriftfigsBiasVar} for plots of the EWGH RMSE for $\lambda=0.01,0.05,0.1$ at $m=100,400,1000$ observations. The results for the GH and EWSA algorithms are also included for comparison. These figures illustrate that competitive results are achieved for all values of $\lambda$ considered for the EWGH algorithm. The value of $\lambda=0.01$ appears to provide the best results. See figure \ref{normalDriftfigs} for the RMSE curve for the EWGH algorithm with $N=6, \, \lambda=0.01$ compared to the GH and EWGH algorithms.

\subsubsection{Exponential Distribution with Drift}
In this section we consider an exponential distribution with mean and variance equal to $1+0.006j$ at update $j$ to simulate data with a non-stationary mean and standard deviation. The dispersion of the data increases as the number of observations grows. See figure \ref{expDriftfigsBiasVar} for plots of the EWGH RMSE for $\lambda=0.01,0.05,0.1$ at $m=100,400,1000$ observations. The results for the GH and EWSA algorithms are also included for comparison. These figures illustrate that competitive results are again achieved for all values of $\lambda$ considered for the EWGH algorithm. The value of $\lambda=0.01$ appears to provide the best results. See figure \ref{expDriftfigs} for the RMSE curve for the EWGH algorithm with $N=6, \, \lambda=0.01$ compared to the GH and EWGH algorithms.

\begin{sidewaysfigure}
\begin{figure}[H]
\subfloat[m=100, p=0.5]{
  \includegraphics[width=62mm]{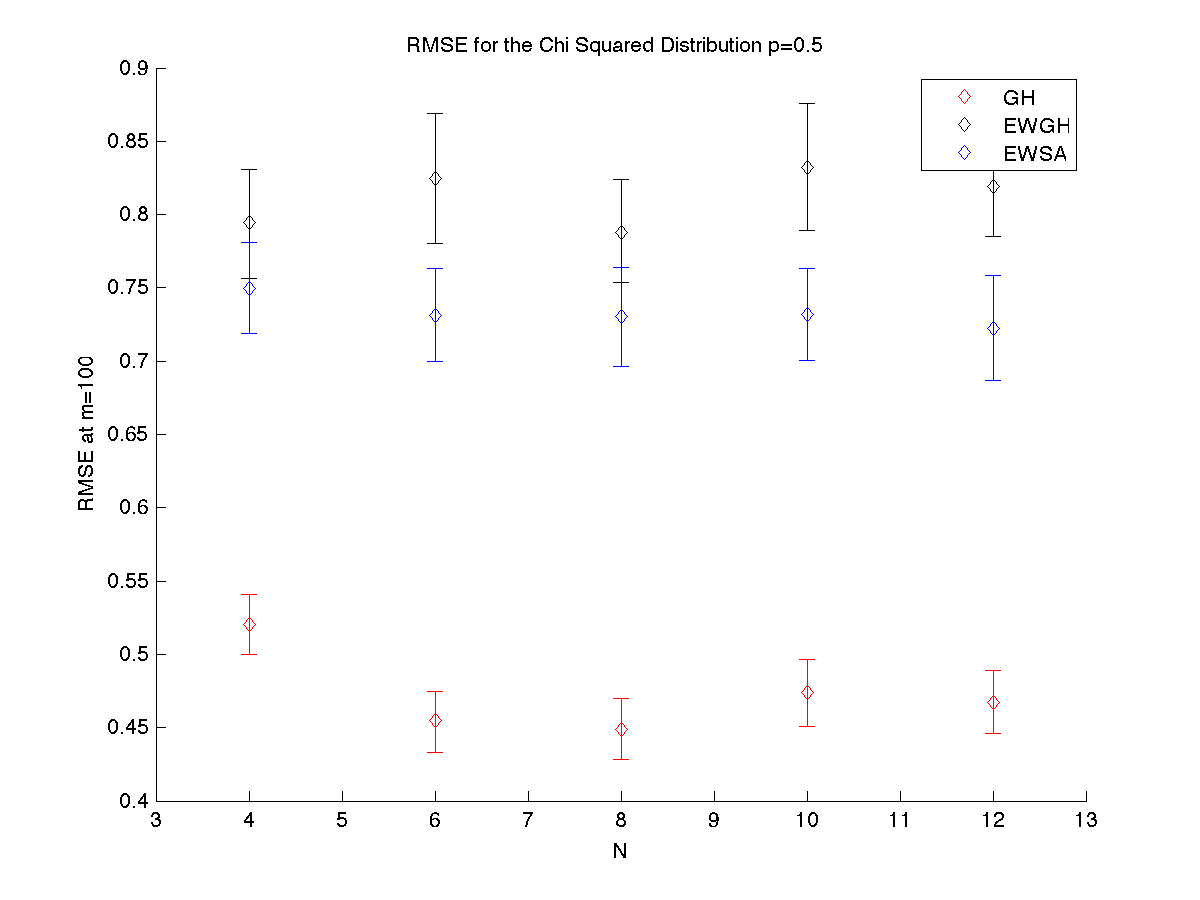}
}
\subfloat[m=400, p=0.5]{
  \includegraphics[width=62mm]{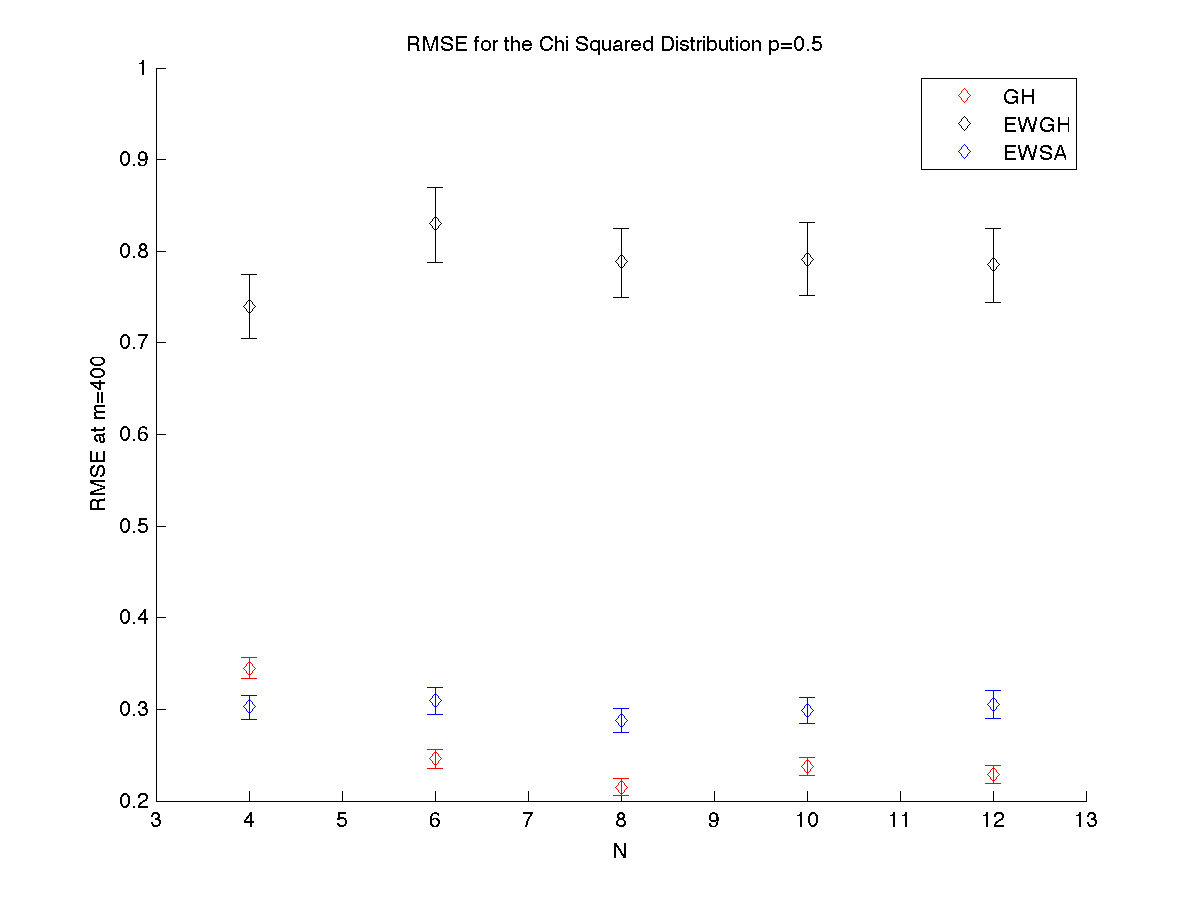}
}
\subfloat[m=4000, p=0.5]{
  \includegraphics[width=62mm]{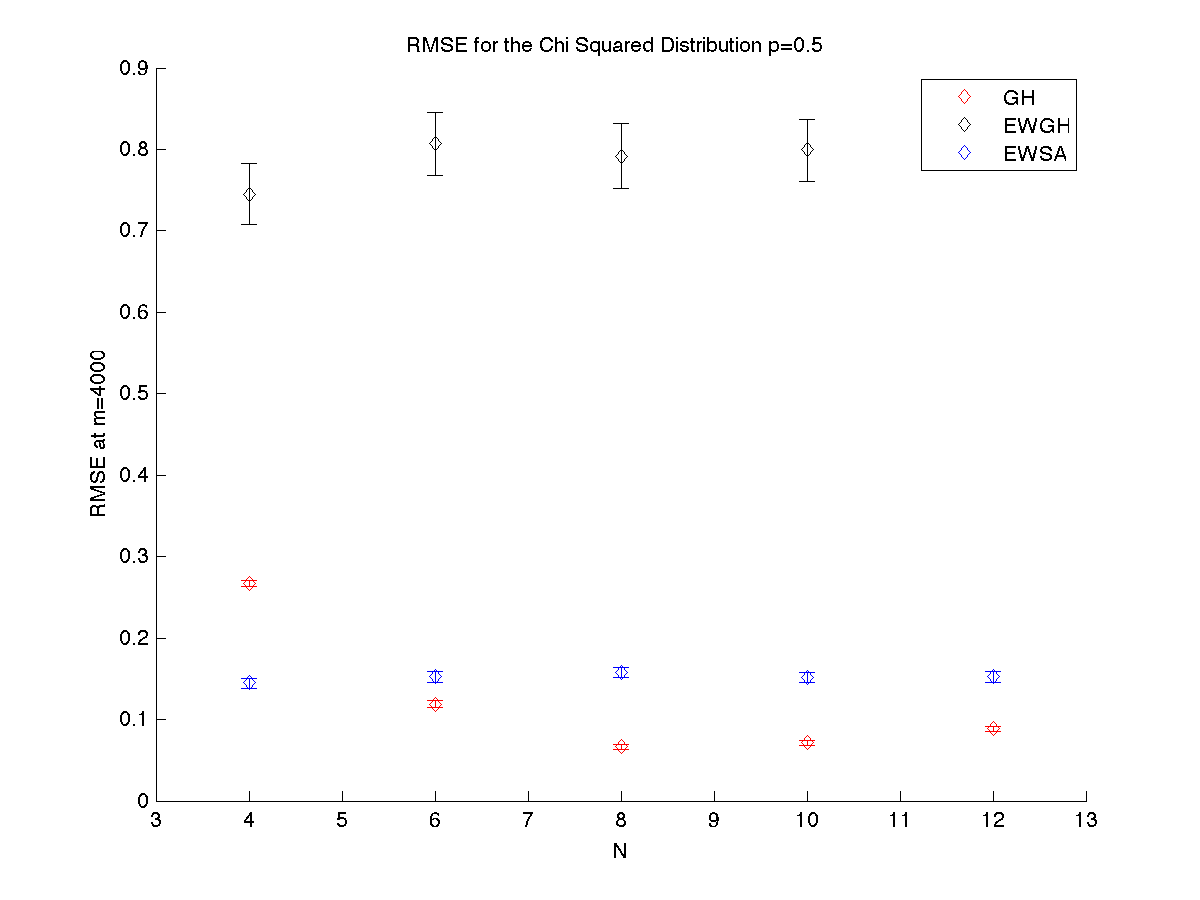}
}
\hspace{0mm}
\subfloat[m=100, p=0.9]{
  \includegraphics[width=62mm]{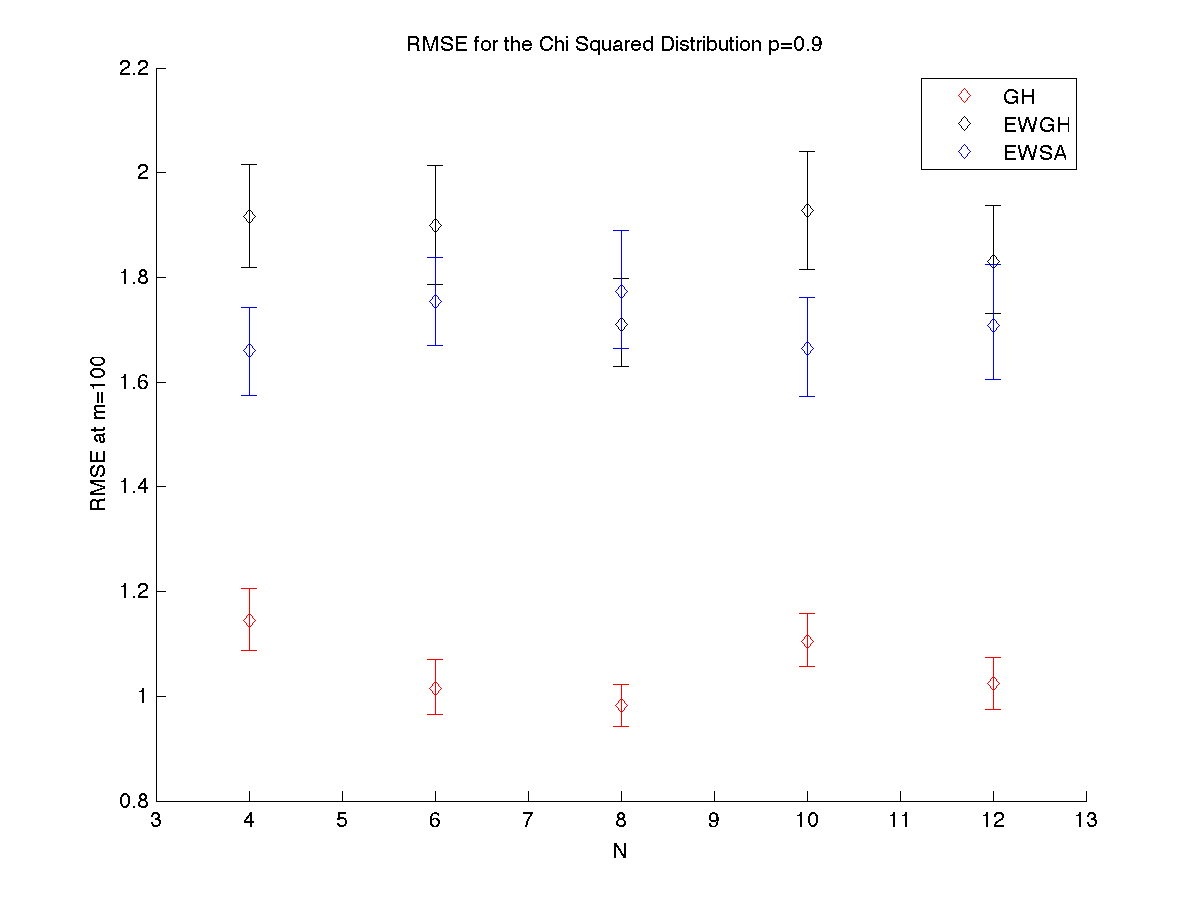}
}
\subfloat[m=400, p=0.9]{
  \includegraphics[width=62mm]{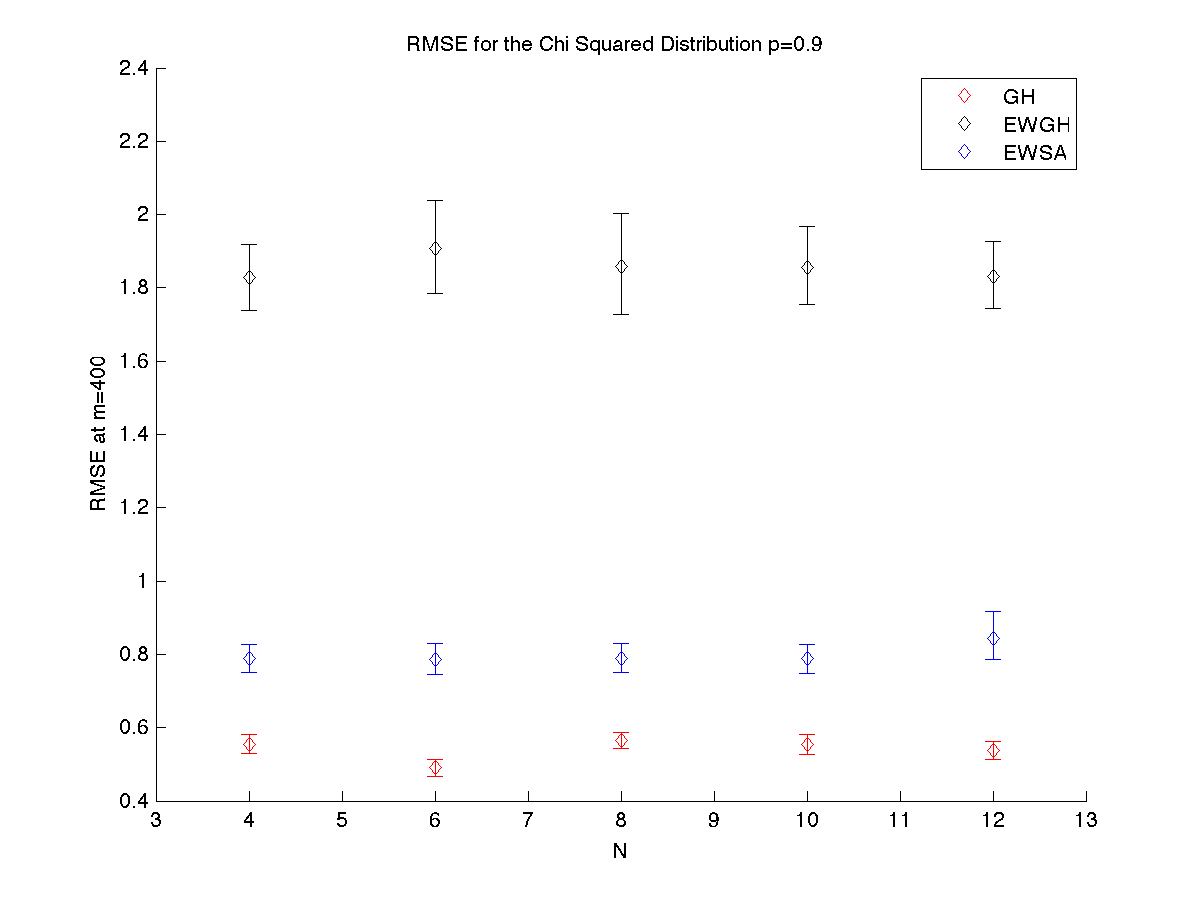}
}
\subfloat[m=4000, p=0.9]{
  \includegraphics[width=62mm]{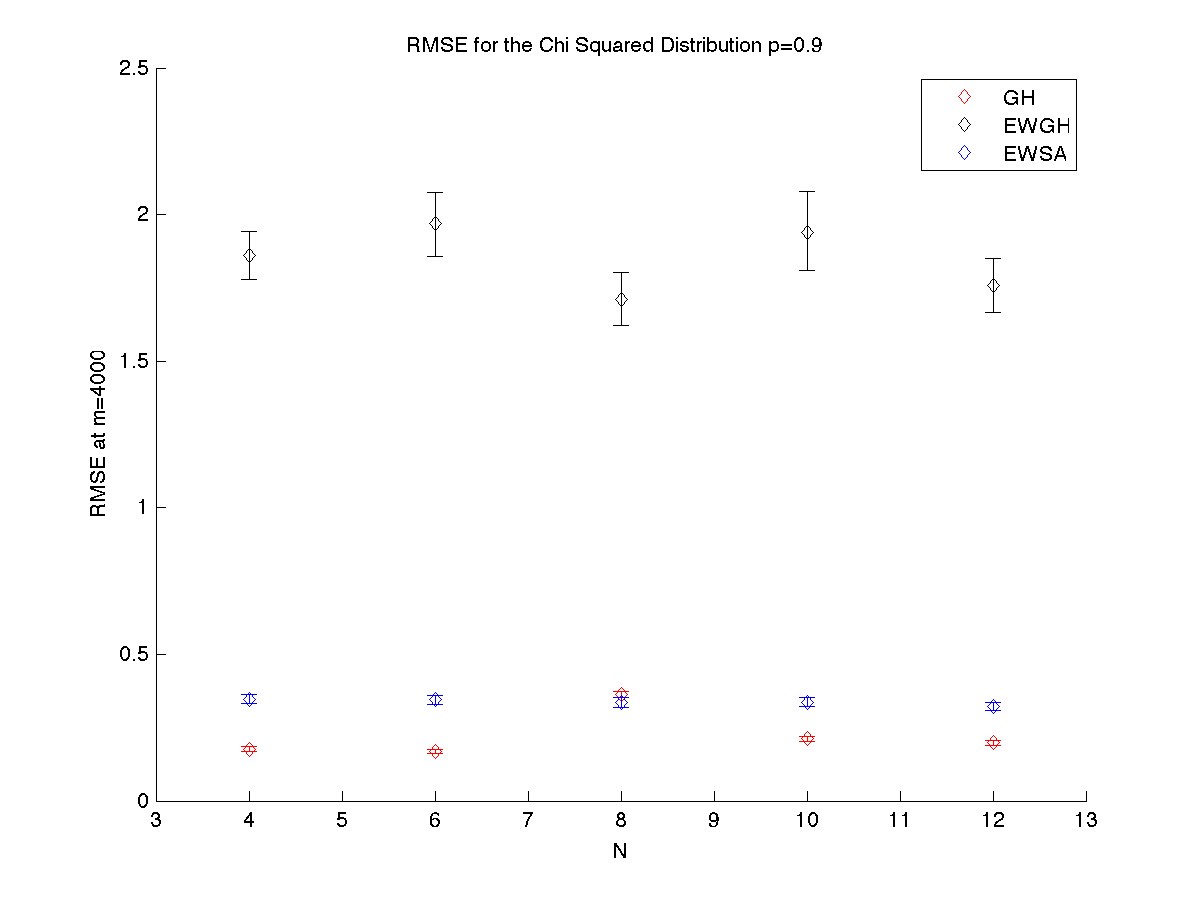}
}
\hspace{0mm}
\subfloat[m=100, p=0.99]{
  \includegraphics[width=62mm]{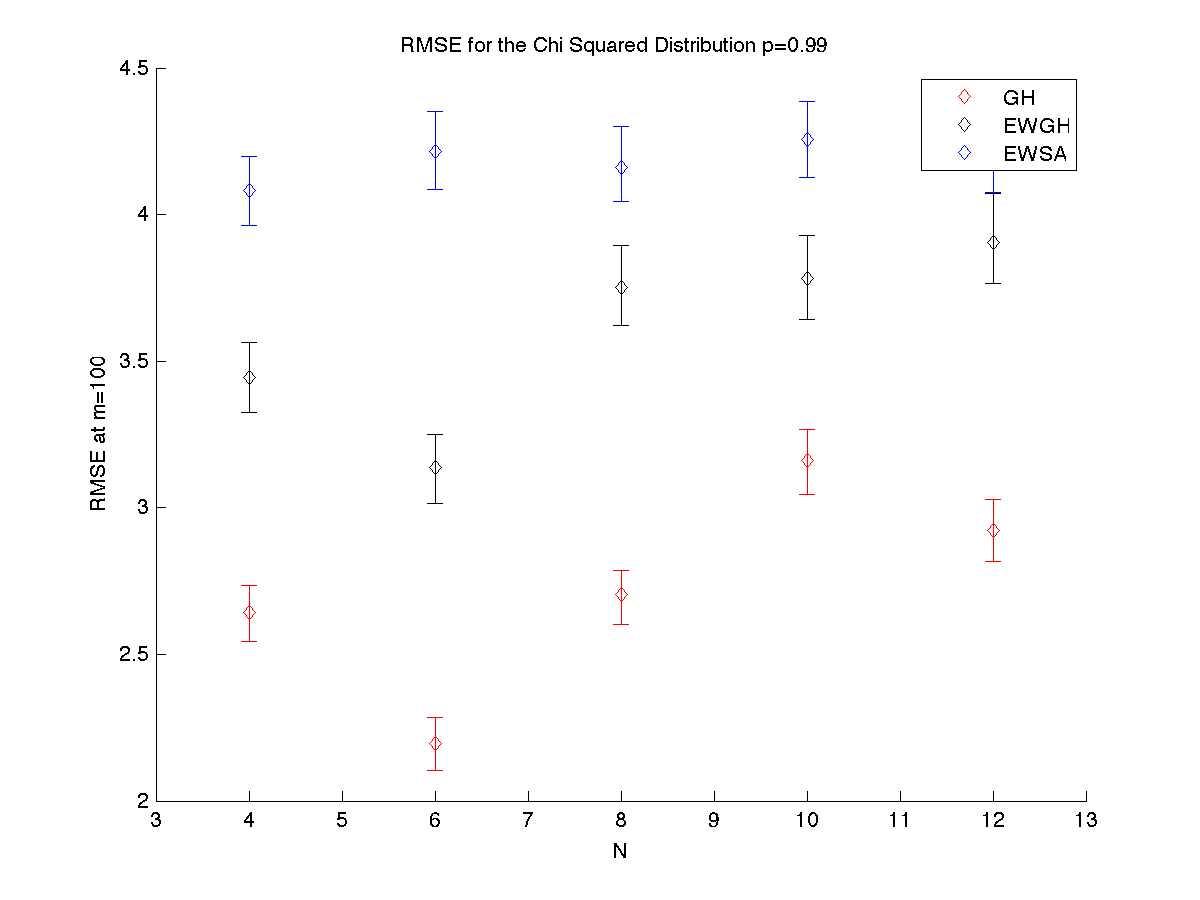}
}
\subfloat[m=400, p=0.99]{
  \includegraphics[width=62mm]{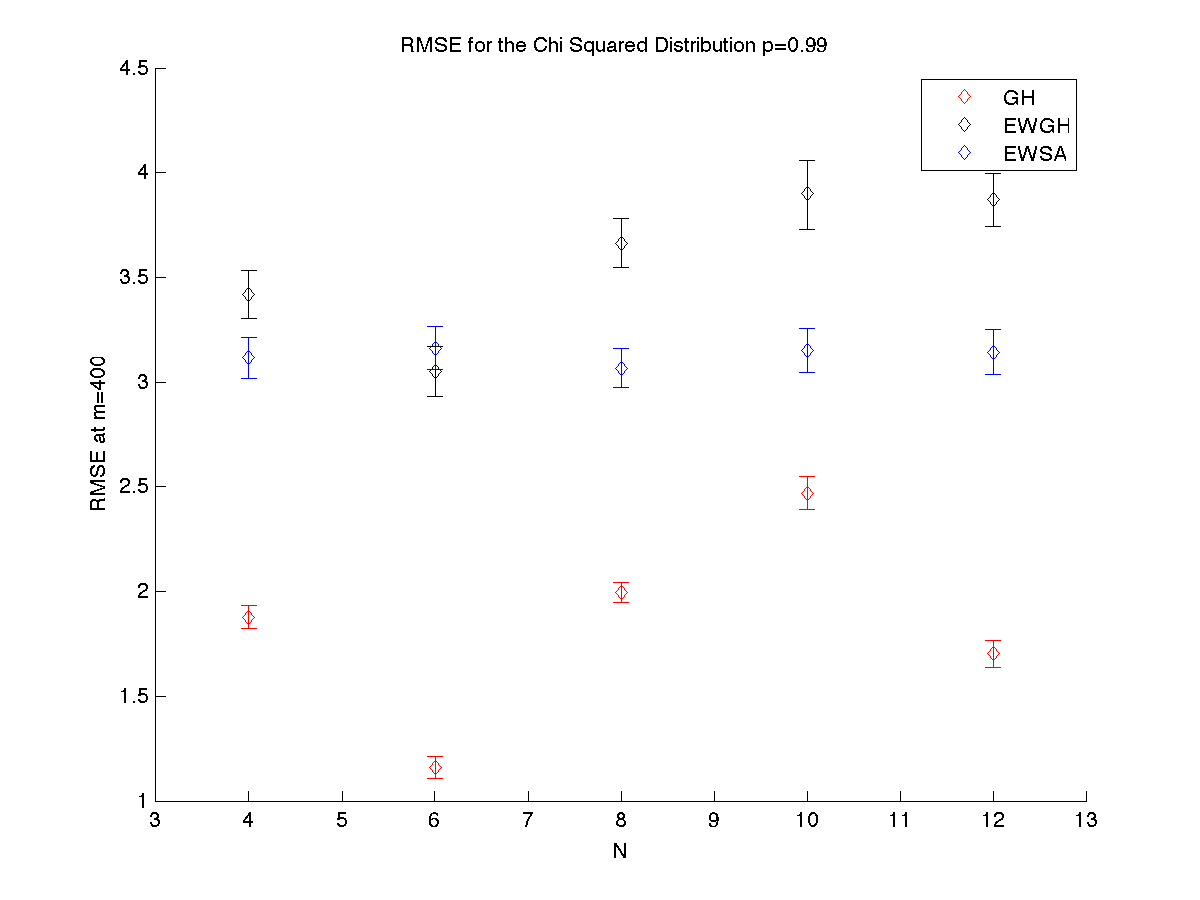}
}
\subfloat[m=4000, p=0.99]{
  \includegraphics[width=62mm]{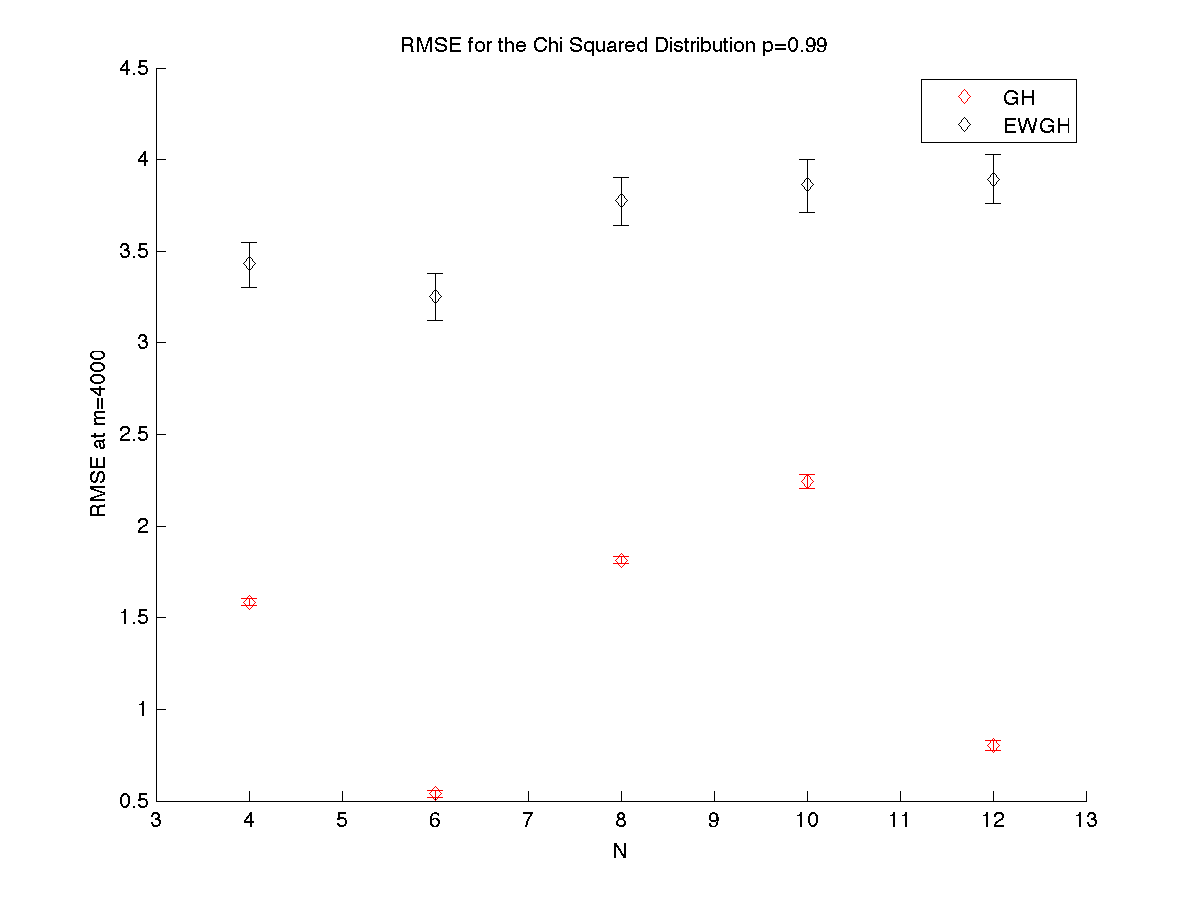}\label{chi2NinetyNine4000}
}
\caption{Chi-Squared Distribution: GH RMSE for $N=4,6,8,10,12$ at $m=100,400,4000$ observations for the $p=0.5, 0.9, 0.99$ quantiles (including 95\% percentile bootstrap confidence intervals). The exact quantiles are 4.3515, 9.2364 and 15.0863 for comparison.\label{chi2figsBiasVar}}
\end{figure}
\end{sidewaysfigure}

\begin{sidewaysfigure}
\begin{figure}[H]
\subfloat[m=100, p=0.5]{
  \includegraphics[width=62mm]{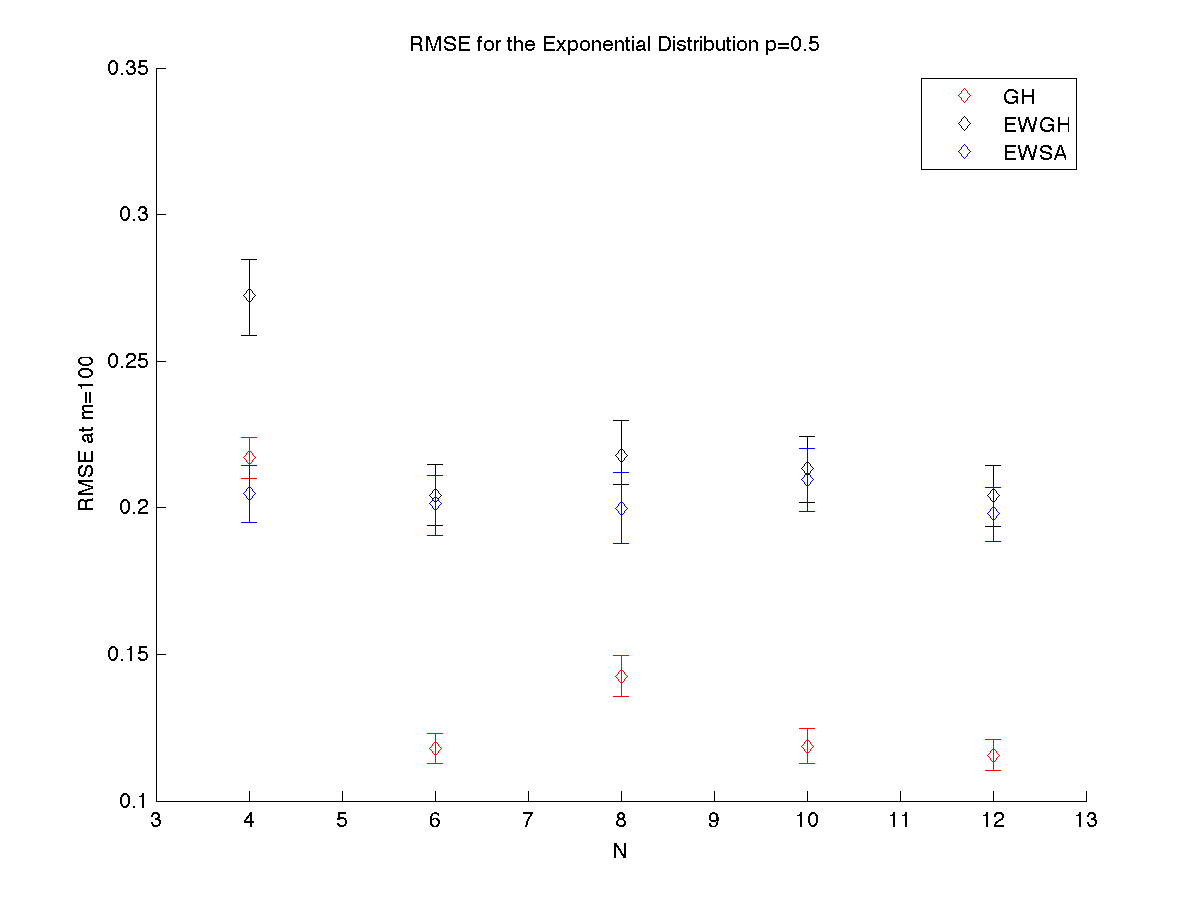}
}
\subfloat[m=400, p=0.5]{
  \includegraphics[width=62mm]{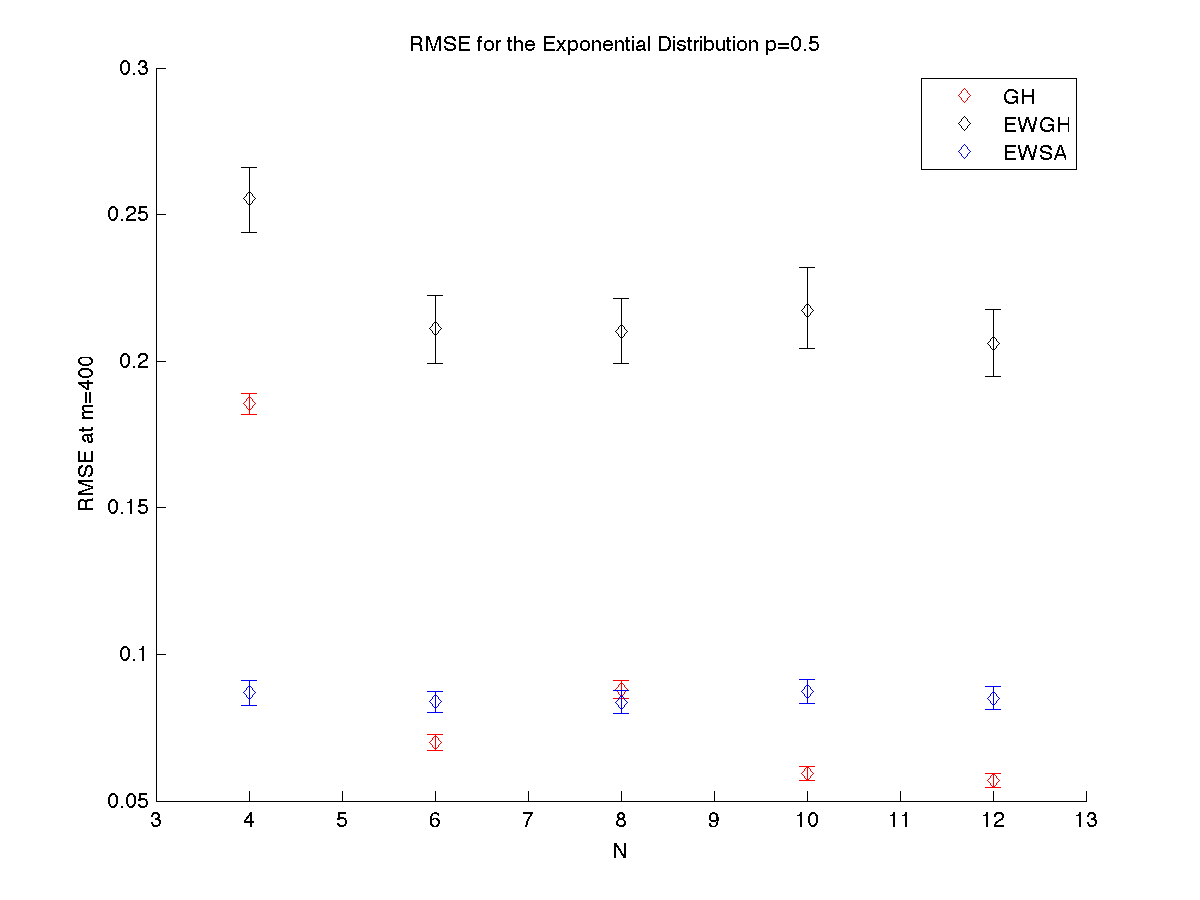}
}
\subfloat[m=4000, p=0.5]{
  \includegraphics[width=62mm]{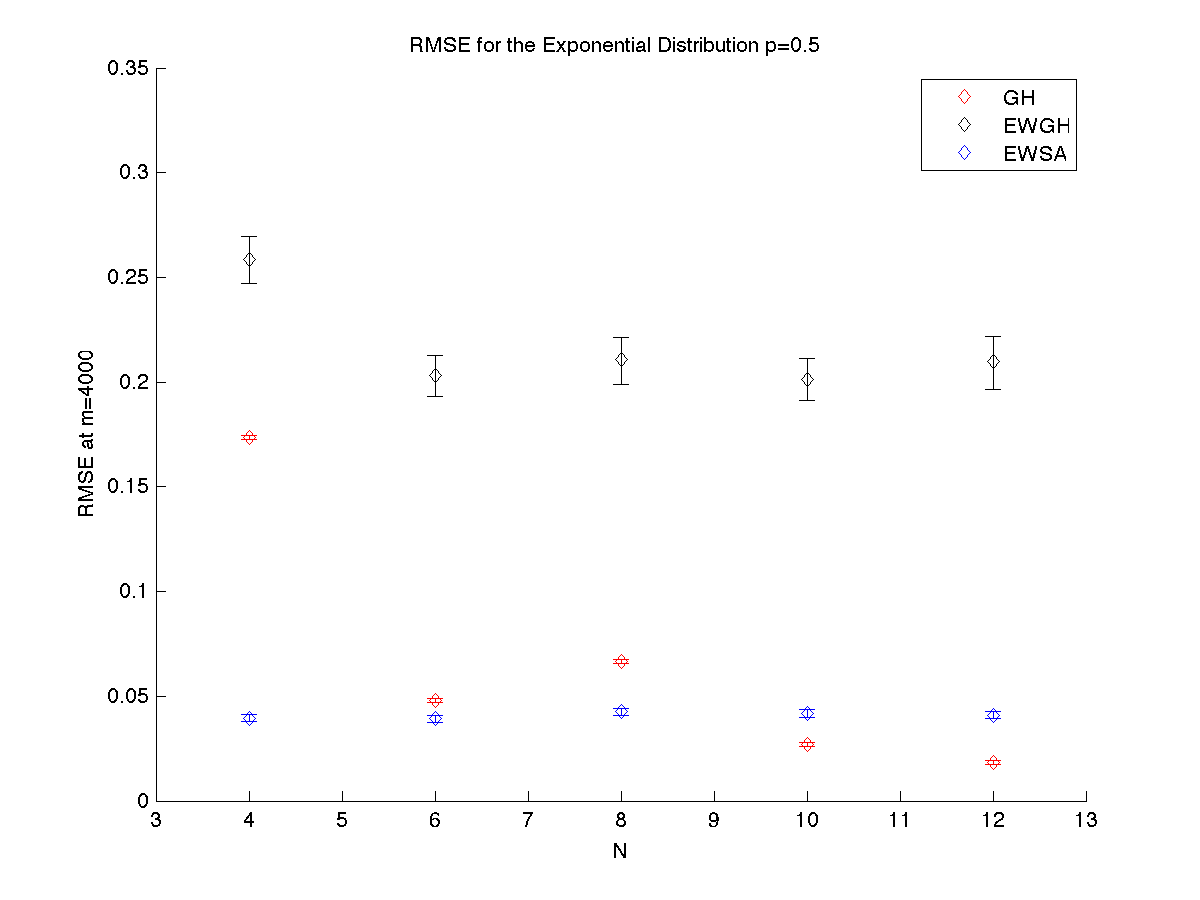}
}
\hspace{0mm}
\subfloat[m=100,p=0.9]{
  \includegraphics[width=62mm]{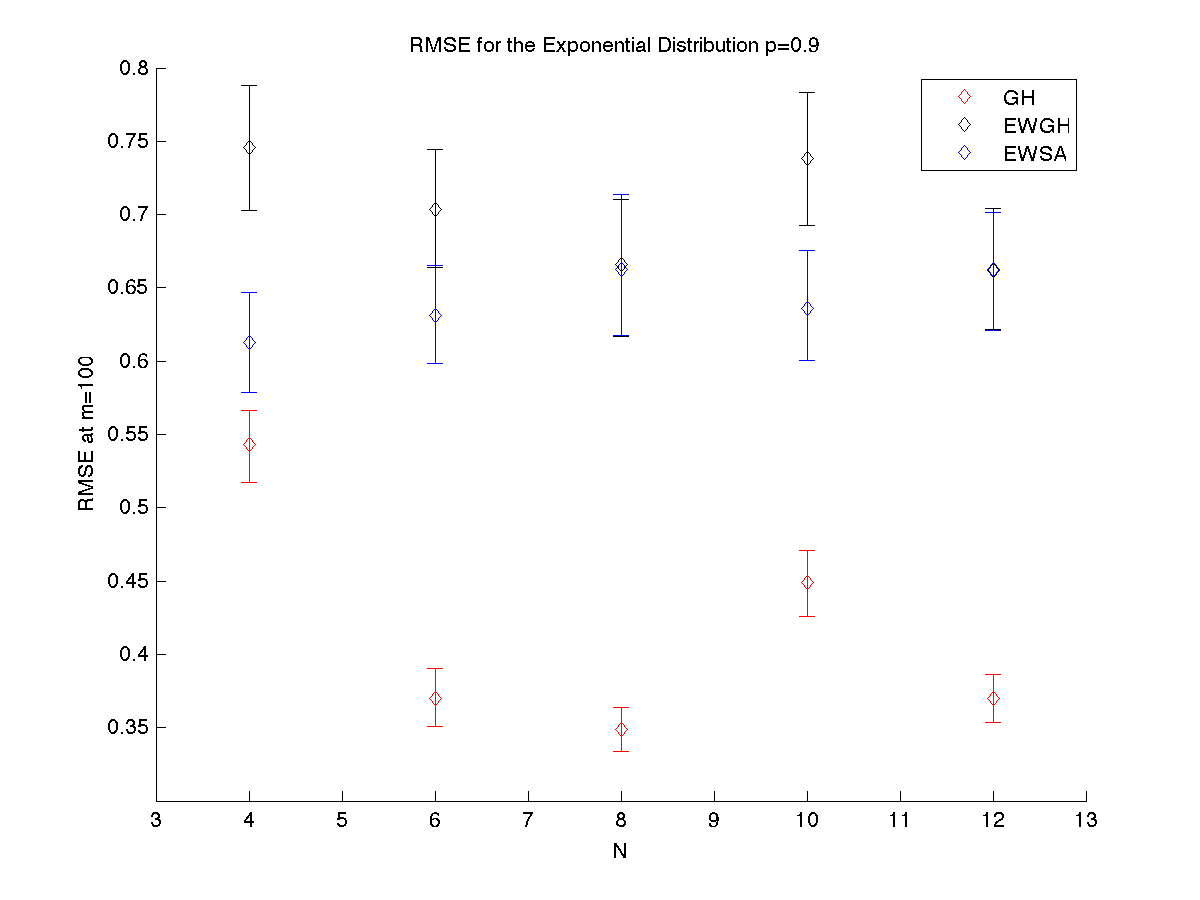}
}
\subfloat[m=400, p=0.9]{
  \includegraphics[width=62mm]{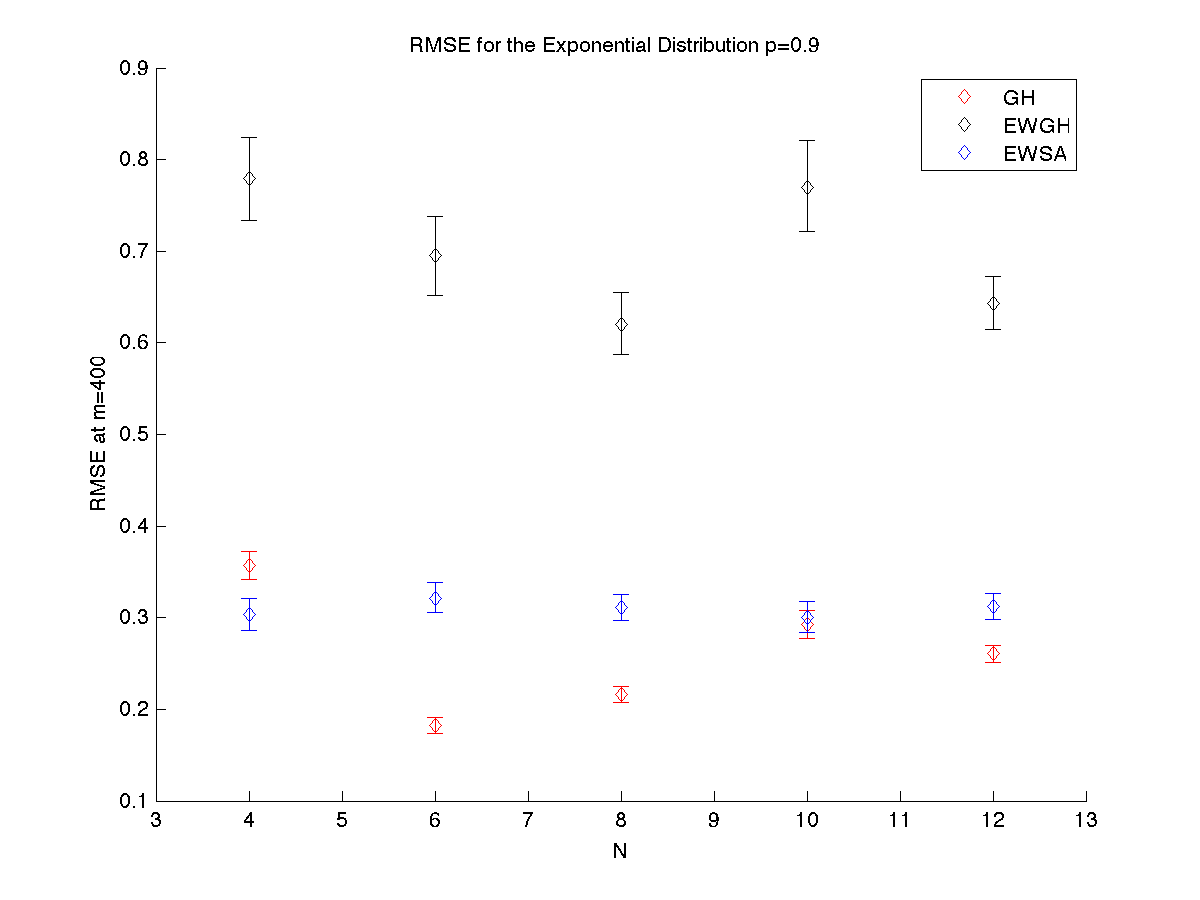}
}
\subfloat[m=4000, p=0.9]{
  \includegraphics[width=62mm]{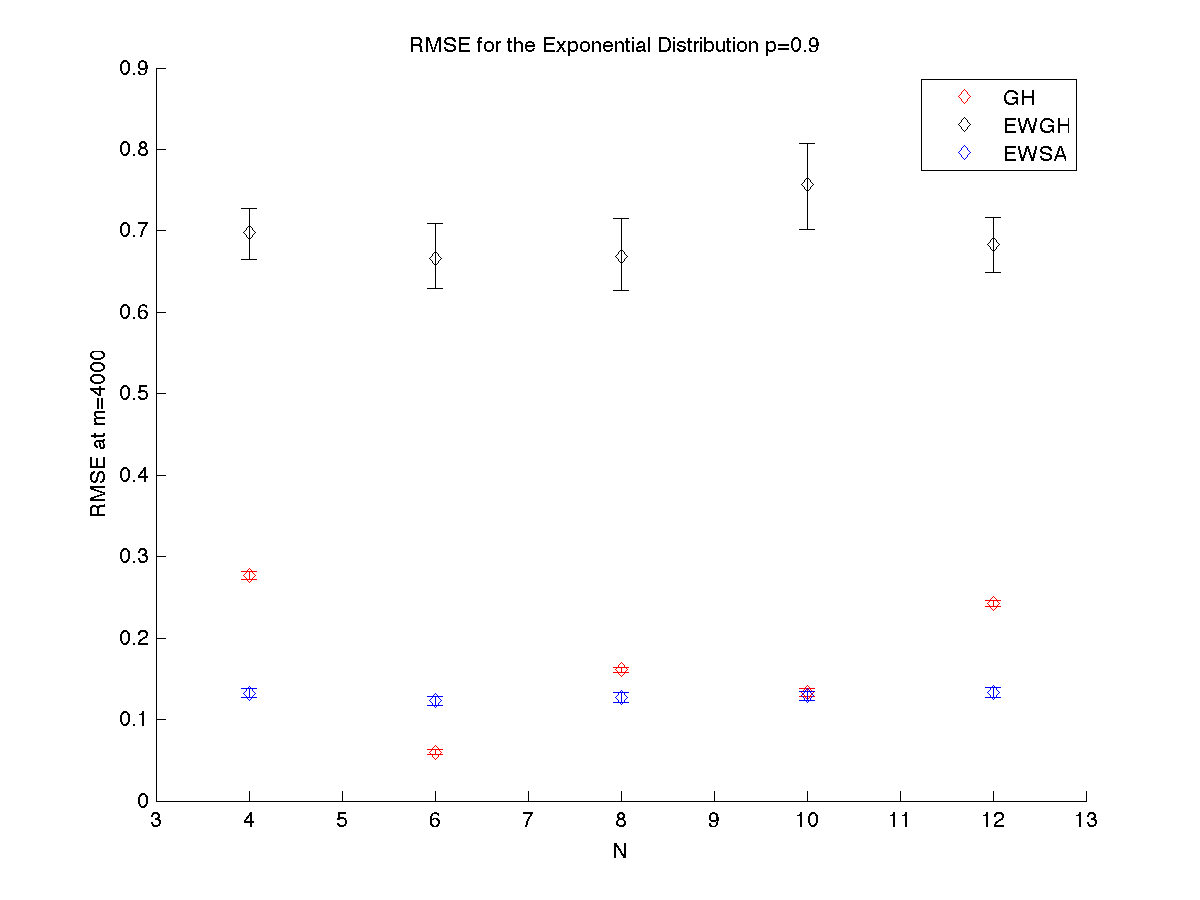}
}
\hspace{0mm}
\subfloat[m=100, p=0.99]{
  \includegraphics[width=62mm]{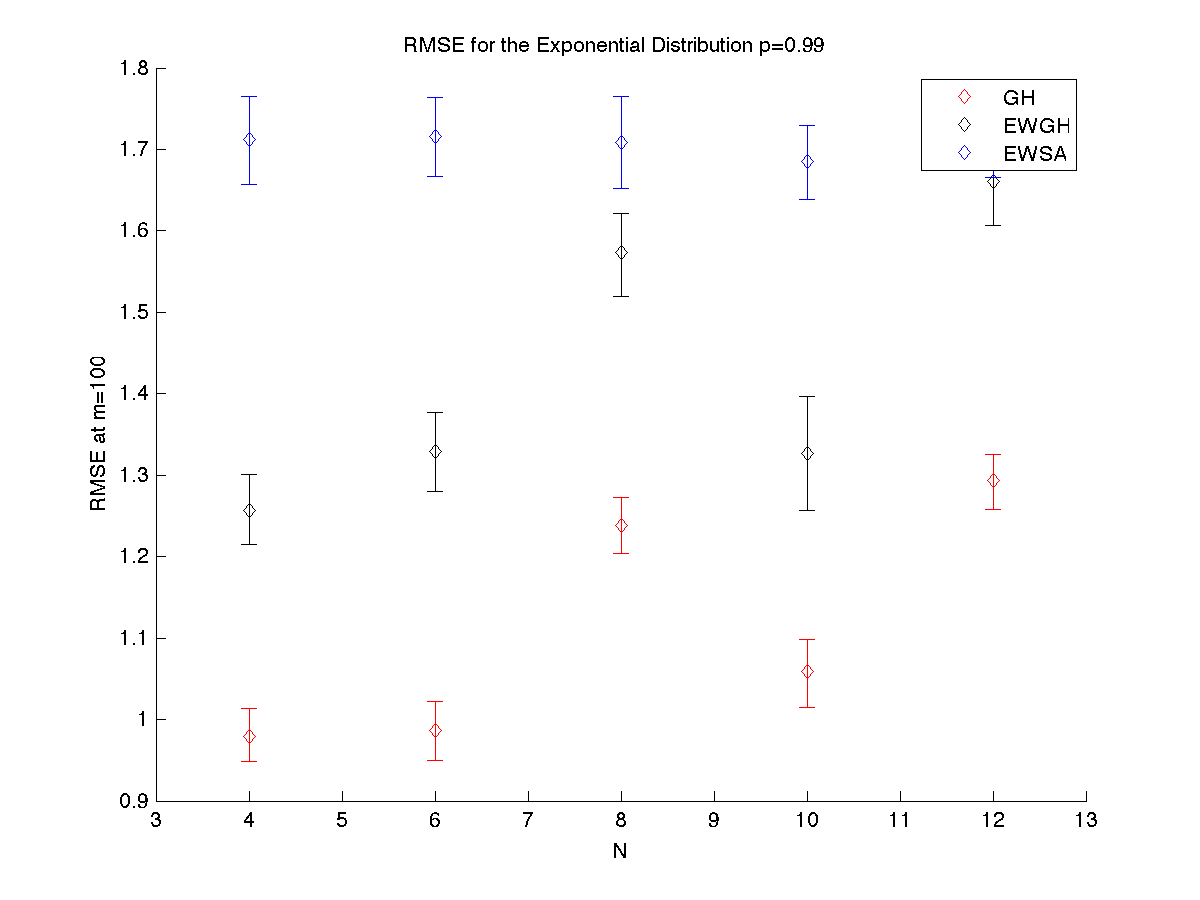}
}
\subfloat[m=400, p=0.99]{
  \includegraphics[width=62mm]{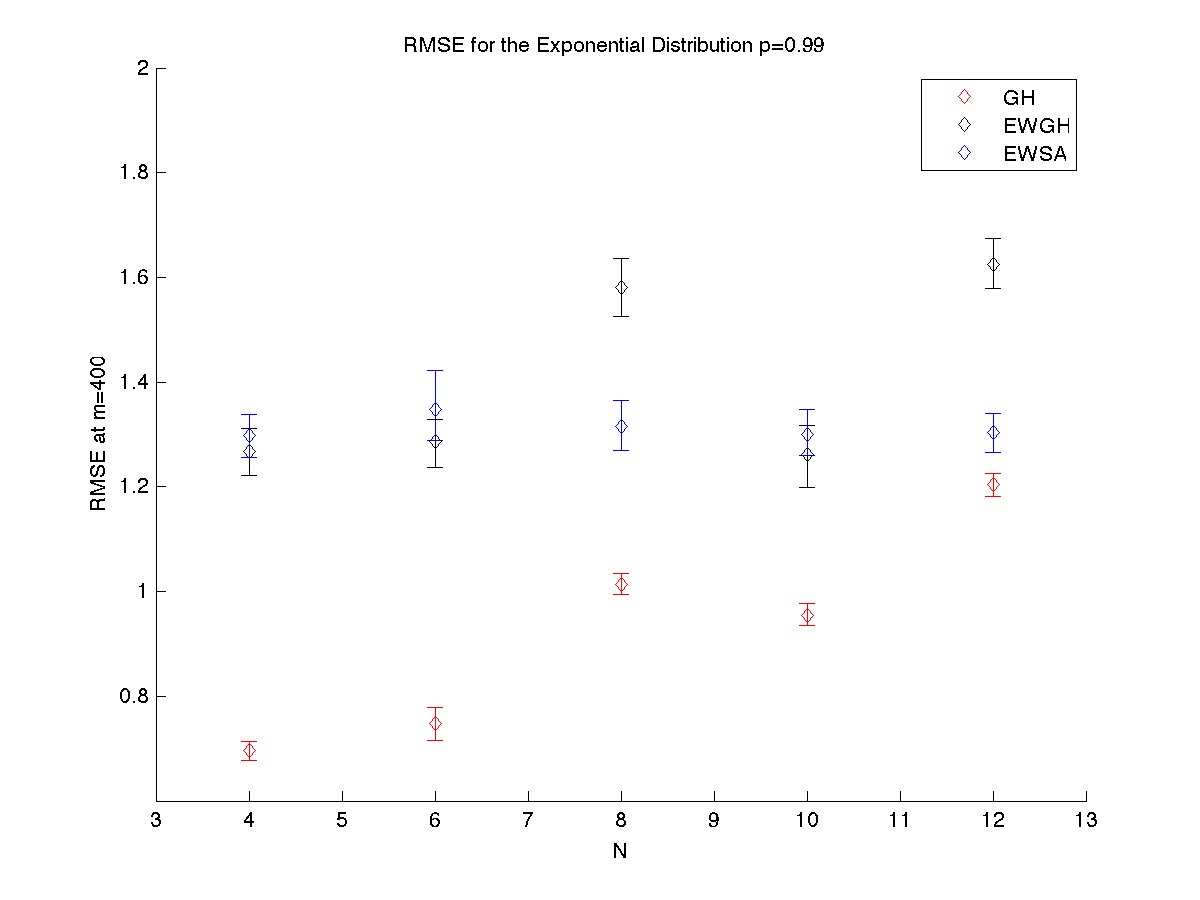}
}
\subfloat[m=4000, p=0.99]{
  \includegraphics[width=62mm]{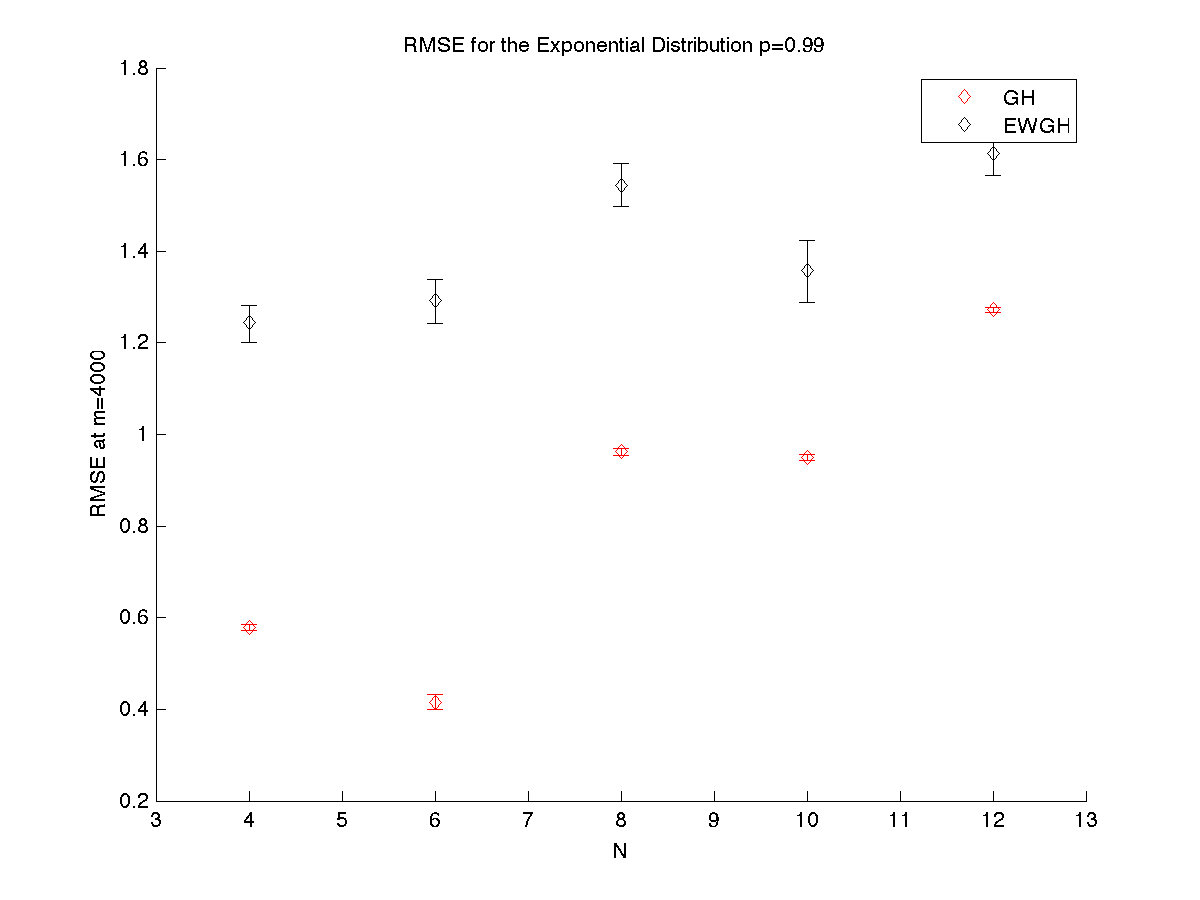} \label{expNinetyNine4000}
}
\caption{Exponential Distribution: GH RMSE for $N=4,6,8,10,12$ at $m=100,400,4000$ observations for the $p=0.5, 0.9, 0.99$ quantiles (including 95\% percentile bootstrap confidence intervals). The exact quantiles are 0.6931, 2.3026 and 4.6052 for comparison.\label{expfigsBiasVar}}
\end{figure}
\end{sidewaysfigure}

\begin{sidewaysfigure}
\begin{figure}[H]
\subfloat[p=0.5]{
  \includegraphics[width=56mm]{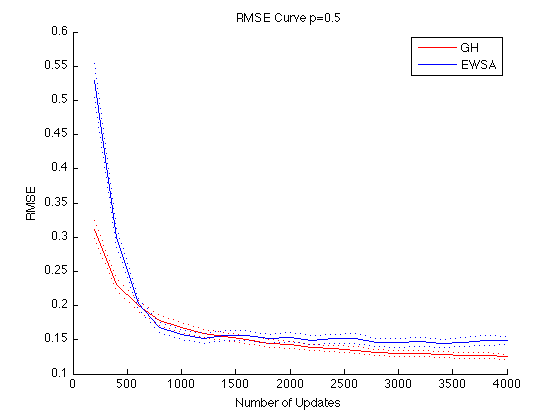}
}
\subfloat[p=0.9]{
  \includegraphics[width=56mm]{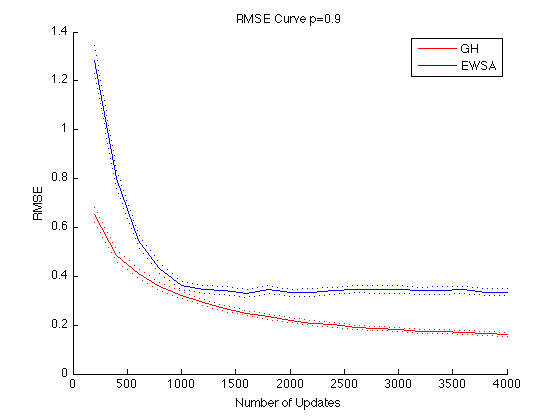}
}
\subfloat[p=0.99]{
  \includegraphics[width=75mm, height=42mm]{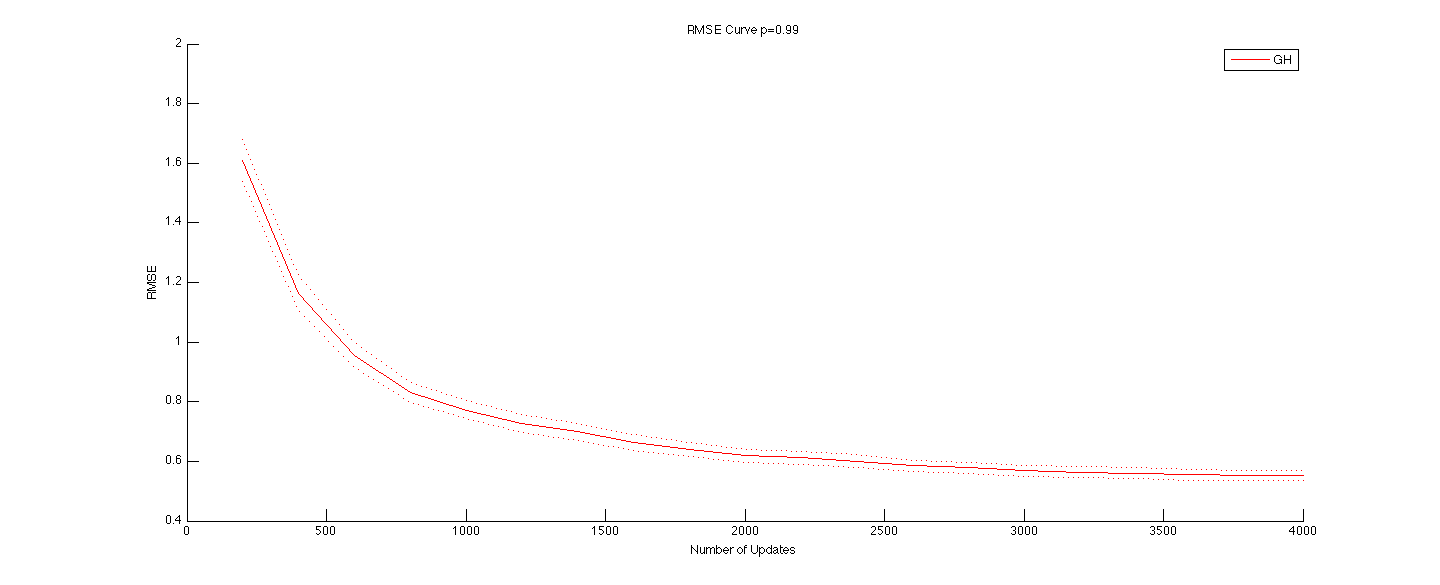}\label{chi2NinetyNineCurve}
}
\caption{RMSE curves associated with the chi-squared distribution with five degrees of freedom for the 0.5, 0.9 and 0.99 quantiles (including 95\% percentile bootstrap confidence intervals). The exact quantiles are 4.3515, 9.2364 and 15.0863 respectively. The Gauss-Hermite algorithm utilises $N=6$.\label{chi2figs}}
\end{figure}

\begin{figure}[H]
\subfloat[p=0.5]{
  \includegraphics[width=56mm]{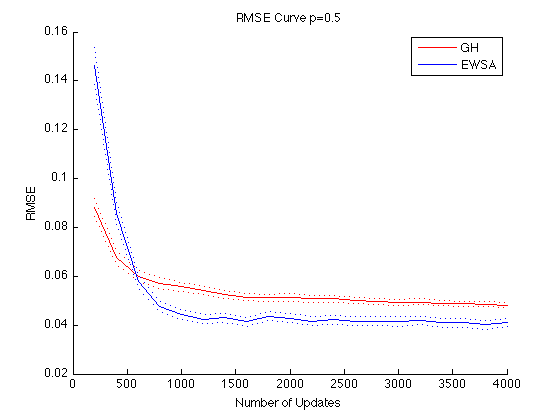}
}
\subfloat[p=0.9]{
  \includegraphics[width=56mm]{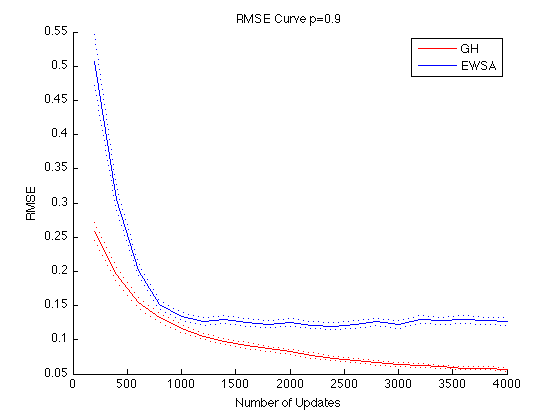}
}
\subfloat[p=0.99]{
  \includegraphics[width=75mm, height=42mm]{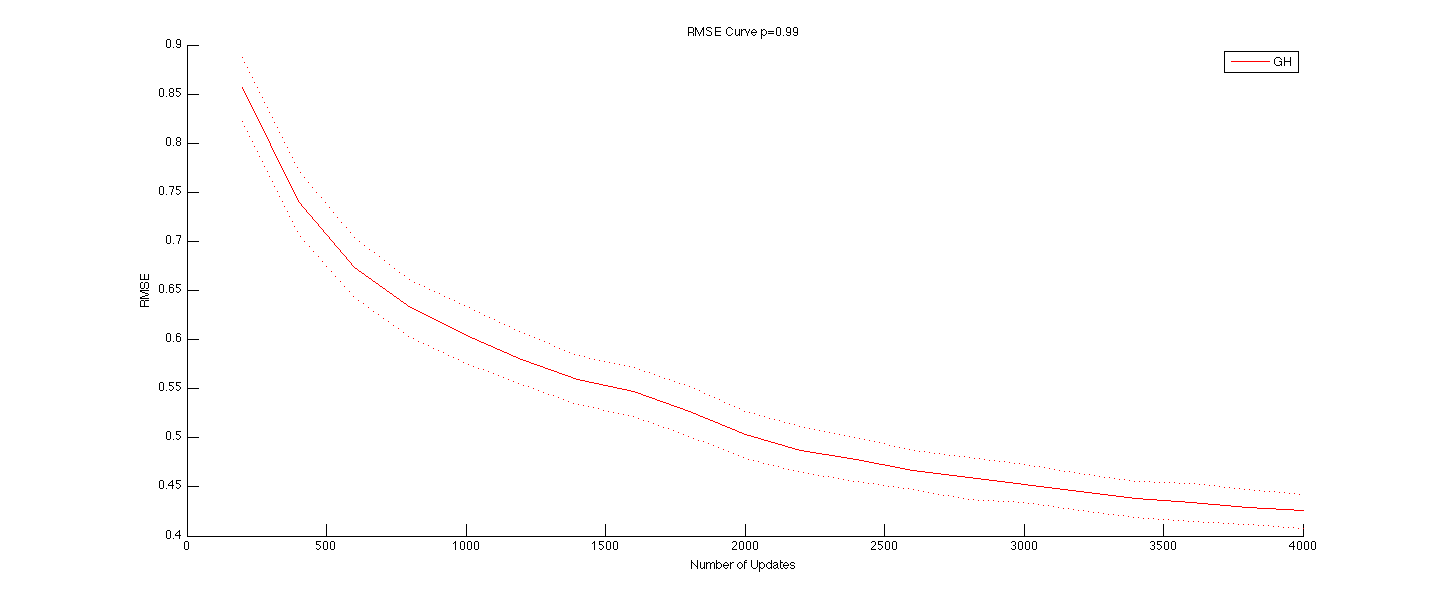} \label{expNinetyNineCurve}
}
\caption{RMSE curves associated with the exponential distribution for the 0.5, 0.9 and 0.99 quantiles (including 95\% percentile bootstrap confidence intervals). The exact quantiles are 0.6931, 2.3026 and 4.6052 respectively. The Gauss-Hermite algorithm utilises $N=6$.\label{expfigs}}
\end{figure}
\end{sidewaysfigure}

\begin{sidewaysfigure}
\begin{figure}[H]
\subfloat[m=100, p=0.5]{
  \includegraphics[width=62mm]{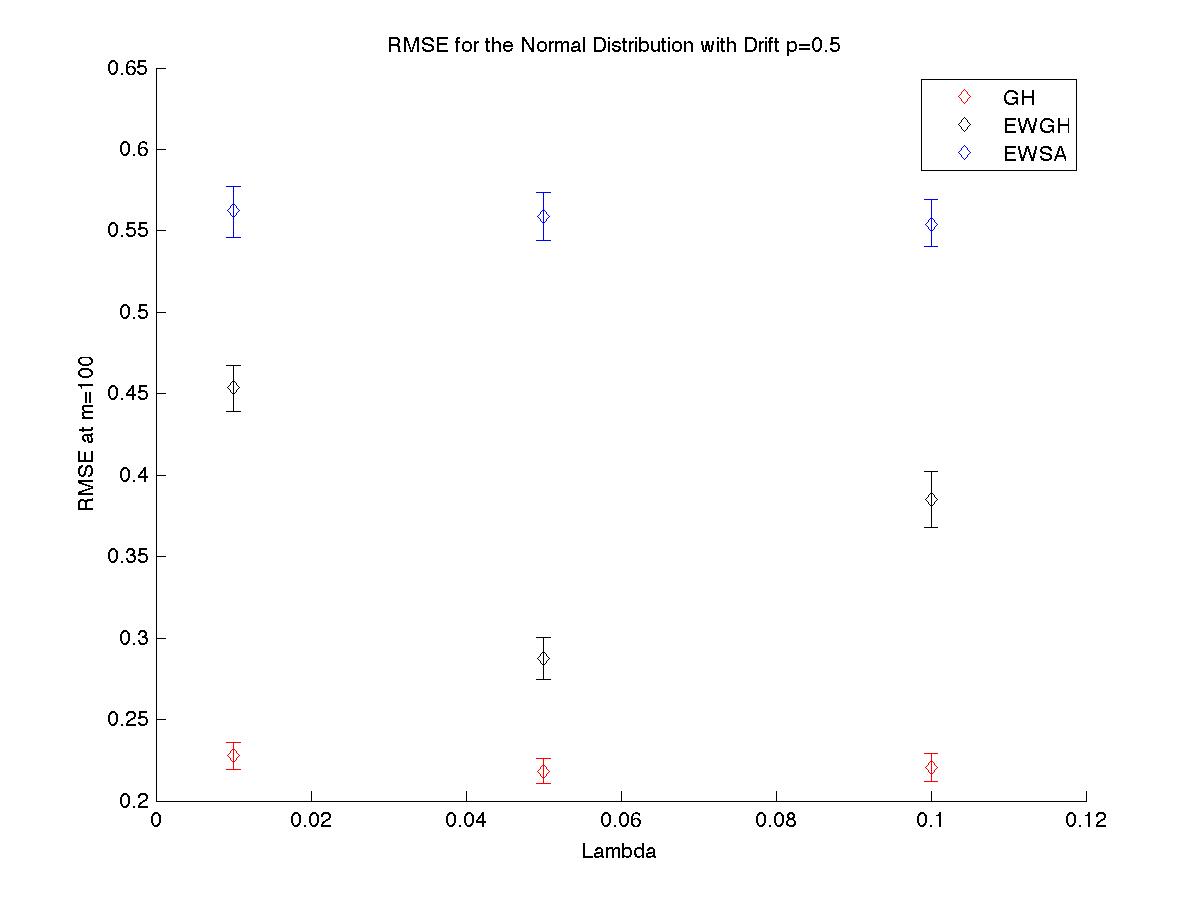}
}
\subfloat[m=400, p=0.5]{
  \includegraphics[width=62mm]{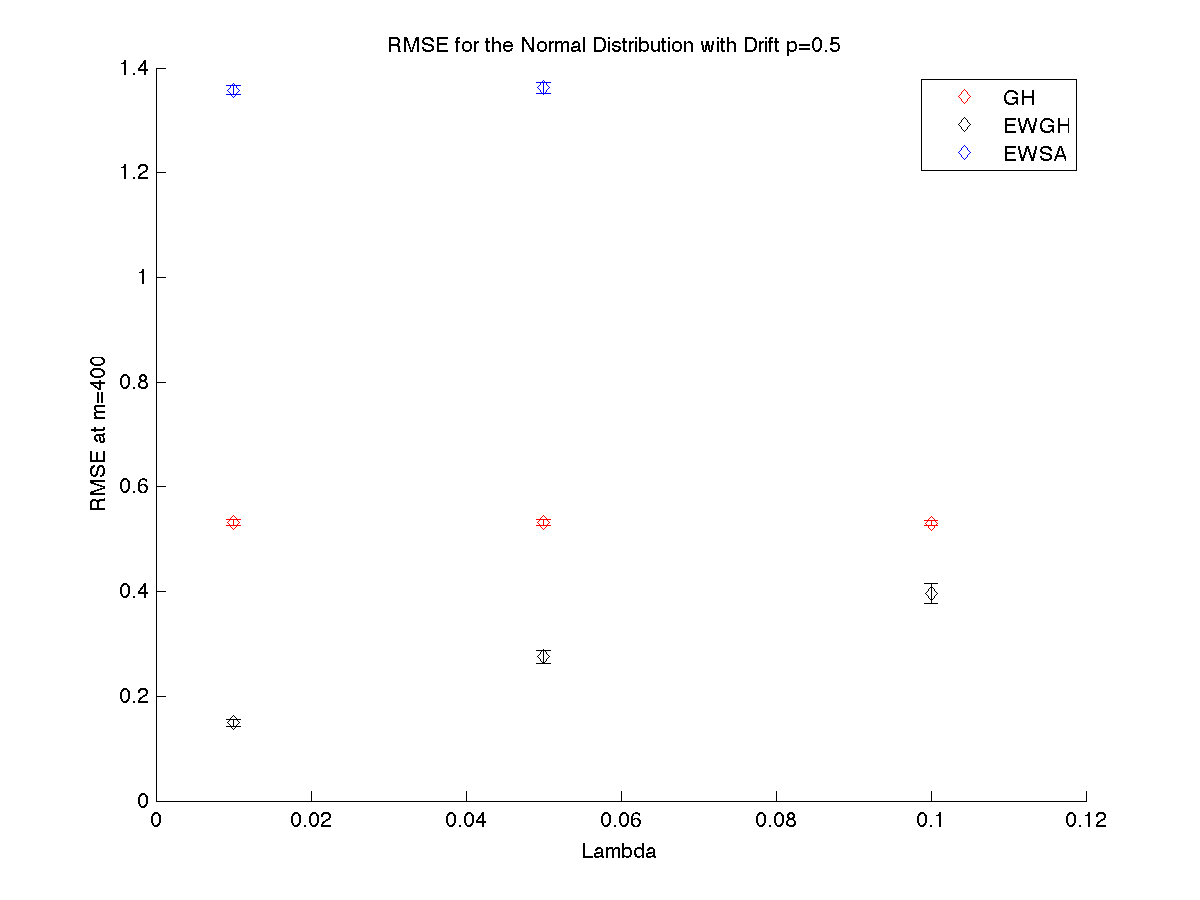}
}
\subfloat[m=1000, p=0.5]{
  \includegraphics[width=62mm]{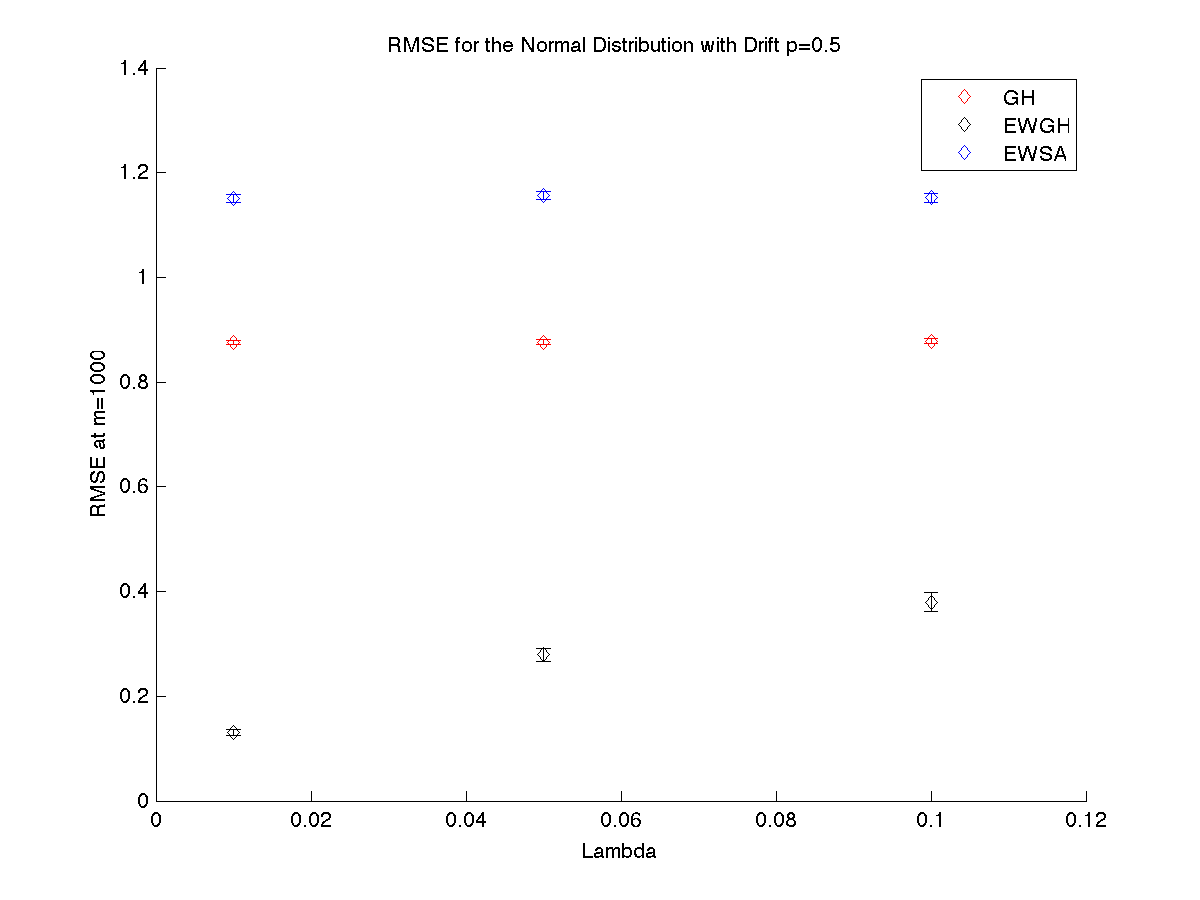}
}
\hspace{0mm}
\subfloat[m=100, p=0.9]{
  \includegraphics[width=62mm]{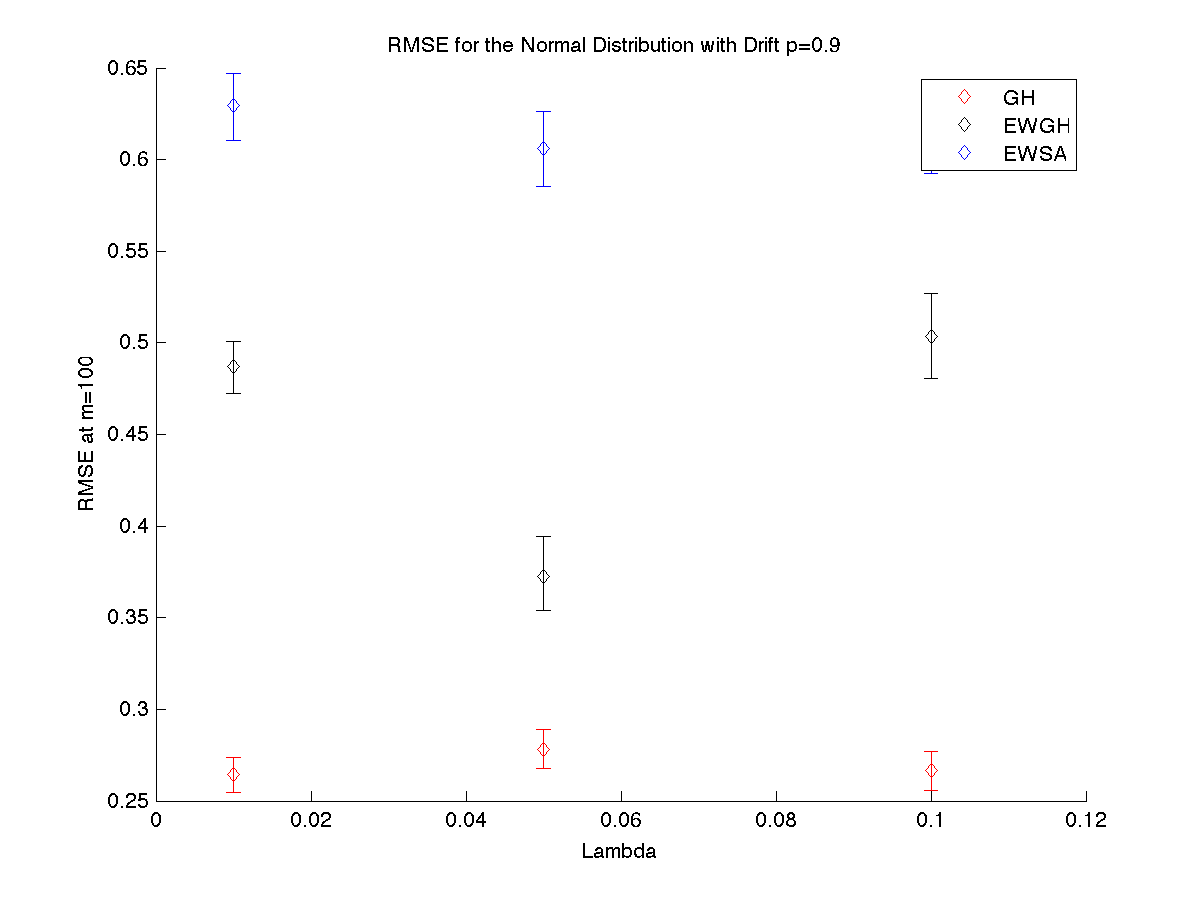}
}
\subfloat[m=400, p=0.9]{
  \includegraphics[width=62mm]{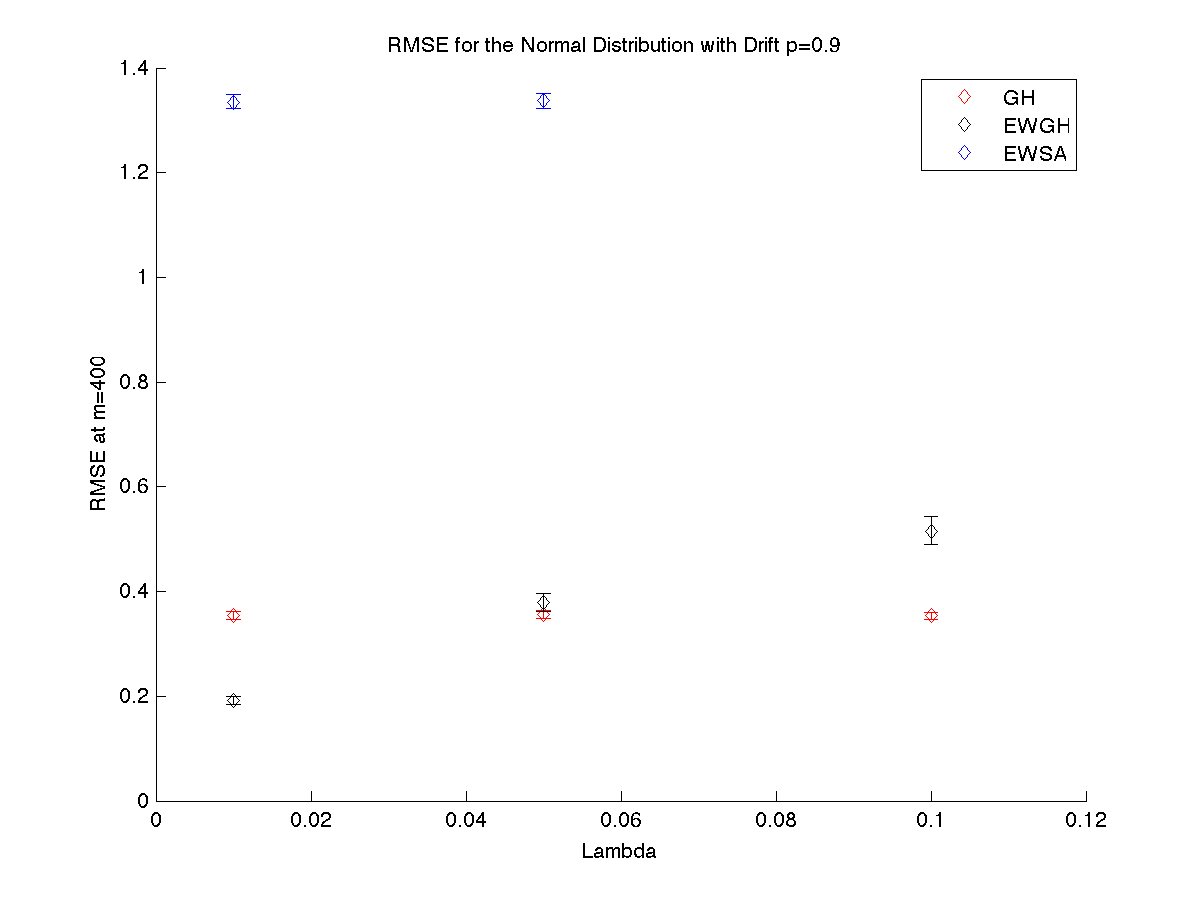}
}
\subfloat[m=1000, p=0.9]{
  \includegraphics[width=62mm]{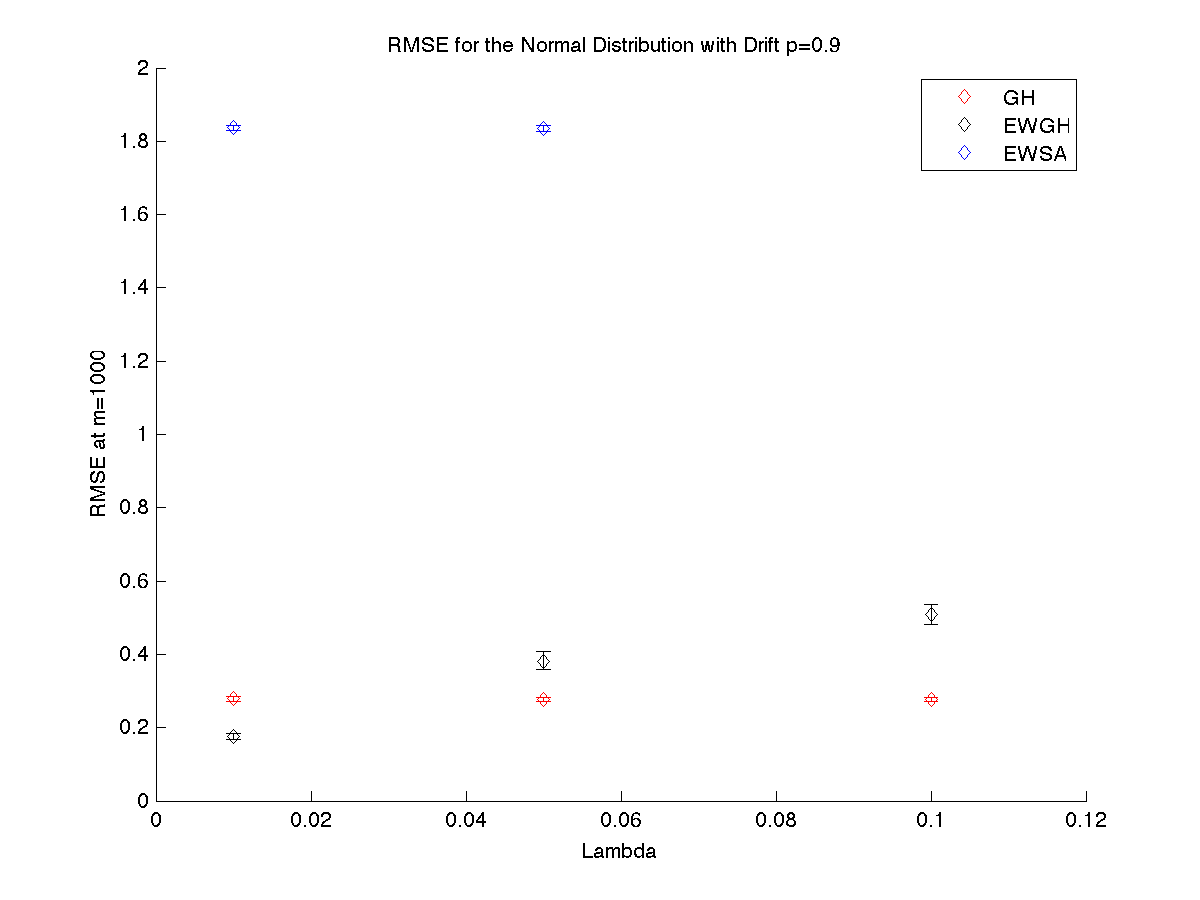}
}
\hspace{0mm}
\subfloat[m=100, p=0.99]{
  \includegraphics[width=62mm]{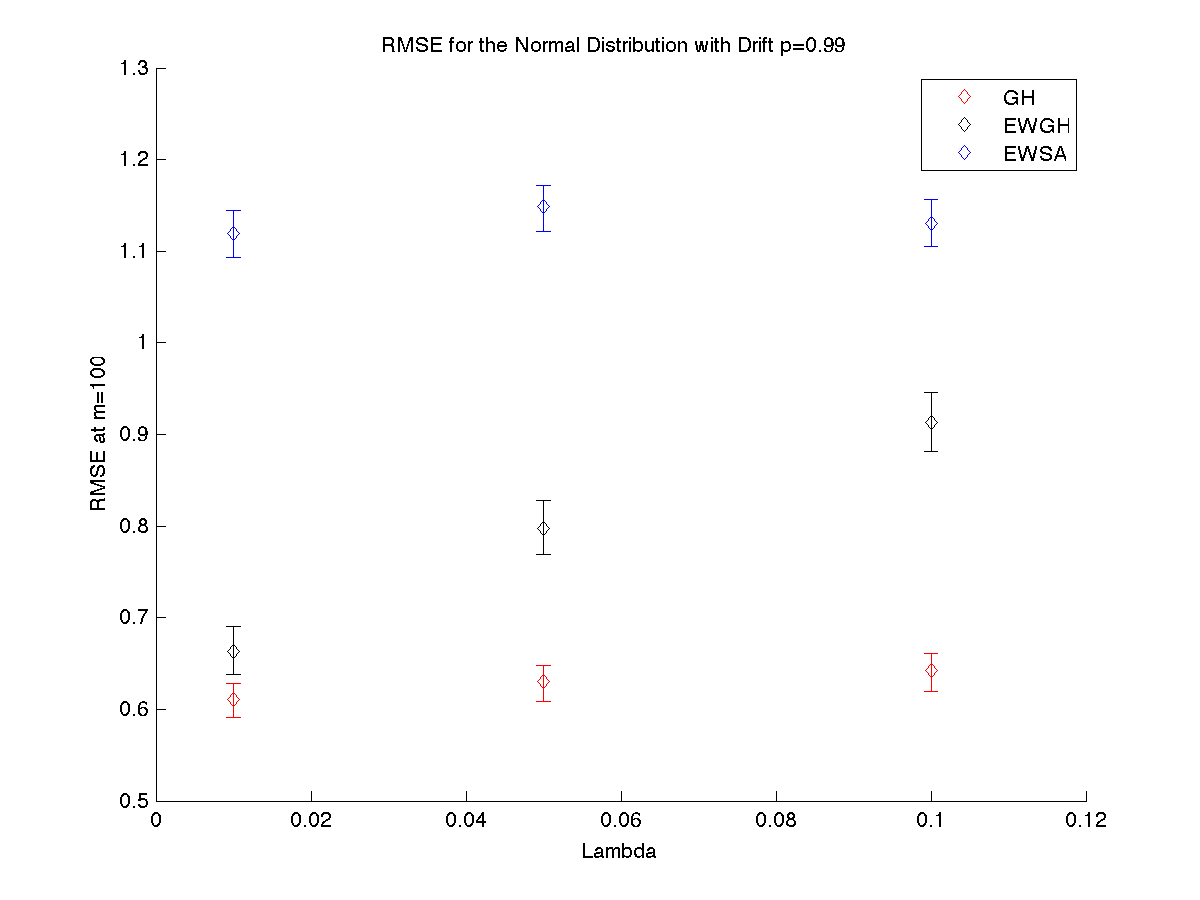}
}
\subfloat[m=400, p=0.99]{
  \includegraphics[width=62mm]{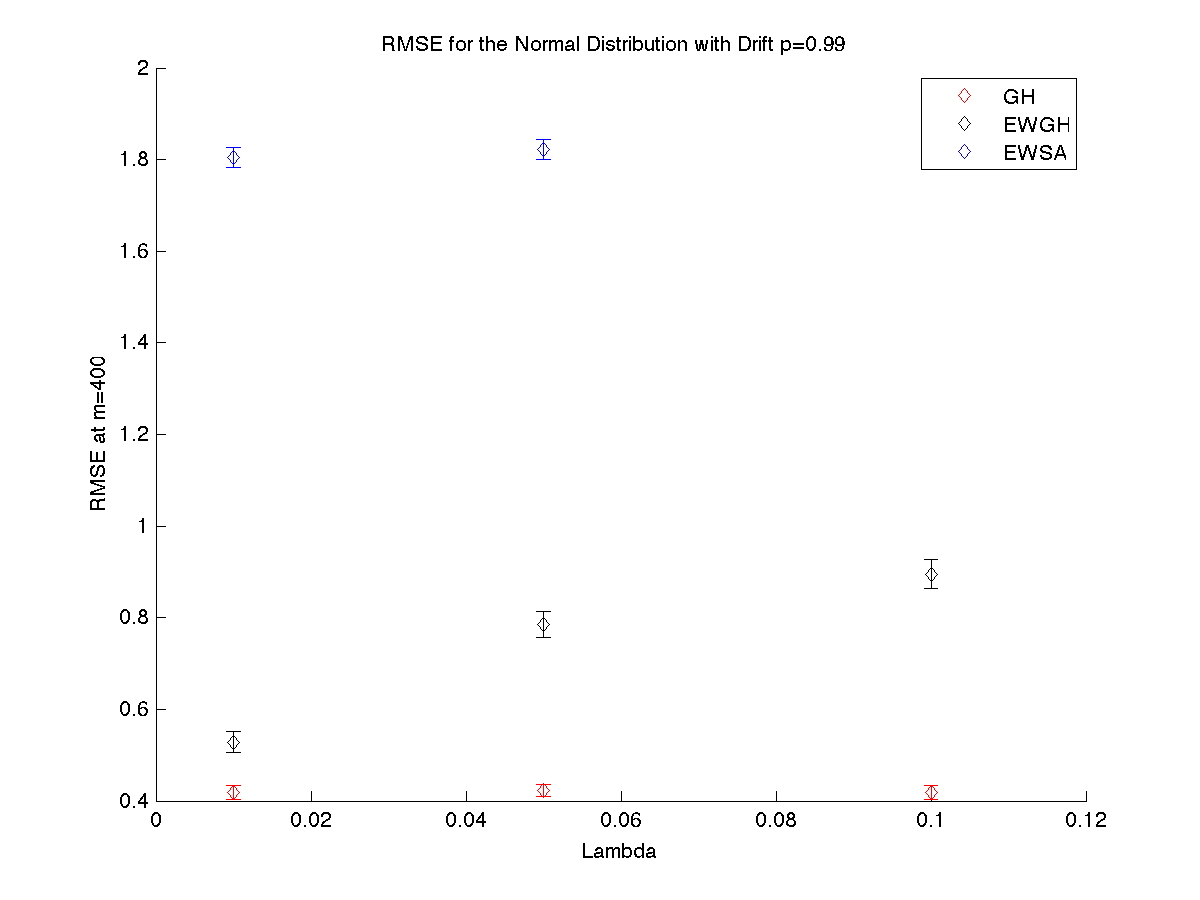}
}
\subfloat[m=1000, p=0.99]{
  \includegraphics[width=62mm]{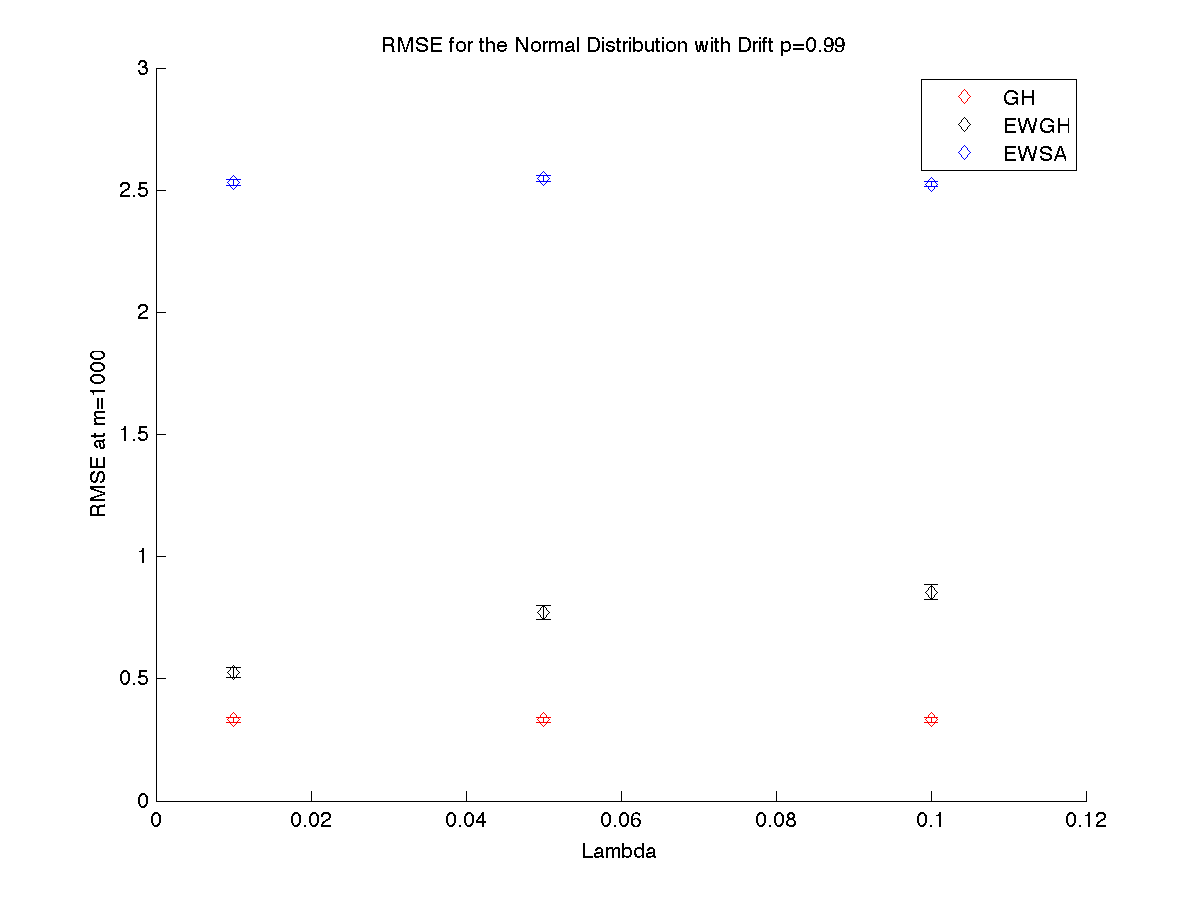}
}

\caption{Normal Distribution with Drift: EWGH RMSE for $\lambda=0.01,0.05,0.1$ at $m=100,400,1000$ observations for the $p=0.5, 0.9, 0.99$ quantiles (including 95\% percentile bootstrap confidence intervals). The GH and EWGH algorithms utilise $N=6$.\label{normalDriftfigsBiasVar}}.
\end{figure}

\end{sidewaysfigure}

\begin{sidewaysfigure}
\begin{figure}[H]
\subfloat[m=100, p=0.5]{
  \includegraphics[width=62mm]{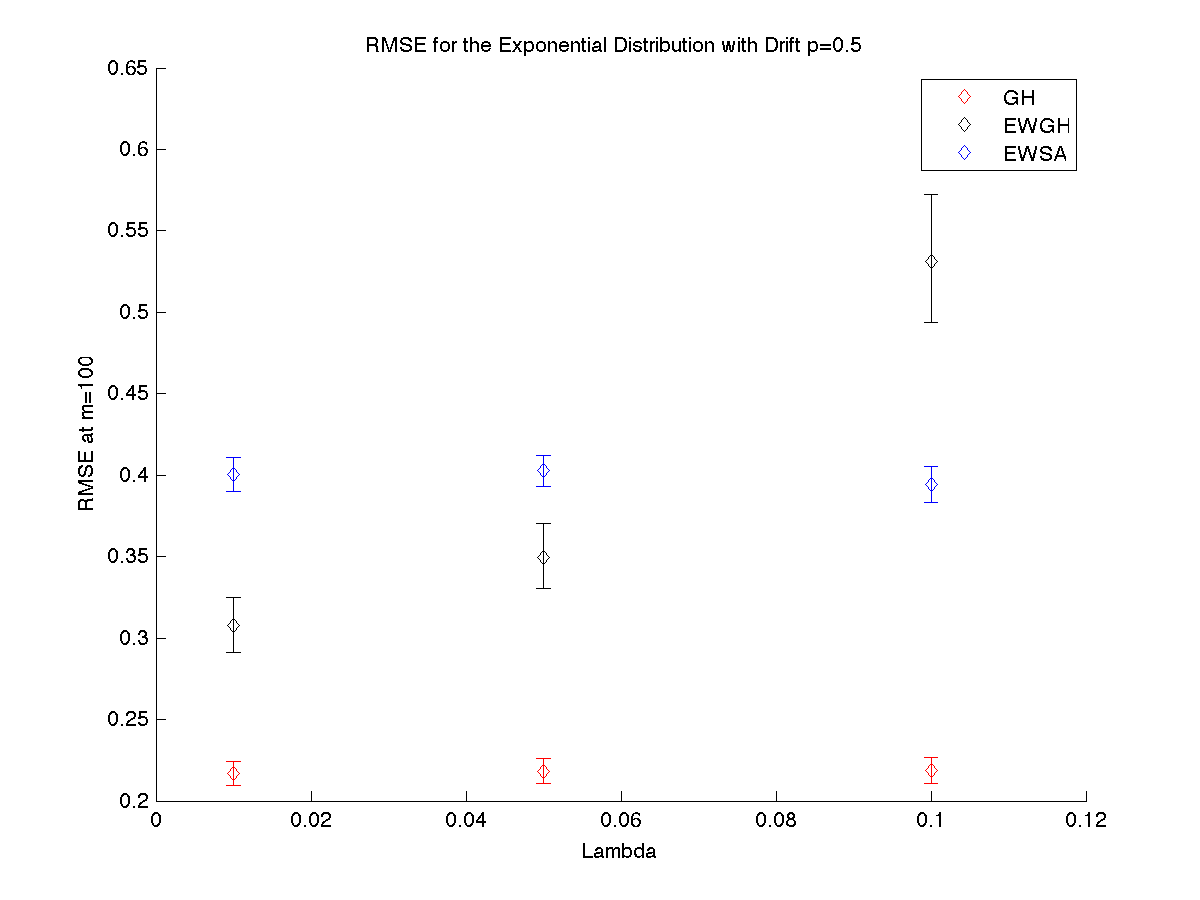}
}
\subfloat[m=400, p=0.5]{
  \includegraphics[width=62mm]{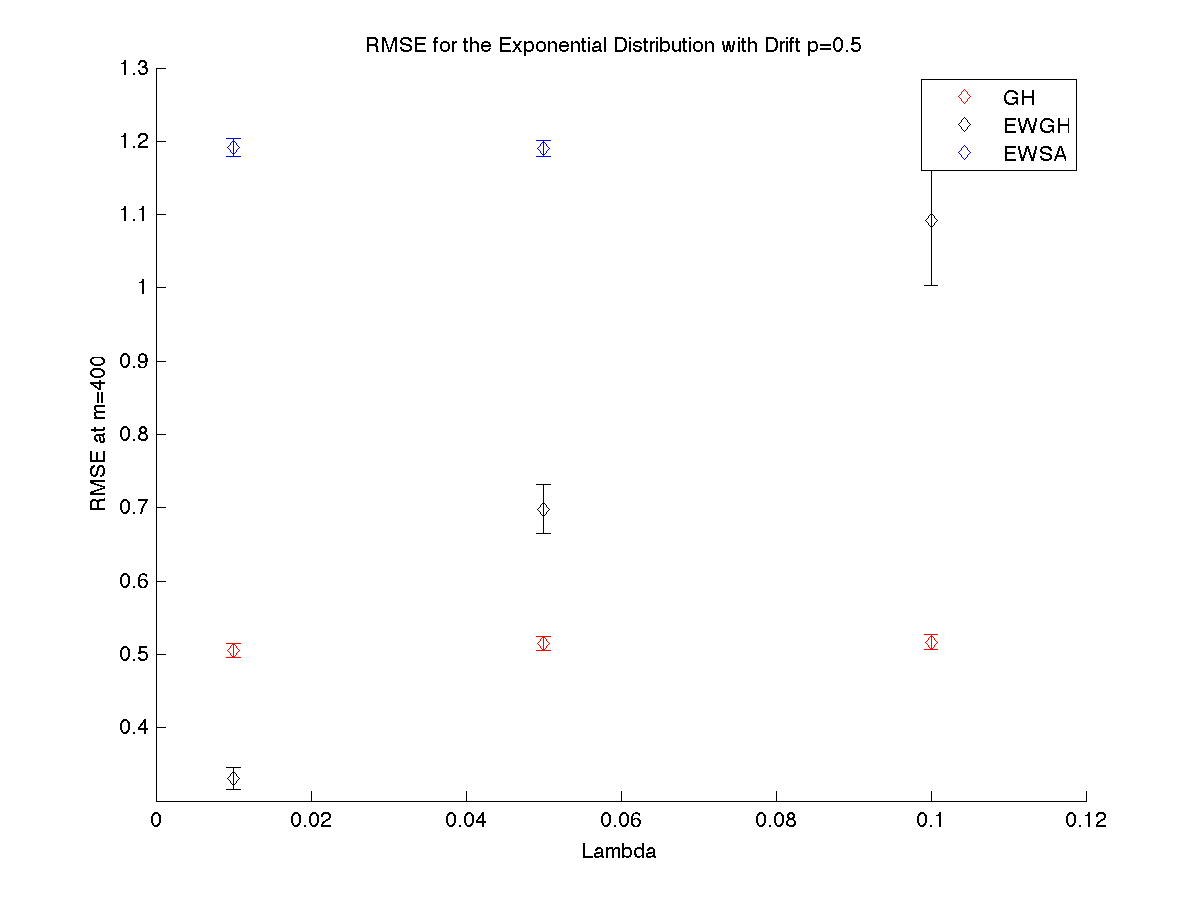}
}
\subfloat[m=1000, p=0.5]{
  \includegraphics[width=62mm]{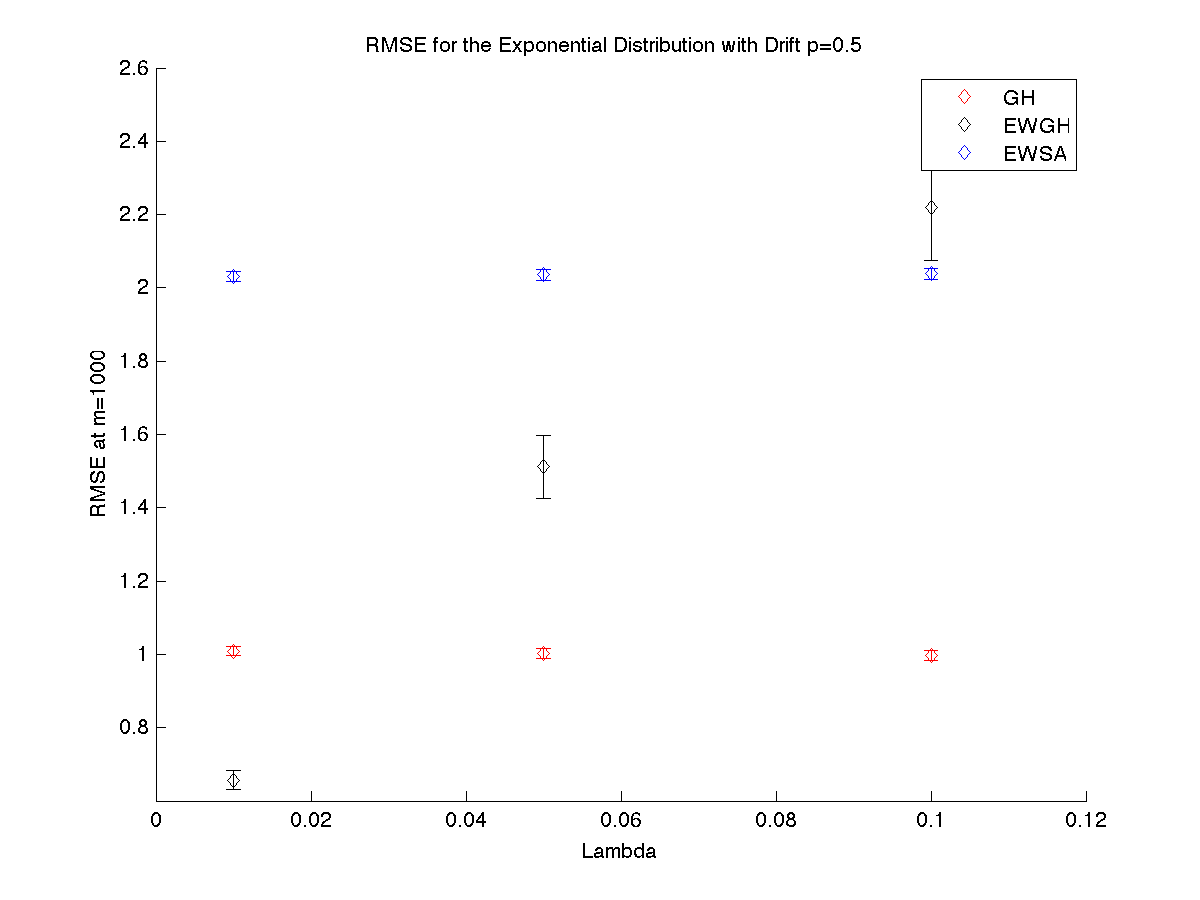}
}
\hspace{0mm}
\subfloat[m=100, p=0.9]{
  \includegraphics[width=62mm]{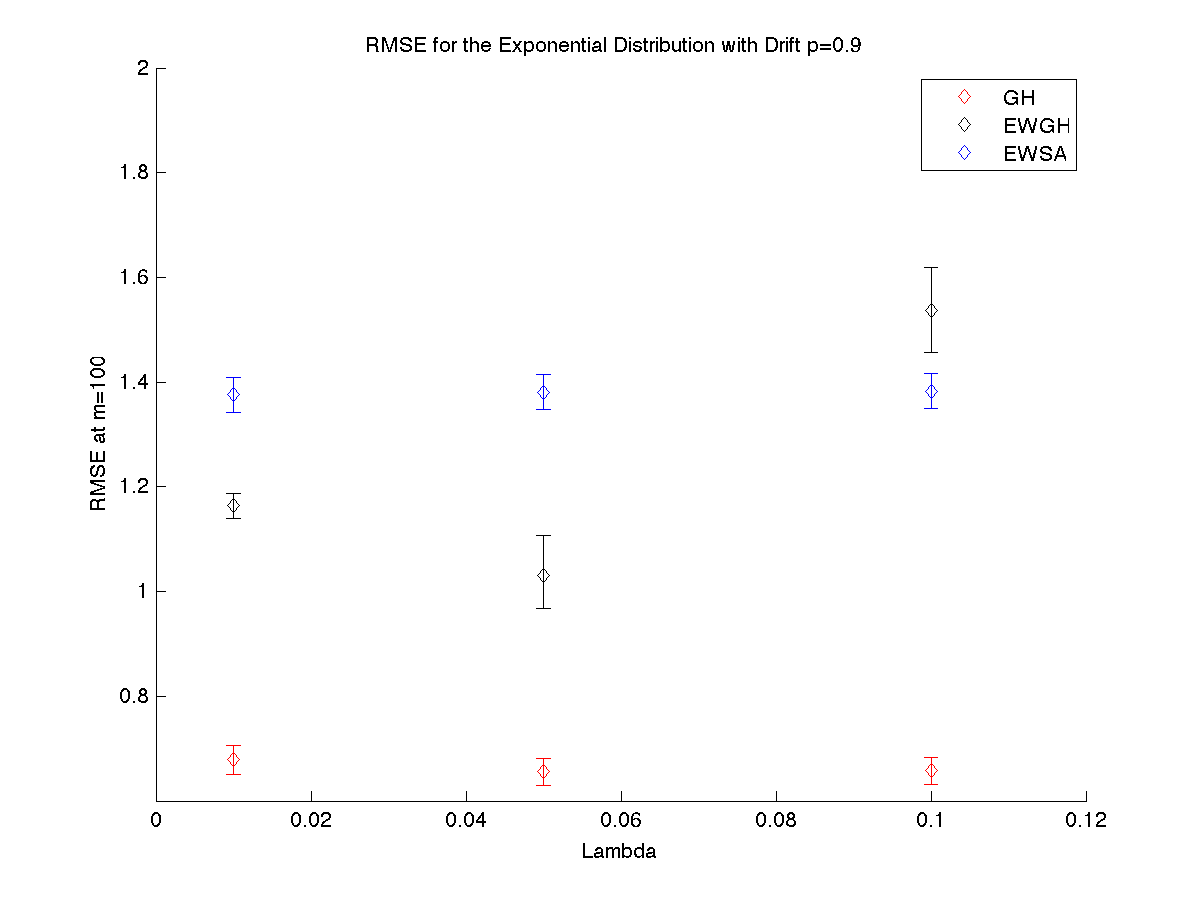}
}
\subfloat[m=400, p=0.9]{
  \includegraphics[width=62mm]{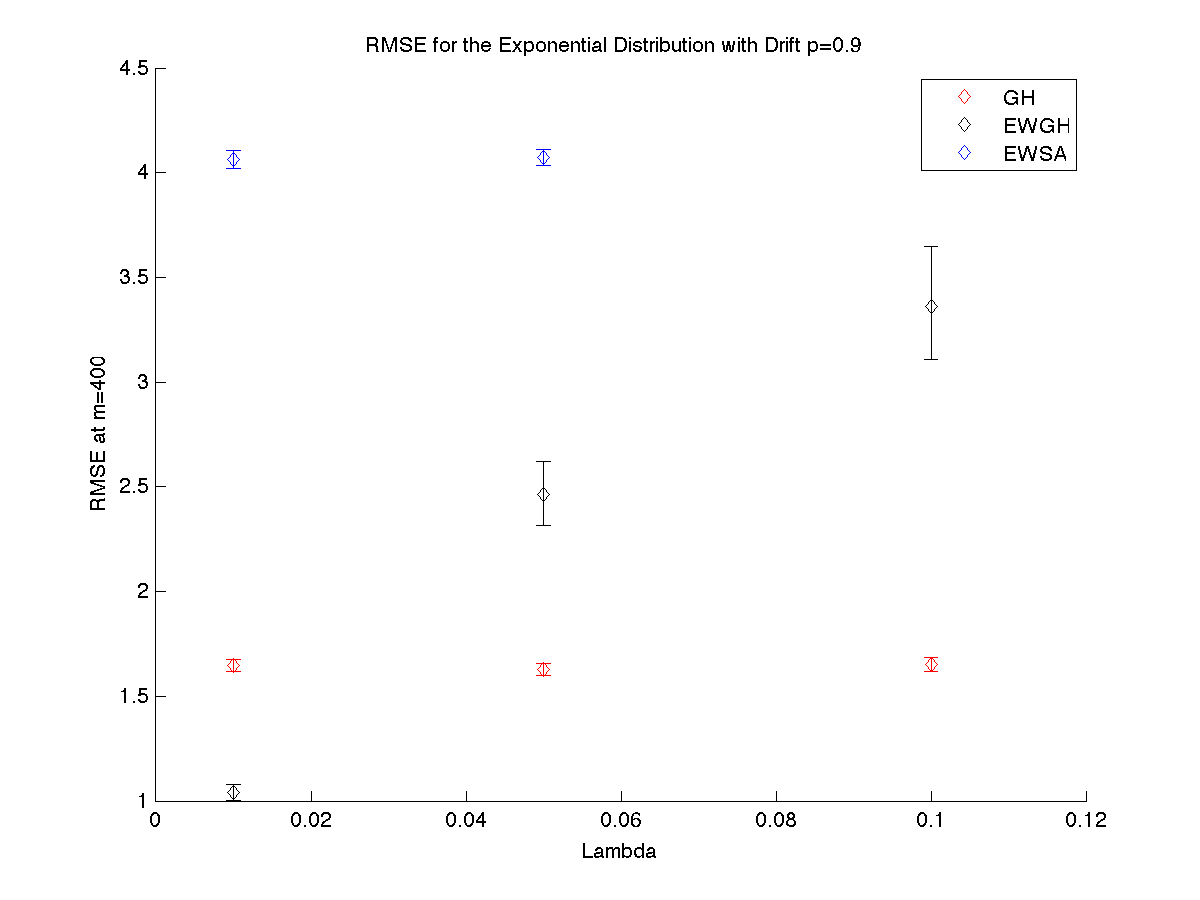}
}
\subfloat[m=1000, p=0.9]{
  \includegraphics[width=62mm]{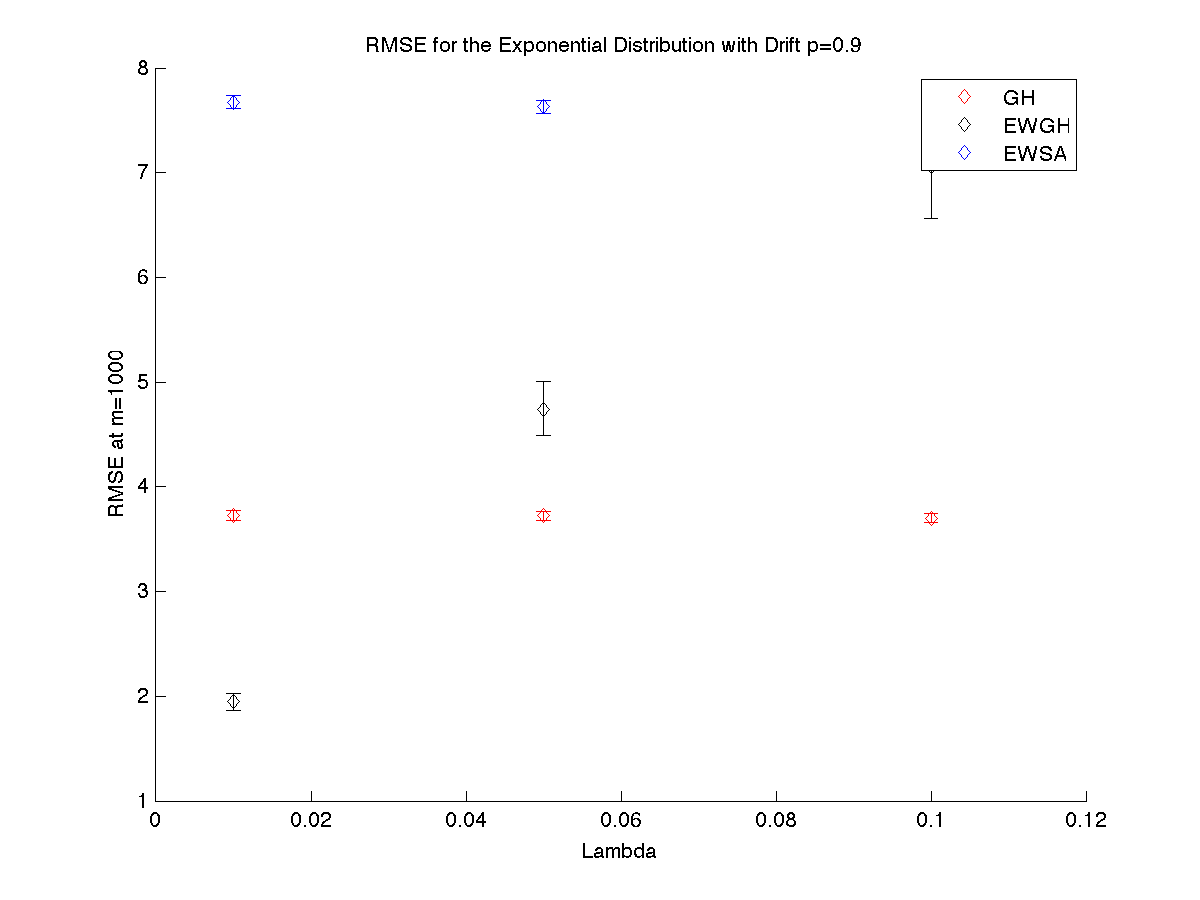}
}
\hspace{0mm}
\subfloat[m=100, p=0.99]{
  \includegraphics[width=62mm]{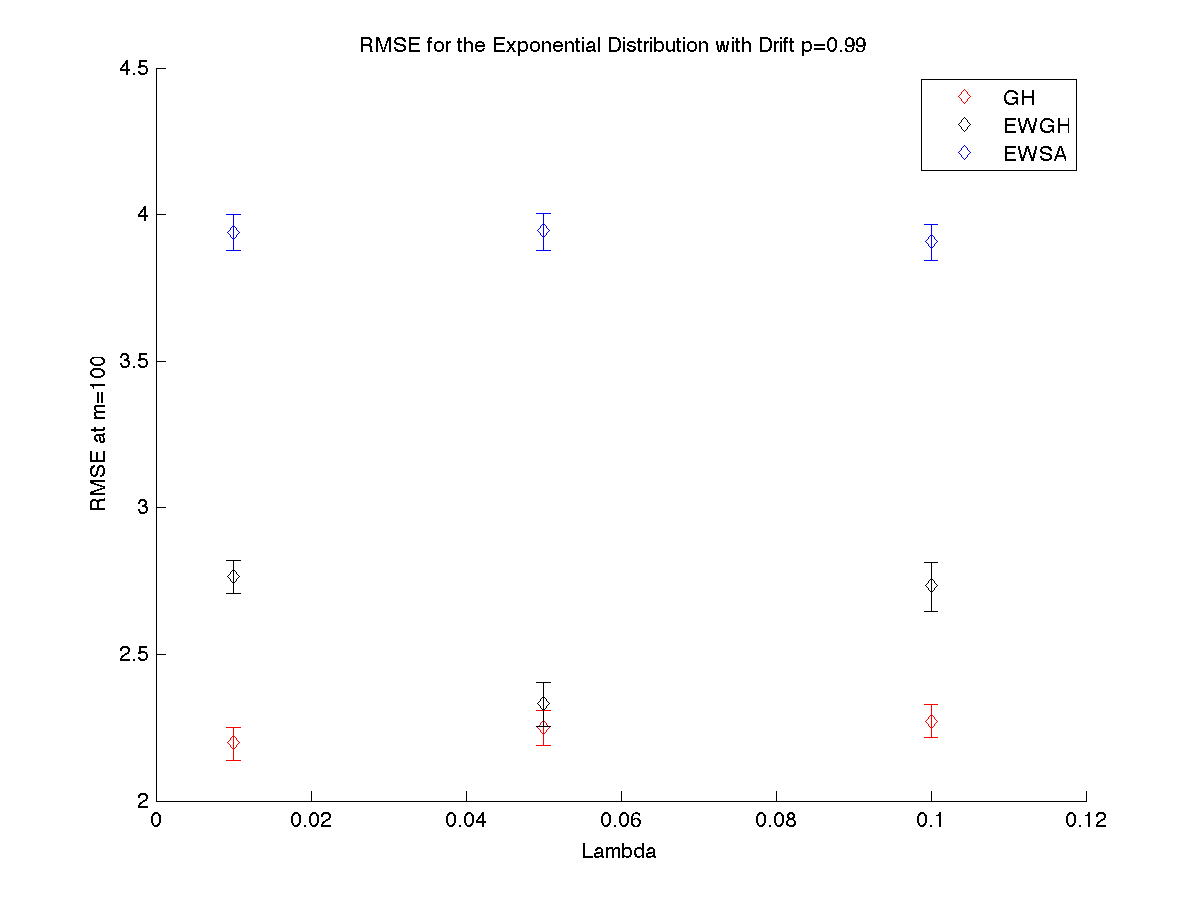}
}
\subfloat[m=400, p=0.99]{
  \includegraphics[width=62mm]{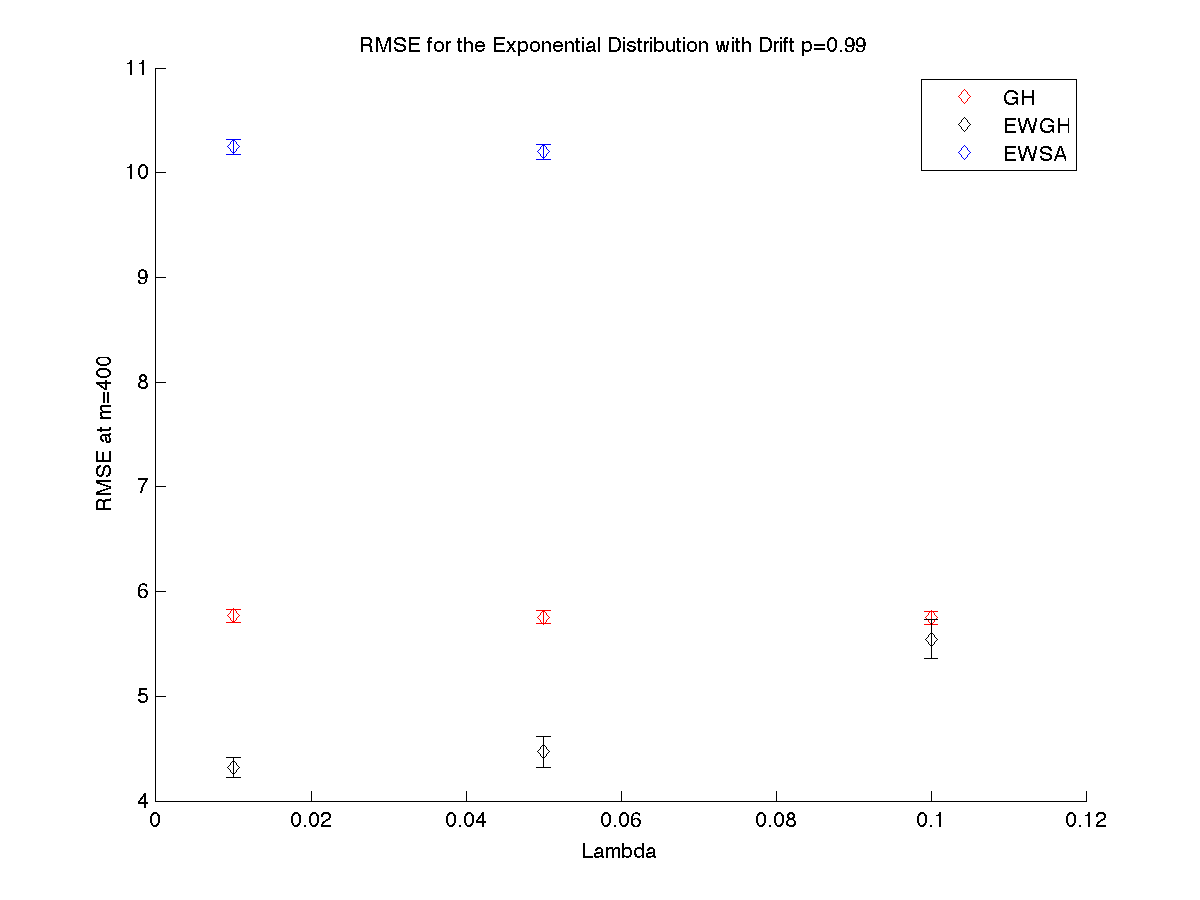}
}
\subfloat[m=1000, p=0.99]{
  \includegraphics[width=62mm]{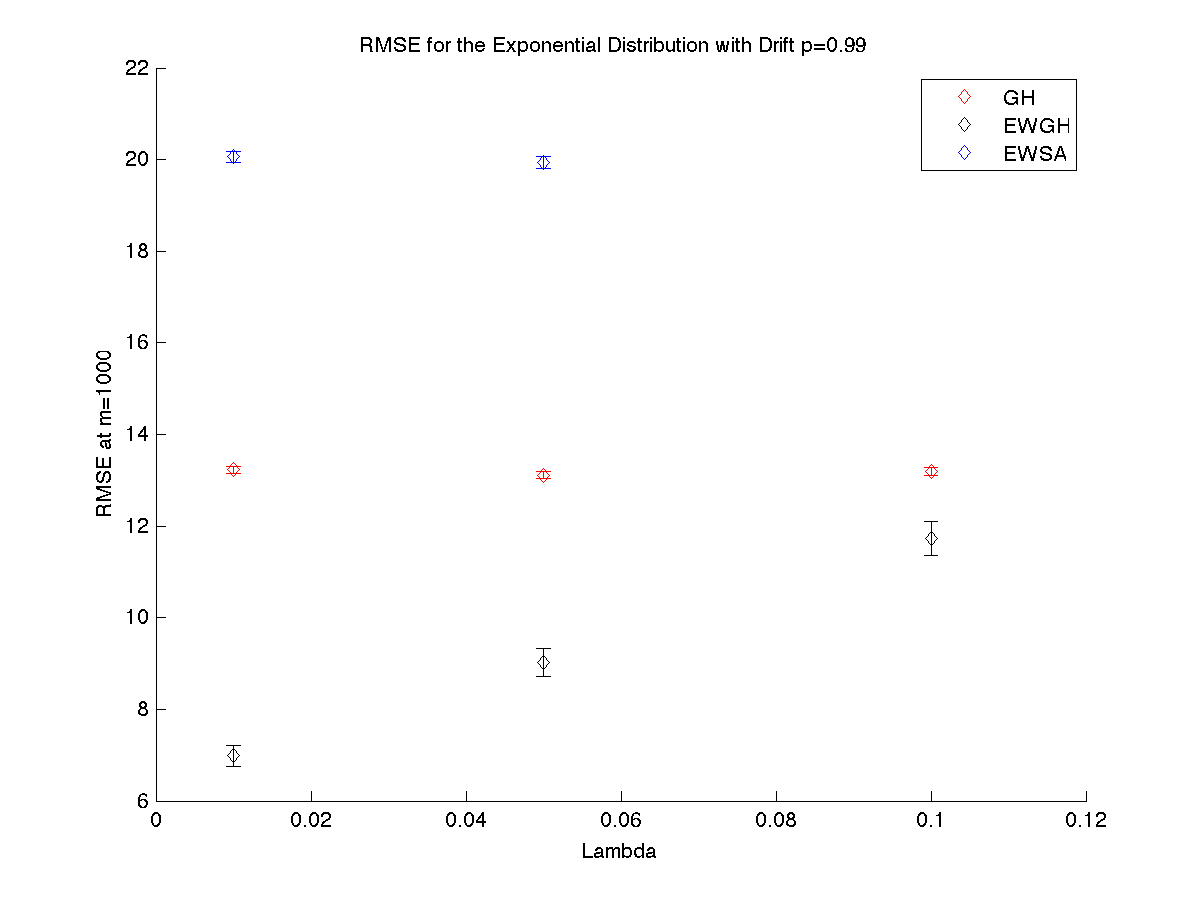}
}

\caption{Exponential Distribution with Drift: EWGH RMSE for $\lambda=0.01,0.05,0.1$ at $m=100,400,1000$ observations for the $p=0.5, 0.9, 0.99$ quantiles (including 95\% percentile bootstrap confidence intervals). The GH and EWGH algorithms utilise $N=6$.\label{expDriftfigsBiasVar}}
\end{figure}
\end{sidewaysfigure}

\begin{sidewaysfigure}
\begin{figure}[H]
\subfloat[p=0.5]{
 \includegraphics[width=62mm]{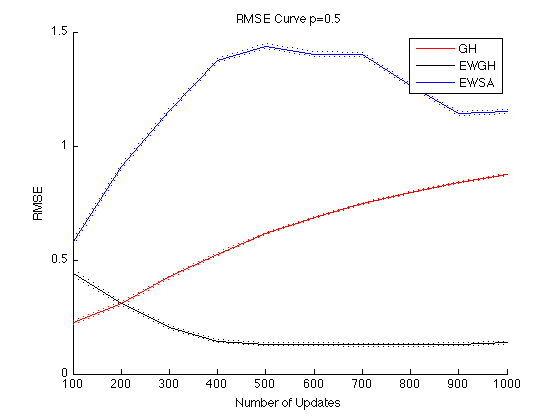}
}
\subfloat[p=0.9]{
  \includegraphics[width=62mm]{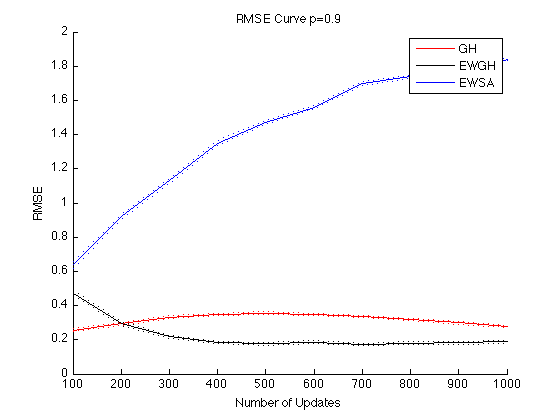}
}
\subfloat[p=0.99]{
  \includegraphics[width=62mm]{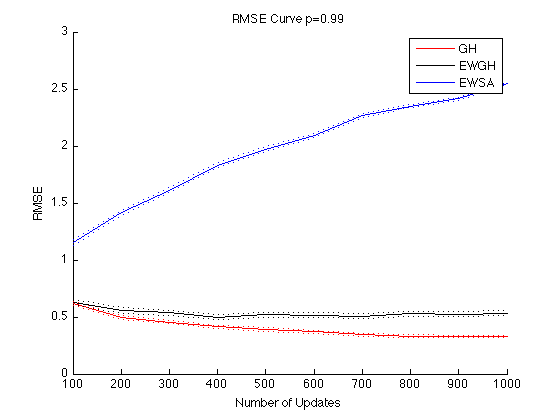}
}
\caption{RMSE curves associated with the normal distribution with mean $=0.006 j$ and standard deviation $=1$ for the 0.5, 0.9 and 0.99 quantiles. The final exact quantiles at $j=m=1000$ are 6, 7.2816 and 8.3263 respectively. The GH and EWGH algorithms utilise $N=6$. The EWGH algorithm utilises $\lambda=0.01$.\label{normalDriftfigs}}
\end{figure}

\begin{figure}[H]
\subfloat[p=0.5]{
  \includegraphics[width=62mm]{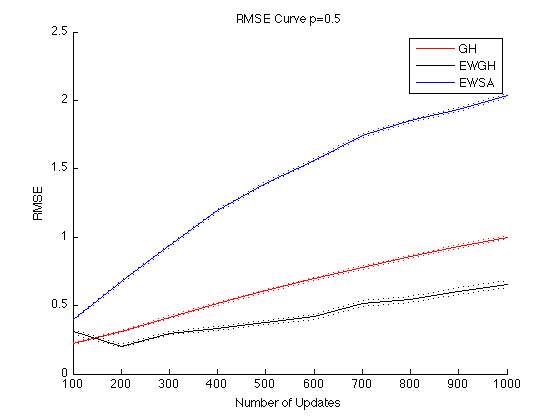}
}
\subfloat[p=0.9]{
  \includegraphics[width=62mm]{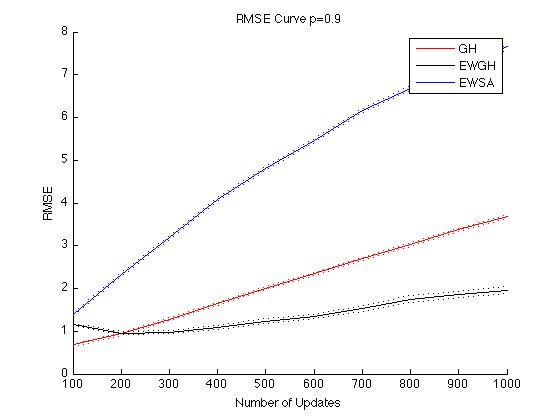}
}
\subfloat[p=0.99]{
 \includegraphics[width=62mm]{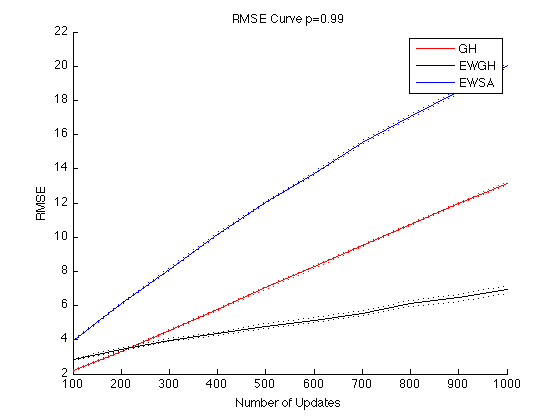}
}
\caption{RMSE curves associated with the exponential distribution with mean and standard deviation  $=1 + 0.006 j$ for the 0.5, 0.9 and 0.99 quantiles. The final exact quantiles at $j=m=1000$ are 4.8562, 16.1319 and 32.2638 respectively. The GH and EWGH algorithms utilise $N=6$. The EWGH algorithm utilises $\lambda=0.01$.\label{expDriftfigs}}
\end{figure}
\end{sidewaysfigure}

\pagebreak
\section{Real Data Results} \label{realresults}

In this section we test the GH and EWGH algorithms on a real data set to evaluate their effectiveness in a non-idealised setting. In particular we consider one month (January 2013) of high frequency forex return data, namely EURUSD spot mid-price returns intervaled at 15 seconds. We begin with the mid-price series:

$$p_{t_{(1)}},p_{t_{(2)}},\dots p_{t_{(m)}} , \quad t_{(i+1)}-t_{(i)} \approx 15 \mbox{ seconds},$$

which we transform online to obtain arithmetic returns:

$$r_{(1)},r_{(2)},\dots r_{(m-1)} , \quad r_{(j)} = p_{t_{(j+1)}}-p_{t_{(j)}}.$$

The summary statistics of the arithmetic returns (in pips, 1 pip = 0.0001) are as follows (data set size = $71797$):

\begin{center}
		\begin{tabular}{ l | c}
		    \hline
		    Statistic & Value (Pips)  \\ \hline
			Mean & 1.0449e-05   \\ 
			Standard Deviation & 1.0852 \\ 
			Skewness &  -0.0222  \\
			Kurtosis & 16.5736\\
		    \hline
		\end{tabular}	
\end{center}

\begin{figure}[H]
    \centering
    \includegraphics[width=100mm]{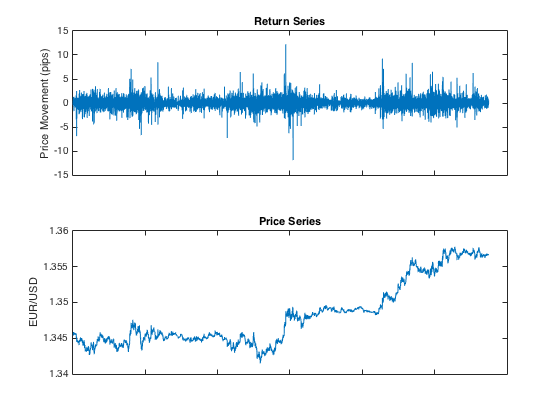}
    \caption{An extract of the forex returns and price series for 2013-01-28 to 2013-01-31. Noteworthy features include the prevalence of return outliers and distinctive changes in return variance. Also apparent are periods of pronounced, but temporary trending behaviour corresponding to changes in the mean of the return distribution.}
    \label{fxPlot}
\end{figure}

Accurately tracking the quantiles of this return series is a non-trivial check of our methods - the distribution of these returns is non-stationary and the frequency of outliers is high. The frequency of outliers is related to the fact that the distribution of the returns is heavy-tailed as evidenced by the high kurtosis. This will probe the dynamic quantile estimation performance in the setting of general non-stationarity as well as evaluating the robustness of the algorithms. Applications of online quantile estimation for financial price series include identifying high frequency trading opportunities, real-time risk estimation (such as the calculation of real time Value at Risk, VaR) and outlier detection. \\

Our test is as follows: we count the number of times the $(i+1)$th observation is smaller than the online estimates of the 0.5 (median), 0.9 and 0.99 quantiles obtained up to observation $i$. These counts are then normalised by the total number of observations. Ideally, the out-of-sample observation should be smaller than the median quantile with probability $0.5$. Similarly, the out-of-sample observation should be smaller than the $0.9$ quantile with probability $0.9$ and it should be smaller than the $0.99$ quantile with probability $0.99$. We compare the observed frequencies to these probabilities and report $95\%$ percentile bootstrap confidence intervals for the observed frequencies. The bootstrap confidence intervals are created by calculating the observed frequencies for each day in the period in question (January 2013), resampling the resultant daily frequencies with replacement (1000 resamples) and providing the 0.025 and 0.975 quantiles of the resampled distribution. We also include the results for the EWSA algorithm for comparison. The parameters used for the algorithms are as follows: $N=6$ for the GH and EWGH algorithms. This choice of $N$ was motivated by the simulation results and computational efficiency considerations. For the EWGH algorithm, $\lambda=0.05$. This choice was motivated by an analysis of an initial sample (not contained within the January 2013 data set) and by the fact that we expect forex return quantiles to vary rapidly around news events and thus a choice of $\lambda$ bigger than $0.01$ (the best choice in the simulation studies) appeared appropriate. The parameters for the EWSA algorithm are the same as in section \ref{simresults}.

\begin{center}
		\begin{tabular}{ | l | l | l | l|}
		    \hline
		    Algorithm & p=0.5 & p=0.9 & p=0.99 \\ \hline
			GH & 0.469 [0.463,0.475] &  0.860 [0.849,0.873] & 0.974 [0.971,0.978] \\ \hline
			EWGH & 0.500 [0.499,0.502] & 0.894 [0.892,0.895] & 0.992 [0.991,0.992]\\ \hline
			EWSA & 0.501 [0.500,0.502] & 0.897 [0.895,0.899] & 0.976 [0.973,0.979] \\
		    \hline
		\end{tabular}\label{tableResults}	
\end{center}

The EWGH algorithm performs well overall and is the only algorithm that does not underestimate the p=0.99 quantile. This provides evidence that not only does the EWGH algorithm perform well in a non-stationary environment, but also that it handles heavy-tailed distributions well. Note that we have performed the same test on other months to check the consistency of the results and the same behaviour emerges (omitted for brevity). We re-iterate that for the GH and EWGH algorithms, we can obtain online estimates of any quantile of interest.  

\section{Conclusion} \label{conclusion}

In this article we have defined a cumulative distribution function estimator based on Hermite series estimators which allows quantiles to be obtained numerically. For probability densities with support on the positive half-real line, we have proven asymptotic MSE consistency (pointwise consistency) as well as asymptotic consistency based on a Cramer-von Mises like criterion for this estimator. We have also provided the associated rates.  In addition, we have demonstrated that quantile estimates obtained from this estimator directly depend on the MISE of the Hermite series density estimator under certain conditions. While these results are novel and interesting in their own right, our particular application of interest is that of sequential quantile estimation. In this setting, the Hermite series based estimators are biased in general. They are still very useful in practice however. In particular, we have introduced algorithms  - based on the Gauss-Hermite expansion - for online quantile estimation in the settings of static quantile estimation and dynamic quantile estimation. These algorithms have $O(1)$ time complexity for \textit{updating} the distribution and quantile function estimates. \\

In the static quantile estimation setting we have exploited the fact that Gauss-Hermite coefficients can be updated in a sequential manner.  To treat dynamic quantile estimation, we have introduced a novel expansion with an exponentially weighted estimator for the Gauss-Hermite coefficients which we have termed the Exponentially Weighted Gauss-Hermite (EWGH) expansion. This expansion should allow the local behaviour of a non-stationary stream of data to be tracked. To make our analysis concrete, we have considered i.i.d data streams and independent, non-identically distributed data streams in our simulation studies and our theoretical analysis. The simulation studies revealed the Gauss-Hermite based algorithms to be competitive with a leading existing algorithm for online quantile estimation. In addition, a test on real forex data confirmed the effectiveness of the EWGH algorithm in a more general and realistic setting and provided evidence that our techniques are effective for heavy-tailed distributions.\\

The particular usefulness of our algorithms is that they allow \textit{arbitrary} quantiles to be estimated in an online manner. They do not require a particular set of quantiles to be specified upfront, which is a limitation of existing algorithms. In obtaining these novel algorithms, we have thus provided a solution to the problem of online distribution function and online quantile function estimation for both stationary and non-stationary data streams. Online estimates of these functions are in fact useful in a broader context in that any function of these quantities can be calculated in an online manner. These online estimates could also be used in online machine learning applications for example.

\section{Acknowledgements}
We would like to thank the reviewers for their very helpful and insightful comments that greatly improved the paper. M.S. would like to thank Stuart Cullender, Norman Ives, Kavishin Pather, Bernhard Steinhardt, Matthew Stephanou and Nils Tania for helpful discussions. M.V. was funded by the National Research Foundation of South Africa.

\appendix

\section{Standardising the observations} \label{standardization}

It is reasonable to assume that in practice, the quality of the fit yielded by the truncated Gauss-Hermite expansion should be better if applied to standardised random variables (with mean equal to zero and standard deviation equal to one). For a random variable $x$ with mean $\mu$ and standard deviation $\sigma$ we have the standardised random variable:

$$\tilde{x} = \frac{x-\mu}{\sigma},$$

with probability density $\tilde{f}(\tilde{x}) = \sigma f(\sigma \tilde{x} + \mu) $. The associated truncated Gauss-Hermite expansion is:

$$\tilde{f}(\tilde{x}) = \sum_{k=0}^{N} a_k H_{k} (\tilde{x}) Z(\tilde{x}),$$

where

$$a_k = \alpha_{k} \int_{-\infty}^{\infty} Z(\tilde{x}) \tilde{f}(\tilde{x}) H_{k} (\tilde{x}) d\tilde{x}.$$ 

Ideally, we would utilise:

\begin{equation}
	\hat{a}_k (\mu,\sigma) = \alpha_{k} \frac{1}{n} \sum_{i=1}^{n} Z(\frac{x_{i}-\mu}{\sigma}) H_{k} (\frac{x_{i}-\mu}{\sigma}), \nonumber
\end{equation}

to estimate the coefficients. In practice however, we do not know the true values of $\mu$ and $\sigma$ and thus we must use estimates of these values, $\hat{\mu}$ and $\hat{\sigma}$. In general, the effect of using these estimates is to \textit{bias} the estimate of $\hat{a}_k$. This can be seen by considering the Taylor expansion of $\hat{a}_k (\hat{\mu} ,\hat{\sigma})$ which implies $E(\hat{a}_k (\hat{\mu} ,\hat{\sigma} ) ) \neq E(\hat{a}_k (\mu,\sigma)) $. \\ 

Despite this bias, standardising using the estimated mean and standard deviation improves the quality of the fit in many  cases. In the sections below we provide online algorithms to estimate the mean and standard deviation in both the static and dynamic setting.  These estimates can then be plugged into the standard Gauss-Hermite coefficients and the exponentially weighted Gauss-Hermite coefficients respectively.

\subsection{Static Quantile Estimation}

The usual estimators of the mean and standard deviation can be calculated in an online way. The mean ($\hat{\mu}_{k}$) can be updated with each new observation as:

$$\hat{\mu}_{1} = x_{1},$$

$$\hat{\mu}_{k} = \frac{1}{k} \left( (k-1) \hat{\mu}_{k-1} + x_{k} \right), \quad k \geq 2.$$

The standard deviation ($\hat{\sigma}_{k}$) can also be estimated in an online manner. To avoid numerical precision problems an algorithm such as that originating from  Welford \cite{welford} can be applied.

$$M_{1} = x_{1},$$

$$S_{1} = 0,$$

$$M_{k} = M_{k-1} + \frac{\left(x_{k} - M_{k-1} \right)}{k},$$

$$S_{k}=S_{k-1} + \left(x_{k} - M_{k-1}\right)\left(x_{k} - M_{k}\right),$$

$$\hat{\sigma}_{k} = \sqrt{\frac{S_{k}}{k-1}}, \quad k \geq 2.$$

\subsection{Dynamic Quantile Estimation}

For the purposes of obtaining a local estimate of the mean and standard deviation we can utilise EWMA estimators:

$$\hat{\mu}_{1} = x_{1},$$
$$\hat{\mu}_{k} = (1-\lambda) \hat{\mu}_{k-1}+\lambda x_{k}, \quad k \geq 2,$$

$$\hat{V}_{1} = 1,$$
$$\hat{V}_{k} = (1-\lambda) \hat{V}_{k-1} + \lambda (x_{k}-\hat{\mu}_{k})^{2}, \quad k \geq 2,$$
$$\hat{\sigma}_{k} = \sqrt{\hat{V}_{k}}.$$

$$0<\lambda \leq 1.$$

\section{MISE of the Exponentially Weighted Gauss-Hermite Expansion} \label{theoryEWGH}

\subsection{MISE of the Exponentially Weighted Gauss-Hermite Expansion for IID Data}

\begin{proposition}\label{MSECoeffBound}
Let the true probability density function $f(x) \in L_2$ and\\ $E|X|^{\frac{2}{3}}<\infty$. The MSE of the coefficients (\ref{GHCoeffEWMA}) have the following upper bound for a sample of $n+1$ observations in the i.i.d. case:

\begin{equation}
	\mbox{MSE}(\hat{a}_k) \leq  \left[ \frac{\lambda}{2-\lambda} \left[1-(1-\lambda)^{2n} \right] + (1-\lambda)^{2n} \right] \frac{\sqrt{\pi} c (k+1)^{-1/2}}{2^{k-1}k!}, \nonumber
\end{equation}

where $c$ is a constant.

\end{proposition}

\begin{proof}

The truncated EWGH expansion of  $f(x)$ is given by:

$$\hat{f}_{N}(x) = \sum_{k=0}^{N} \hat{a}_k H_{k} (x) Z(x),$$ 

where the coefficient estimates are given by (\ref{GHCoeffEWMA}). In what follows we assume an i.i.d sequence of $n+1$ observations has been drawn from $f(x)$ i.e. $\mathbf{x}_{i} \sim f(x)$.\\

The estimator (\ref{GHCoeffEWMA}) can be equivalently written as:

$$\hat{a}_k = \lambda \sum_{j=0}^{n-1} (1-\lambda)^{j} \left[\alpha_{k} Z(\mathbf{x}_{n-j}) H_{k} (\mathbf{x}_{n-j})\right] + (1-\lambda)^{n} \left[\alpha_{k} Z(\mathbf{x}_{0}) H_{k} (\mathbf{x}_{0})\right].$$

The expected value of this estimator, $E\left(\hat{a}_k\right) = a_k$ and thus the estimator is unbiased.\\

The variance, $\mbox{Var}(\hat{a}_k) = E\left[ (\hat{a}_k - a_k)^{2}\right]$ of the estimator is:

\begin{eqnarray}
	\mbox{Var}(\hat{a}_k) &=&	\lambda^{2} \sum_{j=0}^{n-1} (1-\lambda)^{2j} \mbox{Var}\left(\left[\alpha_{k} Z(\mathbf{x}_{n-j}) H_{k} (\mathbf{x}_{n-j})\right]\right) \nonumber\\
	&+& (1-\lambda)^{2n} \mbox{Var}\left(\left[\alpha_{k} Z(\mathbf{x}_{0}) H_{k} (\mathbf{x}_{0})\right]\right)\nonumber\\
	&=&\left[ \frac{\lambda}{2-\lambda} \left[1-(1-\lambda)^{2n} \right] + (1-\lambda)^{2n} \right]\sigma_{X}^{2}(k),
\end{eqnarray}

where $\sigma_{X}^{2}(k) = \mbox{Var}\left[ \alpha_{k} Z(x) H_{k} (x)  \right] = E \left[ \frac{\pi}{\left[2^{k-1}k!\right]^2} \left[Z(x)\right]^2 \left[H_{k}(x)\right]^2 \right] - a_k^{2}$. Note that we have exploited the independence of the observations $\mathbf{x}_{i}$. \\

Now, we can obtain an upper bound on $\sigma_{X}^{2}(k)$ using the properties of Hermite polynomials (following from (\ref{hermiteInequal1}) and (\ref{hermiteInequal2}), see \cite{convergence2}):

$$\sigma_{X}^{2}(k) \leq \frac{\sqrt{\pi} c (k+1)^{-1/2}}{2^{k-1}k!}.$$

This yields the following bound on the MSE of $\hat{a}_k$ (since $\hat{a}_k$ is unbiased, the bound on the MSE of $\hat{a}_k$ is equal to the bound on the variance):

\begin{eqnarray}
	\mbox{MSE}(\hat{a}_k) &=& E\left[ (\hat{a}_k - a_k)^{2}\right]\nonumber\\
	&\leq& \left[ \frac{\lambda}{2-\lambda} \left[1-(1-\lambda)^{2n} \right] + (1-\lambda)^{2n} \right] \frac{\sqrt{\pi} c (k+1)^{-1/2}}{2^{k-1}k!}.
\end{eqnarray}

\end{proof}

\begin{theorem} \label{theoremEWGHMISE1}

Let the true probability density function $f(x) \in L_2$ and $E|X|^{\frac{2}{3}}<\infty$. If $r \geq 1$ derivatives of $f(x)$ exist and $(x-\frac{d}{dx})^r f(x) \in L_{2}$ then the MISE for the EWGH expansion in the i.i.d. case behaves as follows for $N$ sufficiently large and $n \to \infty$:

\begin{equation}
	\mbox{MISE}\,(\hat{f}_{N})\, = O(N^{1/2}) \left[ \frac{\lambda}{2-\lambda} \right] + O(N^{-r}).  \nonumber
\end{equation}

\end{theorem}

\begin{proof}

If $r \geq 1$ derivatives of $f(x)$ exist and $(x-\frac{d}{dx})^r f(x) \in L_{2}$ then\\ $\sum_{k=N+1}^{\infty} \frac{2^{k-1}k!}{\sqrt{\pi}} a_{k}^{2} = O(N^{-r})$ (see \cite{convergence1},\cite{convergence2}). Combining this fact and the bound on $\mbox{MSE}(\hat{a}_k)$ (proposition \ref{MSECoeffBound}) with the expression for the MISE (\ref{MISE}) we obtain the following:

\begin{eqnarray}
	\mbox{MISE}\,(\hat{f}_{N}) &=& \sum_{k=0}^{N} \frac{2^{k-1}k!}{\sqrt{\pi}} E\left[(\hat{a}_k - a_k)^{2}\right] + \sum_{k=N+1}^{\infty} \frac{2^{k-1}k!}{\sqrt{\pi}} a_{k}^{2}\nonumber \\
	&\leq& \left[ \frac{\lambda}{2-\lambda} \left[1-(1-\lambda)^{2n} \right] + (1-\lambda)^{2n} \right]  c \sum_{k=0}^{N}  (k+1)^{-1/2}\nonumber\\
	&+& O(N^{-r})\nonumber \label{MISECalc} 
\end{eqnarray}

For $N$ sufficiently large and $n \to \infty$ we have:

$$\mbox{MISE}\,(\hat{f}_{N})\, = O(N^{1/2}) \left[ \frac{\lambda}{2-\lambda} \right] + O(N^{-r}).$$

\end{proof}

\subsection{MISE of the Exponentially Weighted Gauss-Hermite Expansion for Non-identically Distributed Data} \label{theoryEWGHdep}

In this section we extend the results derived for i.i.d data to non-identically distributed, independent data. We consider the case where we observe $s+1$ observations from a probability distribution $f_{1} \in L_{2}$ followed by a further $t$ observations from a second distribution $f_{2} \in L_{2}$ i.e. we assume an independent sequence of $s+1$ observations $\mathbf{x}_{0}, \dots, \mathbf{x}_{s} \sim f_{1}(x)$ followed by an independent sequence of $t$ observations, $\mathbf{x}_{s+1}, \dots, \mathbf{x}_{t+s} \sim f_{2}(x)$. We consider this a fundamental example of non-identically distributed data. We evaluate and bound the MISE of $\hat{f}_N(x)$ compared to $f_{2}(x)$ at observation $t+s$. We denote the true Gauss-Hermite coefficients of $f_{1}(x)$ as $a_k^{(1)}$ and the true Gauss-Hermite coefficients of $f_{2}(x)$ as $a_k^{(2)}$.\\

\begin{proposition}\label{MSECoeffBoundNonIndep}
The MSE of the coefficients (\ref{GHCoeffEWMA}) compared to $a_k^{(2)}$ have the following upper bound after $s+1$ independent observations from  $f_{1} \in L_{2}$ followed by $t$ independent observations from $f_{2} \in L_{2}$ provided $E|X|^{\frac{2}{3}}<\infty$ for both distributions:

\begin{align*}
	&\left[E\left(\hat{a}_k - a_k^{(2)}\right)^{2}\right]\\ \nonumber 
	&\leq \frac{c \sqrt{\pi} (k+1)^{-1/2}}{2^{k-1}k!}\left[4(1-\lambda)^{2t} + \frac{\lambda}{2-\lambda} \left[1-(1-\lambda)^{2(s+t)} \right] + (1-\lambda)^{2(s+t)} \right].\nonumber
\end{align*}

\end{proposition}

\begin{proof}

The expected value of the exponentially weighted estimator for the\\ Gauss-Hermite coefficients is:

\begin{eqnarray}
	E\left(\hat{a}_k\right) &=& \lambda \sum_{j=0}^{s+t-1} (1-\lambda)^{j} E\left(\left[\alpha_{k} Z(\mathbf{x}_{s+t-j}) H_{k} (\mathbf{x}_{s+t-j})\right]\right) \nonumber\\
	&+& (1-\lambda)^{s+t} E\left(\left[\alpha_{k} Z(\mathbf{x}_{0}) H_{k} (\mathbf{x}_{0})\right]\right) \nonumber\\
	&=& \lambda a_k^{(1)} \sum_{j=t}^{s+t-1} (1-\lambda)^{j} + \lambda a_k^{(2)} \sum_{j=0}^{t-1} (1-\lambda)^{j}  + (1-\lambda)^{s+t} a_k^{(1)} \nonumber\\
	&=& a_k^{(2)} + (1-  \lambda)^{t} [a_k^{(1)}-a_k^{(2)}].\nonumber\\
\end{eqnarray}

Thus the squared bias of the estimator compared to the true Gauss-Hermite coefficient $a_{k}^{(2)}$ is:

$$\left[E\left(\hat{a}_k\right) - a_k^{(2)}\right]^2 = (1-  \lambda)^{2t} [a_k^{(1)}-a_k^{(2)}]^2.$$

Similarly, the variance $\mbox{Var}(\hat{a}_k) = E\left[ (\hat{a}_k - E\left(\hat{a}_k\right))^{2}\right]$ of the estimator is:

\begin{align}
	\mbox{Var}(\hat{a}_k) &=\lambda^{2}[\sigma^{(1)}_{X}(k)]^{2}\sum_{j=t}^{s+t-1}(1\!-\!\lambda)^{2j}\!+\!\lambda^{2}[\sigma^{(2)}_{X}(k)]^{2}\sum_{j=0}^{t-1}(1\!-\!\lambda)^{2j}\! \nonumber \\
	&+\!(1\!-\!\lambda)^{2(s+t)}[\sigma^{(1)}_{X}(k)]^{2}\nonumber\\
	&= [\sigma^{(1)}_{X}(k)]^{2}  \left[ \frac{\lambda}{2-\lambda} (1-\lambda)^{2t} + \left[ 1-\frac{\lambda}{2-\lambda}\right](1-\lambda)^{2(s+t)}\right]  \nonumber\\
	&+[\sigma^{(2)}_{X}(k)]^{2} \frac{\lambda}{2-\lambda} \left[ 1-(1-\lambda)^{2t}\right],
\end{align}

where  $[\sigma^{(1)}_{X}(k)]^{2} = \mbox{Var}\left[ \alpha_{k} Z(x) H_{k} (x)  \right], x \sim f_{1}(x)$ and \\
$[\sigma^{(2)}_{X}(k)]^{2} = \mbox{Var}\left[ \alpha_{k} Z(x) H_{k} (x)  \right], x \sim f_{2}(x)$. \\

Thus the mean squared error of $\hat{a}_k$ compared to $a_k^{(2)} $ is:

\begin{eqnarray}
	\left[E\left(\hat{a}_k - a_k^{(2)}\right)^{2}\right] &=& \left[E\left(\hat{a}_k\right) - a_k^{(2)}\right]^2 + \mbox{Var}(\hat{a}_k) \nonumber\\
	&=& (1-  \lambda)^{2t} [a_k^{(1)}-a_k^{(2)}]^2\nonumber\\
	&+& [\sigma^{(1)}_{X}(k)]^{2}  \left[ \frac{\lambda}{2-\lambda} (1-\lambda)^{2t} + \left[ 1-\frac{\lambda}{2-\lambda}\right](1-\lambda)^{2(s+t)}\right]  \nonumber\\
	&+& [\sigma^{(2)}_{X}(k)]^{2} \frac{\lambda}{2-\lambda} \left[ 1-(1-\lambda)^{2t}\right].
\end{eqnarray}

The MSE therefore depends on the difference between the true Gauss-Hermite coefficients of $f_{1}(x)$  and $f_{2}(x)$ and the number of observations since the switch between the distributions (along with the value of $\lambda$). We can bound the MSE above using the properties of Hermite polynomials:

\begin{align}
	\left[E\left(\hat{a}_k - a_k^{(2)}\right)^{2}\right] &\leq \frac{c \sqrt{\pi} (k+1)^{-1/2}}{2^{k-1}k!} \times \nonumber\\
	&\times \left[4(1-\lambda)^{2t} + \frac{\lambda}{2-\lambda} \left[1-(1-\lambda)^{2(s+t)} \right] + (1-\lambda)^{2(s+t)} \right].\nonumber\\
	\label{MSENon1}
\end{align}

\end{proof}

\begin{theorem} \label{theoremEWGHMISE2}

Given $s+1$ independent observations from  $f_{1} \in L_{2}$ followed by $t$ independent observations from $f_{2} \in L_{2}$, if $r \geq 1$ derivatives of $f_{2}(x)$ exist and $(x-\frac{d}{dx})^r f_{2}(x) \in L_{2}$ and $E|X|^{\frac{2}{3}}<\infty$ for both distributions then the MISE for the EWGH expansion behaves as follows for $N$ sufficiently large:

\begin{align}
	\mbox{MISE}\,(\hat{f}_{N})\, &= O(N^{1/2}) \left[4(1-\lambda)^{2t} + \frac{\lambda}{2-\lambda} \left[1-(1-\lambda)^{2(s+t)} \right] + (1-\lambda)^{2(s+t)} \right]\nonumber\\
	&+ O(N^{-r}). \nonumber
\end{align}

\end{theorem}

\begin{proof}

If $r$ derivatives of $f(x)$ exist and $(x-\frac{d}{dx})^r f(x) \in L_{2}$ then\\ $\sum_{k=N+1}^{\infty} \frac{2^{k-1}k!}{\sqrt{\pi}} a_{k}^{2} = O(N^{-r})$ (see \cite{convergence1},\cite{convergence2}). Combining this fact with the bound on the MSE in proposition \ref{MSECoeffBoundNonIndep} and the expression for the MISE (\ref{MISE}), we obtain:

\begin{align}
\mbox{MISE}\,(\hat{f}_{N})\, &=O(N^{1/2}) \left[4(1-\lambda)^{2t} + \frac{\lambda}{2-\lambda} \left[1-(1-\lambda)^{2(s+t)} \right] + (1-\lambda)^{2(s+t)} \right]\nonumber\\
&+ O(N^{-r}).\nonumber
\end{align}

\end{proof}

\bibliography{articleRefs}
\bibliographystyle{imsart-number}

\end{document}